%% file: oopsla.tex
\pdfoutput=1
\PassOptionsToPackage{scale=0.95}{zi4} 
\iffalse 
\documentclass[acmsmall,screen]{acmart}
\setcopyright{rightsretained}
\else
\documentclass[acmsmall,screen]{acmart}
\setcopyright{cc}
\fi
\acmPrice{}
\acmDOI{10.1145/3586050}
\acmYear{2023}
\copyrightyear{2023}
\acmSubmissionID{oopslaa23main-p131-p}
\acmJournal{PACMPL}
\acmVolume{7}
\acmNumber{OOPSLA1}
\acmArticle{98}
\acmMonth{4}
\received{2022-10-28}
\received[accepted]{2023-02-25}

\input{defs}

\title{Exact Recursive Probabilistic Programming}

\author{David Chiang}
\orcid{0000-0002-0435-4864}
\email{dchiang@nd.edu}

\author{Colin McDonald}
\orcid{0009-0000-8581-0237}
\email{cmcdona8@nd.edu}

\affiliation{
  \department{Department of Computer Science and Engineering}
  \institution{University of Notre Dame}
  \city{Notre Dame}
  \state{IN}
  \country{USA}
}

\author{Chung-chieh Shan}
\orcid{0000-0002-0339-6405}
\email{ccshan@indiana.edu}
\affiliation{
  \department{Department of Computer Science}
  \institution{Indiana University}
  \city{Bloomington}
  \state{IN}
  \country{USA}
}

\ccsdesc[500]{Software and its engineering~Formal language definitions}
\ccsdesc[300]{Mathematics of computing~Probability and statistics}
\ccsdesc[300]{Software and its engineering~Recursion}

\keywords{probabilistic programming, recursive types, linear types}

\begin{document}

\begin{abstract}
Recursive calls over recursive data are useful for generating probability distributions, and probabilistic programming allows computations over these distributions to be expressed in a modular and intuitive way. Exact inference is also useful, but unfortunately, existing probabilistic programming languages do not perform exact inference on recursive calls over recursive data, forcing programmers to code many applications manually. We introduce a probabilistic language in which a wide variety of recursion can be expressed naturally, and inference carried out exactly. For instance, probabilistic pushdown automata and their generalizations are easy to express, and polynomial-time parsing algorithms for them are derived automatically. We eliminate recursive data types using program transformations related to defunctionalization and refunctionalization. These transformations are assured correct by a linear type system, and a successful choice of transformations, if there is one, is guaranteed to be found by a greedy algorithm.
\end{abstract}

\titlenote{This material is based upon work supported by the National Science Foundation under Award Nos.~CCF-2019266 and CCF-2019291\@.}

\maketitle

\section{Introduction}

\begin{figure}
  \resizebox{\textwidth}{!}{%
    \hyphenpenalty=1000
    \begin{tikzpicture}[
        node distance=1.8cm,
        every node/.style={draw,rectangle,align=center,text width=2.23cm,minimum height=1.9cm,inner xsep=0.09cm},
        every edge/.style={draw,-latex,transform canvas={yshift=-6pt},auto=left,every node/.append style={draw=none,text width=1.7cm,minimum height=0cm,font={\small}}}
      ]
    \node(n1) { Program with recursive types and first-class functions };
    \node(n2)[right=of n1] { Program with finite types and first-class functions };
    \node(n3)[right=of n2] { Monotone system of polynomial \smash{equations} };
    \node(n4)[right=of n3] { Exact distribution over values };
  
    \draw (n1) edge node{De- and refunctionalization} node[auto=right]{\cref{sec:recursive}} (n2);
    \draw (n2) edge node{Denotational semantics} node[auto=right]{\cref{sec:denotational}} (n3);
    \draw (n3) edge node{Solve for least fixed \smash{point}} node[auto=right]{\cref{sec:mspe}} (n4);
    \end{tikzpicture}%
  }
\caption{How a PERPL program is compiled and evaluated.
    By ``recursive type'', we mean types such as \lstinline{String} that
    are defined recursively and inhabited by an infinite number of values.
    By ``finite type'', we mean types such as \lstinline{Bool} that
    are inhabited by a finite number of values.}
\label{fig:big_picture}
\end{figure}
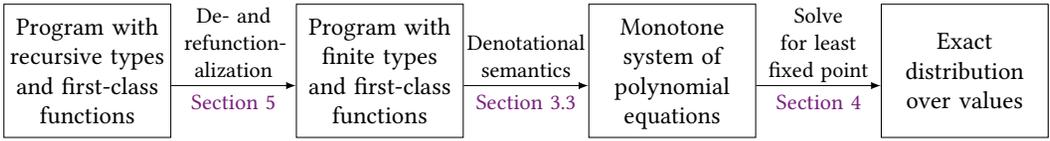

Probability models in natural language processing (NLP\@), as well as computational biology and other fields, often involve recursive computations on recursive data structures.
Sequences of words (or other basic units) have been modeled, for example, using higher-order Markov chains \citep{chen+goodman:1999}, hidden Markov models \citep{rabiner-speech}, finite automata and transducers \citep{mohri-1997-finite}, and chain-structured conditional random fields (CRFs) \citep{lafferty+al:2001}.
Neural models that use CRFs \citep{huang+al:2015} and connectionist temporal classification \citep{graves+al:2006}, which are just finite transducers, are at or near the state of the art in a number of NLP tasks.
Trees, used for representing syntactic structures and sometimes other linguistic structures, have been modeled using probabilistic context-free grammars (PCFGs) \citep{booth+thompson:1973}, tree automata and transducers, or tree-structured CRFs, again performing at or near the state of the art in constituency parsing \citep{stern-etal-2017-minimal}.

In order to support applications such as optimization of parameters by gradient-based methods, it~is important that algorithms for these models can efficiently produce exact results.
Typically they use dynamic programming, but more generally, they can be thought of as solving systems of linear or polynomial equations. Traditionally, these algorithms are coded by hand, which is tedious and error-prone. Worse, small changes to the model can often require nontrivial changes to the algorithm.
For this reason, various toolkits have been developed, for example, for finite-state machines \citep{riley-etal-2009-openfst} and tree transducers \citep{may+knight:2006}. These provide reliable implementations of algorithms, but at the expense of flexibility: small changes to the model can often take it outside the class of model for which the toolkit was designed.

Probabilistic programming languages (PPLs) promise to avoid both manual coding and inflexible toolkits.
The promise begins with an intuitive way to express a distribution, namely by describing a generative process:
a run of a probabilistic program appears to generate a random outcome by making
choices and scoring branches, but those probabilistic side
effects are just notation for a distribution whose
associated quantities---such as total weight or expected value---need to be
computed.
A~distribution expressed in a PPL can be changed easily by editing the program, without needing to manually redesign inference algorithms.

In PPLs, recursive models such as those in NLP are naturally represented as recursive calls and data.
For example, to find the probability of a string under a PCFG\@, it would be natural to write code that recursively generates random trees and sums the probabilities of trees yielding that string.
Indeed, PCFGs served as the central example in the first PPL paper \citep{koller-effective}.
However, the example ultimately did not work out, and efficient inference on these representations remains a longstanding challenge.
The general inference methods offered by PPLs today, often based on sampling, are not suitable in many applications because they take too long, don't compute gradients, or exhibit too much variance.
Some PPLs do support exact inference---beginning with the ``stochastic Lisp'' of \citet{koller-effective}, and more recently with languages like IBAL \citep{pfeffer:2005}, Fun \citep{borgstrom+:2011}, FSPN \citep{stuhlmuller-dynamic}, \textsc{BernoulliProb} \citep{claret+:2013}, PSI \citep{gehr+:2016,gehr+:2020}, Hakaru \citep{walia-efficient}, Dice \citep{holtzen2020dice}, SPPL \citep{saad+:2021}, and ProbZelus \citep{atkinson-semi-symbolic}---but they impose severe limitations on recursion.

In this paper, we introduce PERPL (Probabilistic Exact Recursive Programming Language).
The~denotation of a PERPL program is given by a system of polynomial equations, which can be solved efficiently using numerical methods.
Unlike many general-purpose PPLs, PERPL performs \emph{exact} inference, in the sense that even if it solves the equations iteratively (by Newton's method), any desired level of accuracy is guaranteed by a known number of iterations \citep{stewart+:2015}.
Moreover, it~can differentiate the equations and solve for the derivatives, permitting gradient-based optimization.
And unlike other PPLs that do support exact inference, PERPL can express first-class, higher-order functions, unbounded recursive calls, and some recursive data structures.
For example, a PCFG parser written in PERPL appears to generate (infinitely many) trees and sum the probabilities of those (exponentially many) trees that yield a given string; yet it compiles to a cubic-sized system of equations whose solution is equivalent to the CYK algorithm. We also show programs using stacks or stacks-of-stacks for which PERPL automatically derives polynomial-time algorithms.
We demonstrate that these asymptotic speedups result in orders-of-magnitude differences in a benchmark.

The key to our approach is a linear type system \citep{girard-linear,walker:2005}, which enables first-class functions to be used efficiently and recursive types to be eliminated correctly using defunctionalization \citep{reynolds:1972,danvy-defunctionalization} and refunctionalization \citep{danvy-refunctionalization}, even in the presence of probability effects.
\Cref{fig:big_picture} depicts our pipeline of steps that turns
a naturally expressed probabilistic program into an exact distribution over return values.

Our contributions are
\begin{enumerate}
    \item a nonstandard denotational semantics for a PPL\@, in which a $\lambda$-abstraction from type~$\tau_1$ to type~$\tau_2$ denotes not a distribution over a set of functions, which has cardinality $|\tau_2|^{|\tau_1|}$, but a distribution over a set of pairs, which has cardinality $|\tau_1| \cdot |\tau_2|$ (\cref{sec:denotational});
    \item using defunctionalization and refunctionalization to eliminate recursive types (\cref{sec:recursive});
    \item showing that a linear type system ensures that both of the above are correct, in spite of
          probability effects (\cref{thm:preserve,thm:d-correct});
    \item compiling programs with unbounded recursion into systems of equations and solving them; in particular, loops can be evaluated exactly and directly, not iteratively (\cref{sec:solve_mspe});
    \item compiling natural probabilistic models into polynomial-time parsers for context-free grammars (\cref{sec:parsing_examples}), pushdown automata (\cref{eg:pda}), and their generalization to tree-adjoining grammars and embedded pushdown automata (\cref{sec:epda});
    \item translating programs with affinely used functions to programs with only linearly used functions (\cref{sec:affine});
    \item an open-source implementation based on a compositional, semantics-preserving translation from PERPL to \emph{factor graph grammars} \citep{chiang+riley:2020} (\cref{sec:translation}).
\end{enumerate}

\input{motivation}
\input{ppl}

\input{semantics}

\input{recursive}
\input{examples}
\input{benchmark}

\section{Related Work}

\paragraph{Exact inference and recursion}

PERPL appears to be the first PPL that performs exact inference while
allowing both unbounded recursive calls and \emph{unbounded recursive data}.
Here we consider other PPLs that support exact inference.

The ``stochastic Lisp'' of \citet{koller-effective} computes exact distributions by a generalization of variable elimination.
Many PPLs that followed enable exact inference via graph representations.
IBAL \citep{pfeffer:2001} compiles to a factor graph; 
Fun \citep{borgstrom+:2011} compiles to factor graphs with gates \citep{minka+winn:2008}. Dice \citep{holtzen2020dice} and \textsc{BernoulliProb} \citep{claret+:2013} compile to binary\slash algebraic decision diagrams \citep{bryant:1986,bahar+:1997}, which are also used by the probabilistic model checker PRISM \citep{kwiatkowska-prism}. SPPL \citep{saad+:2021} compiles to sum-product networks \citep{poon+domingos:2011}, and FSPN \citep{stuhlmuller-dynamic} generalizes sum-product networks to represent recursive dependencies.
PSI \citep{gehr+:2016,gehr+:2020} and ProbZelus \citep{atkinson-semi-symbolic} perform exact inference symbolically.
Finally, the expectations---and more generally \emph{moments}---of distributions can be computed automatically by representing them with polynomials, interval arithmetic, and linear recurrences \citep{sankaranarayanan-reasoning,bartocci-automatic,bouissou-uncertainty,moosbrugger-moment}.

Most of these PPLs, to our knowledge, do not support unbounded recursive calls or loops, in the sense of allowing programs where for any $N$, there is a branch of computation with recursion\slash iteration depth greater than $N$ and nonzero weight.
There are four exceptions.
\textsc{BernoulliProb} handles loops using fixed-point iteration (whereas PERPL always compiles loops to linear equations then solves them directly).
FSPN handles recursive calls as in PERPL\@, by solving a system of equations.
ProbZelus handles a \emph{stream} of observations by maintaining a symbolic state, which can represent and update a distribution exactly.
The methods for computing moments just mentioned are focused on unbounded loops, but do not handle other recursive calls such as those in a PCFG\@.

Most of these PPLs, to our knowledge, do not allow unbounded recursive data such as strings, trees, or stacks.
Stochastic Lisp does, evaluating them lazily so they can be of unbounded or infinite size without necessarily causing nontermination.
However, \citet{pfeffer:2005} later noted that the interaction between lazy evaluation and memoization
was problematic.
This line of research continued in IBAL with a new evaluation strategy for lazy memoization, but later languages in this line, such as Figaro \citep{pfeffer:2016}, supported exact inference only for programs with bounded recursive calls and data.
Our use of defunctionalization can also be seen as lazy evaluation and covers similar cases to those that stochastic Lisp and IBAL intended to. But because it requires delayed computations to depend only on finite types, memoization is straightforward.
Also, lazy evaluation
would not turn programs using stacks into polynomial-time algorithms; refunctionalization does.

\paragraph{Probabilistic function semantics}

The semantics of first-class probabilistic functions has
long been known to require more than endowing a domain of functions with
measurable-space structure~\citep{aumann-borel}.
Compared to recent semantics of probabilistic functions and recursion
using quasi-Borel spaces~\citep{heunen-convenient} and logical
relations~\citep{wand-contextual}, our elementary treatment is limited
to models where inference without sampling is possible,
by solving MSPEs.

\paragraph{Linear logic}

We build on linear logic \citep{girard-linear,walker:2005} in several ways.
First, the contrast between ``eager'' evaluation of~$\otimes$ and ``lazy'' evaluation of~$\with$ draws on the typical computational interpretation of linear logic~\citep{abramsky-computational}.
Second, polarity \citep{girard-new}, which we use to allow local nonlinear bindings (\cref{sec:robust}), has been used to structure the type system of a PPL \citep{ehrhard-probabilistic}.
Finally, our treatment of linearly used functions as nondeterministic pairs, which dates back to encoding $\lambda$-calculus variables as Prolog variables \citep{pereira-prolog}, is reminiscent of the symmetric monoidal closed category $\mathcal{R}^\Pi$ of \citet{laird-weighted}, whose tensor product~$\otimes$ as well as exponential~$\llp$ are just the Cartesian product of sets.
Whereas the PPLs in this literature \citep{laird-weighted,laird-higher-order,ehrhard-full,ehrhard-probabilistic}
allow functions to be used any number of times and interpret them by infinitary means,
we use linearity to \emph{limit} expressivity to finitely many weight variables, reducing inference to solving an MSPE.

\paragraph{De- and refunctionalization}

De- and refunctionalization \citep{reynolds:1972,danvy-defunctionalization,danvy-refunctionalization,rendel-automatic} are not new, but using defunctionalization to eliminate non-functions appears novel.
Defunctionalization has been proven correct before \citep{nielsen-denotational,banerjee-design,pottier-polymorphic}, but not for a language with effects other than nontermination.
Refunctionalization is defined as the (left) inverse of defunctionalization, so the correctness of one follows from that of the other. However, our $\mathcal{R}$ transform incorporates additional steps of \emph{disentangling} and \emph{merging apply functions} presented by~\citet{danvy-refunctionalization}.

\section{Conclusion}

PERPL allows programmers to write probabilistic code using unbounded recursive calls and data and still maintain exact inference. We have seen that the compilation to MSPEs can automatically derive inference algorithms that originally took significant intellectual effort to discover.

We have implemented a PERPL compiler in Haskell,
including nonlinear bindings for positive types (\cref{sec:robust}), affine bindings for all types (\cref{sec:affine}), and defunctionalization and refunctionalization (\cref{sec:eliminate_mu,sec:infer_dr_sequence}).
Rather than compile directly to an MSPE\@, it~compiles to an FGG (\cref{sec:translation}). Our implementation of FGGs, which employs the inference methods of \cref{sec:solve_mspe} and makes them automatically differentiable using PyTorch \citep{pytorch}, will be the subject of a future paper. Both implementations are released as open-source software.\footnote{\begin{tabular}[t]{@{}ll}PERPL: \url{https://github.com/diprism/perpl}\quad FGGs: \url{https://github.com/diprism/fggs}\end{tabular}}

\bibliography{references}

\makeatletter
\par\bigskip\noindent{\small\normalfont\@received}\par
\def\@received{}
\makeatother

\clearpage
\EveryShipout{\addtocounter{TotPages}{-1}} 

\appendix

\input{appendix}

\section{Relaxing Linearity}
\label{sec:relaxing-linearity}

Three extensions to our linear type system make it easier to write practical programs without fundamentally affecting the expressivity of the language or the efficiency of the implementation.
These extensions are
\begin{enumerate}
    \item allowing local nonlinear bindings of \emph{positive} type;
    \item allowing linear variables to be discarded unused (in other words, used affinely); and
    \item allowing local nonlinear recursive bindings that do not use linear variables.
\end{enumerate}
The example programs above already use the first two extensions.
We describe each of these extensions in turn.

\input{positive}
\input{affine}

\input{letrec}

\section{Translating to Factor Graph Grammars}
\label{sec:fggs}

\emph{Factor graph grammars} (FGGs) \citep{chiang+riley:2020} are hyperedge replacement grammars (HRGs) \citep{bauderon+courcelle:1987,habel+kreowski:1987,drewes+:1997} for generating sets of factor graphs. They have recently been proposed as a unified formalism that can describe both repeated and alternative substructures \citep{chiang+riley:2020}.
Moreover, inference can be performed on an FGG without enumerating all the factor graphs generated.

\input{fgg}
\input{translation}

\input{epda}
\input{pcfg-gen}

\end{document}

%% file: defs.tex
\usepackage{amsmath}
\allowdisplaybreaks

\newenvironment{shortcases}
    {\left\lbrace\begin{array}{@{}l@{\quad}l@{}}}
    {\end{array}\right.}

\ifdefined\proof\else
    \usepackage{amsthm}
    \newtheorem{theorem}{Theorem}
    \newtheorem{lemma}[theorem]{Lemma}
    \newtheorem{proposition}[theorem]{Proposition}
    
    \theoremstyle{definition}
    \newtheorem{definition}[theorem]{Definition}
    \newtheorem{example}[theorem]{Example}
\fi

\ifdefined\citep\else
    \usepackage{natbib}
\fi

\ifdefined\mathbb\else
    \usepackage{amsfonts}
\fi

\ifdefined\llbracket\else
    \usepackage{stmaryrd}
\fi

\usepackage{mathtools}

\usepackage{calc}
\usepackage{tikz}
\usetikzlibrary{shapes,positioning,calc,fit,backgrounds,snakes,matrix}
\tikzset{node distance=1cm and 0.25cm}
\tikzset{dice/.style={black, dash pattern=on 2pt off 1pt, every mark/.style={solid}, mark=square, mark size=1.5pt}}
\tikzset{perpl/.style={blue!40!red, mark=o, mark size=1.5pt}}
\tikzset{var/.style={draw,circle,fill=white,inner sep=1.5pt,minimum size=8pt}}
\tikzset{ext/.style={var,fill=black,text=white}}
\tikzset{fac/.style={draw,fill=white,rectangle}}
\tikzset{tent/.style={font={\scriptsize},auto}}
\tikzset{subgraph/.style={draw,cloud}}
\tikzset{plate/.style={draw,rectangle,rounded corners,prefix after command={\pgfextra{\tikzset{every label/.style={font={\footnotesize}}}}}}}
\tikzset{every picture/.style={baseline=-2.5pt}}
\tikzset{every label/.style={font={\footnotesize}}}
\newcommand\extinline[1]{\tikz[baseline]\node[ext,inner sep=1pt,anchor=base]{#1};}

\usepackage{tikz-qtree}
\usepackage{pgfplots}
\usepackage{csvsimple}
\usepackage{wrapfig}

\newcommand{\codefont}{\tt}
\newcommand{\keywordfont}{\tt\bfseries}

\newcommand\codelambda{$\lambda$}

\usepackage{listings}
\makeatletter
\lst@AddToHook{OnEmptyLine}{\vspace{\dimexpr-\baselineskip+\smallskipamount}}
\makeatother

\usepackage{comment}
\usepackage{xcolor}
\usepackage{verbatim}
\usepackage{array}

\newcommand{\asst}{\xi}

\newcommand{\type}{\mathit{type}}
\newcommand{\att}{\mathit{att}}
\newcommand{\ext}{\mathit{ext}}
\newcommand{\fac}{\mathit{fac}}
\newcommand{\dom}{\mathit{dom}}
\newcommand{\vlab}{\mathit{vlab}}
\newcommand{\elab}{\mathit{elab}}

\newcommand{\alt}{\mathrel{\big|}}
\newcommand{\casealt}{\mathbin{\codefont{|}}}
\newcommand{\casearrow}{\rightarrow}

\newcommand{\bnf}{\mathrel{::=}}
\newcommand{\bnfalt}{\mathrel{\makebox[\widthof{$\bnf$}][r]{$\alt$}}}

\newcommand{\longalt}{\mathrel{\makebox[\widthof{$\longrightarrow$}]{$\Bigm|$}}}

\newcommand{\nt}[1]{\mathsf{#1}} 

\newcommand{\kw}[1]{\mathrel{\text{\keywordfont #1}}}
\newcommand{\kwe}[1]{\text{\keywordfont #1}}
\newcommand{\id}[1]{\text{\codefont #1}}

\newcommand{\seq}{\mathrel;}

\newcommand{\fv}[1]{\ensuremath{\text{FV}(#1)}}
\newcommand{\domain}[1]{\ensuremath{\mathop{\text{dom}} #1}}

\newcommand{\val}{\mathsf{v}}

\newcommand{\llp}{\multimap}
\usepackage{cmll}

\newcommand{\withtermleft}{\langle}
\newcommand{\withtermright}{\rangle}
\newcommand{\withterm}[1]{\withtermleft #1 \withtermright}
\newcommand{\withdenotation}[2]{{#1}:{#2}}

\newcommand{\timesterm}[1]{(#1)}
\newcommand{\unittype}{\kwe{Unit}}
\newcommand{\unitterm}{\mathinner{()}}
\newcommand{\unitdenotation}{\mathinner{()}}

\newcommand{\ctor}[1]{\mathop{\kern0pt #1}} 
\newcommand{\ctorset}{\mathcal{C}}
\newcommand{\plusdenotation}[2]{{\ctor{#1}{#2}}}

\newcommand{\arrowdenotation}[2]{{#1}\mapsto{#2}}

\usepackage{mathpartir}

\newcommand{\dirac}[1]{\bigl\{\bigl(1,#1\bigr)\bigr\}}
\newcommand{\denote}[1]{\llbracket #1 \rrbracket}

\newcommand{\env}[1]{\tilde{#1}} 

\usepackage[capitalise,nameinlink,noabbrev]{cleveref}

\makeatletter
\DeclareRobustCommand{\crefnosort}[1]{\begingroup\@cref@sortfalse\cref{#1}\endgroup} 
\makeatother

\newif\iftwocolumn
\twocolumnfalse

\newcommand{\eqby}[1]{\stackrel{\smash[t]{\text{\makebox[1.5em][c]{\ifx\relax#1\relax\else\labelcref{#1}\fi}}}}{=}}

%% file: motivation.tex
\section{Motivation}
\label{sec:motivation}

In this section, we present a sequence of examples with increasing demands on expressivity. First, a~single coin flip (\cref{ex:coin}); second, unbounded loops and recursive calls (\cref{ex:loop,ex:inconsistent}); third, recursive data types (\cref{ex:parity-producer-consumer}) and their elimination (\cref{ex:parity-producer-consumer_defunc,ex:parity-tree-refun}). The last of these will also demonstrate the need for (linearly used) first-class functions.

\begin{example}[Probability] \label{ex:coin}
  PERPL\@, like other PPLs, has a mechanism for nondeterministically sampling from probability distributions. To sample from a Bernoulli distribution, we can write:
\begin{lstlisting}
define flip : Bool =
  amb (factor $p$ in true) (factor $q$ in false)
flip  
\end{lstlisting}
where $p$ and $q$ are (metavariables for literal) nonnegative real weights. The distribution \lstinline{flip} returns \lstinline{true} with weight $p$ and \lstinline{false} with weight $q$. To get a Bernoulli distribution, we would require $p+q=1$, but in general, weights do not have to sum to one. Then the final expression \lstinline{flip} calls the distribution, so the value of the whole program is \lstinline{true} with weight $p$ and \lstinline{false} with weight $q$.
\end{example}

\begin{example}[Loops] \label{ex:loop}
To simulate a fair coin using an unfair coin, flip the unfair coin twice. If the two flips disagree, take the first flip; otherwise, repeat \citep{von-neumann-various}.
We can express this unbounded process as:
\begin{lstlisting}
define fair : Bool =
  let x = flip in let y = flip in if x = y then fair else x
fair
\end{lstlisting}
PERPL turns this program into the linear equations
$t = pq + (p^2 + q^2) t$ and
$f = qp + (p^2 + q^2) f$,
then solves for $t$ and $f$. If $p+q=1$, then the answer is $t=f=\tfrac12$.
\end{example}

\begin{example}[Recursive calls] \label{ex:inconsistent}
The PCFG \citep{booth+thompson:1973}
\begin{align}
\label{e:inconsistent-cfg}
    S &\xrightarrow{p} SS &S &\xrightarrow{q} a
\end{align}
defines a distribution over trees, under which the total probability of all \emph{finite} trees is the least nonnegative solution of
\begin{equation}
    \label{e:inconsistent-equation}
    z = p z^2 + q \text.
\end{equation}
If $p+q=1$, then $z = \min\bigl(1, \frac{1-p}{p}\bigr)$.
This probability is computed by the following PERPL program.
\begin{lstlisting}[mathescape=true]
define gen : Unit =
  if flip then let () = gen in let () = gen in ()    -- $S \rightarrow S S$
          else ()                                    -- $S \rightarrow a$
gen
\end{lstlisting}
Although this program always returns \lstinline{()}, the weight associated with \lstinline{()} is the desired probability~$z$.

Notably, this program makes recursive calls of unbounded depth. The denotational semantics of this program will turn out to be the equation~\eqref{e:inconsistent-equation}, which can be solved using standard methods. (In~this case, the quadratic formula gives a closed-form solution; in general, we might need to resort to iterative methods like Newton's method, which is guaranteed to converge to the correct answer.)
\end{example}

\begin{example}[Recursive data] \label{ex:parity-producer-consumer}
In CFG parsing, we want to know the total weight of all derivations of a CFG that yield a given string. 
The following program defines a type \lstinline{String} for strings over the alphabet $\{a\}$, a function \lstinline{gen} that generates strings from the CFG \eqref{e:inconsistent-cfg}, and a function \lstinline{equal} that tests whether the generated string is the given one. The program returns a distribution over Booleans: the weight of \lstinline{true} is the total weight of all CFG derivations that yield the string $aaa$, and the weight of \lstinline{false} is the total weight of all other (finite) CFG derivations.
\begin{lstlisting}
data Nonterminal = S
data Terminal = A
data String = Nil | Cons Terminal String

define gen (lhs: Nonterminal) (acc: String) : String =
  case lhs of S => if flip then gen S (gen S acc) else Cons A acc

define equal (xs: String) (ys: String) : Bool =
  case xs of
    Nil => case ys of Nil => true | Cons y _ => false
    Cons x xs => case ys of Nil => false | Cons y ys => x = y and equal xs ys

equal (gen S Nil) (Cons A (Cons A (Cons A Nil)))
\end{lstlisting}

Processing pipelines like this are common, and writing them in a PPL enables intuitive expression and modular reuse. 
In particular, Bayesian inference is easy to express.
For instance, to predict the next word conditioned on a given prefix, we can leave the \lstinline{gen} function as~is,
change \lstinline{equal} to
\begin{lstlisting}
define next_word (xs: String) (ys: String) : Terminal =
  case xs of
    Nil => fail
    Cons x xs => case ys of Nil => x
                            Cons y ys => if x = y then next_word xs ys else fail
\end{lstlisting}
where \lstinline{fail} denotes the zero distribution, and change the
last line of code to
\begin{lstlisting}
next_word (gen S Nil) (Cons A (Cons A (Cons A Nil)))
\end{lstlisting}
\end{example}
Because the \lstinline{String} type is recursive and infinite, it~is not trivial for PERPL to convert a program like \cref{ex:parity-producer-consumer} to a finite system of equations.
The rest of this section shows how PERPL manages to compile this program to efficient exact inference.
One way to proceed is to apply a whole-program transformation that replaces the recursive
type by a nonrecursive data structure that represents all the places in the code
where it is constructed. This is closely related to \emph{defunctionalization}
\citep{reynolds:1972,danvy-defunctionalization}.

\begin{example}[Defunctionalization]
\label{ex:parity-producer-consumer_defunc}
In \cref{ex:parity-producer-consumer}, there were two uses of \lstinline{String}: for the strings generated by \lstinline{gen}, and for the input string $aaa$ to be parsed. These can be automatically distinguished. Below, we rename them to \lstinline{GenString} and \lstinline{InputString}, respectively. Defunctionalization changes the latter into a nonrecursive type.
\begin{lstlisting}
data GenString = GenNil | GenCons Terminal GenString
data InputString = InputNil | InputCons Terminal Position
data Position = 0 | 1 | 2 | 3

define gen (lhs: Nonterminal) (acc: GenString) : GenString =
  case lhs of S => if flip then gen S (gen S acc) else GenCons A acc

define shift (ys: Position) : InputString =
  case ys of 0 => InputCons A 1 | 1 => InputCons A 2
             2 => InputCons A 3 | 3 => InputNil

define equal (xs: GenString) (ys: Position) : Bool =
  case xs of
    GenNil => case shift ys of InputNil => true | InputCons y _ => false
    GenCons x xs => case shift ys of InputNil => false
                                     InputCons y ys => x = y and equal xs ys

equal (gen S GenNil) 0
\end{lstlisting}
We introduced a new type \lstinline{Position} with four values corresponding to the four places in the original program where an \lstinline{InputString} was constructed, all in the last line. A \lstinline{Position} can be thought of as a potential \lstinline{InputString}, or a tail of the input string. It~is converted into an actual \lstinline{InputString} by the function \lstinline{shift}.

Due to this transformation, neither \lstinline{Position} nor \lstinline{InputString} is recursive, although \lstinline{GenString} still is.
If we read these programs as
call-by-value random generators, then the transformation has interleaved
the producer and consumer of the input string and deforested it away.
\end{example}

Our example still uses the recursive type \lstinline{GenString}. We can get rid of it using
\emph{refunctionalization} \citep{danvy-defunctionalization,danvy-refunctionalization}, a~whole-program
transformation that replaces the recursive type by
the computation that consumes it.

\begin{example}[Refunctionalization] \label{ex:parity-tree-refun}
\Cref{ex:parity-producer-consumer_defunc} only consumes a \lstinline{GenString} in one place: the
\lstinline{case xs} expression in \lstinline{equal}. Refunctionalization
introduces functions \lstinline{gen_nil} and \lstinline{gen_cons} whose bodies are the branches of this \lstinline{case} expression, and replaces the constructors of \lstinline{GenString} with these functions.
\begin{lstlisting}
define gen_nil : Position -> Bool =
  \\ys: Position. case shift ys of InputNil => true | InputCons y _ => false
define gen_cons (x: Terminal) (xs: Position -> Bool) : Position -> Bool =
  \\ys: Position. case shift ys of InputNil => false
                                  InputCons y ys => x = y and xs ys

define gen (lhs: Nonterminal) (acc: Position -> Bool) : Position -> Bool =
  case lhs of S => if flip then gen S (gen S acc) else gen_cons A acc

(gen S gen_nil) 0
\end{lstlisting}
All recursive types are gone now. However, refunctionalization has changed values formerly of type \lstinline{GenString} to type \lstinline{Position -> Bool}, necessitating the use of first-class functions.

The denotational semantics of this program will be a system of equations, and solving these equations is equivalent to the CYK algorithm for CFG parsing. The parser can be generalized to an arbitrary CFG\@, whether in Chomsky normal form, by extending \lstinline{Nonterminal}, \lstinline{Terminal}, and \lstinline{gen}.
\end{example}

Not all programs of interest are structured as producer--consumer pipelines. In \cref{sec:pda_examples}, we will give examples of programs that use stacks, whose production (by pushes) and consumption (by pops) follow no fixed order.

%% file: ppl.tex
\section{A Probabilistic Programming Language}
\label{sec:ppl}

We now define PERPL more formally.
It has two main distinctive features:
\begin{itemize}
\item probability effects, which allow a program to nondeterministically pursue multiple branches, each possibly with a different weight or probability, and
\item a linear type system, in which values of certain types must be consumed exactly once.
\end{itemize}

\subsection{Syntax and Typing}
\label{sec:typing}

The syntax rules are shown in \cref{fig:syntax}.

\subsubsection{Probability}

Evaluating an expression may involve making a probabilistic choice.
The expression $\kw{amb} e_1~e_2$ chooses between $e_1$ and~$e_2$, so, for example,
$({\kw{amb} \kwe{true}~\kwe{false}}, {\kw{amb} \kwe{true}~\kwe{false}})$
splits the current branch of computation into four branches, in which the pair evaluates to $(\kwe{true}, \kwe{true})$, $(\kwe{true}, \kwe{false})$, $(\kwe{false}, \kwe{true})$, and $(\kwe{false}, \kwe{false})$, each with weight~1.
The expression $\kw{factor} w \kw{in} e$, where $w$ is a nonnegative number, multiplies the weight of the current branch by~$w$ and evaluates~$e$.
The expression $\kw{fail}$ has weight~0, so it terminates the current branch.

\input{syntax}

\subsubsection{Linearity}

We use a linear type system for two reasons, both illustrated in \cref{sec:motivation}.
First, defunctionalization, by changing some computations into data structures, delays them.
For example, the calls to \lstinline{Cons} in \cref{ex:parity-producer-consumer} are delayed to inside \lstinline{shift} in \cref{ex:parity-producer-consumer_defunc}. The fact that \lstinline{ys} is consumed only once ensures that any probabilistic effects in the input string would not be duplicated by the reordering.

Second, refunctionalization wraps computations inside $\lambda$-abstractions (as seen in \cref{ex:parity-tree-refun}), and it~would be prohibitively expensive to let a $\lambda$-abstraction from type~$\tau_1$ to type~$\tau_2$ denote a distribution over $|\tau_2|^{|\tau_1|}$ functions.
But if it is used only once
(as \lstinline{xs} and \lstinline{acc} are in \cref{ex:parity-tree-refun}),
it~can denote a distribution over input--output pairs. Since there are only $|\tau_1|\cdot|\tau_2|$ such pairs, this is far more manageable.

For both of these reasons (detailed below in \crefnosort{sec:correctness_mu_stub,sec:denotational}, respectively), we decree that functions must be used \emph{linearly} \citep{girard-linear,walker:2005}, in that, once introduced, they must be used exactly once.
For example, the following program is not well-typed, because it calls \lstinline{f} twice.
\begin{lstlisting}[belowskip=0pt]
let f = \\x. amb x (not x) in (f true, f true)
\end{lstlisting}

\subsubsection{Algebraic Data Types}
\label{sec:adt}

Because of linearity, there are two kinds of tuples \citep{abramsky-computational}.
In~a \emph{multiplicative} tuple $\timesterm{e_1, \ldots, e_n}$,
of~type $\tau_1\otimes\cdots\otimes\tau_n$, all components are computed regardless of
demand, and they are all consumed together.
In the common case where $n=0$, we write $\unitterm$ for the empty tuple and $\unittype$ for its type.
In~an \emph{additive} tuple $\withterm{e_1,\ldots,e_n}$, of~type $\tau_1\with\cdots\with\tau_n$,
just one component is computed depending on which one the context demands,
and only the demanded one is consumed.
Additive tuples also allow $n=0$, but we don't need to notate that case in this paper.

Disjoint union types $\tau \bnf \ctor{c}\tau \oplus \cdots \oplus \ctor{c}\tau$ and terms $e \bnf \ctor{c} e$ are tagged with \emph{constructors}~$c$ drawn from an infinite set~$\ctorset$.
In our example programs, \kwe{data} type declarations can be understood as introducing aliases for unions of multiplicative tuples.
In \cref{ex:parity-producer-consumer_defunc} for instance, \lstinline{Position} is an alias for $\ctor{\id{0}} \unittype \oplus \ctor{\id{1}} \unittype \oplus \ctor{\id{2}} \unittype \oplus \ctor{\id{3}} \unittype$, inhabited by terms like $\ctor{\id{2}} \unitterm$.
In turn, \lstinline{InputString} is an alias for $\ctor{\id{InputNil}} \unittype \oplus \ctor{\id{InputCons}} \mathinner{(\id{Terminal} \otimes \id{Position})}$, inhabited by terms like $\ctor{\id{InputCons}} \mathinner{(\ctor{\id{A}} \unitterm, \ctor{\id{2}} \unitterm)}$.
Here the constructors are $\id{0}$, $\id{1}$, $\id{2}$, $\id{3}$, $\id{InputNil}$, $\id{InputCons}$, and~$\id{A}$.

\subsubsection{Global Definitions}

The \lstinline{define} keyword introduces a global definition, which can recursively use any global variable.
Because global variables are nonlinear, they allow us to express recursion, as \cref{ex:loop,ex:inconsistent,ex:parity-producer-consumer,ex:parity-producer-consumer_defunc,ex:parity-tree-refun} demonstrated.
Global variables are `call-by-name' in the sense that they are bound to computations rather than values: each use of a global variable evaluates to a fresh copy of its definition's right-hand side.
So the following program is well-typed, because each use of \lstinline{f} creates a new $\lambda$-expression that is used once:
\begin{lstlisting}
define f = \\x. amb x (not x); (f true, f true)
\end{lstlisting}
But the following is not well-typed, because \lstinline{g} is used twice.
\begin{lstlisting}
define f = \\x. amb x (not x); let g = f in (g true, g true)
\end{lstlisting}
And in the following program, \lstinline{b} samples from $\kw{amb} \kwe{true}~\kwe{false}$ each time it is evaluated (in other words, its nondeterminism is ``hot''), so the program has four branches, not two:
\begin{lstlisting}[belowskip=0pt]
define b = amb true false; (b, b)
\end{lstlisting}

\subsubsection{Typing}

\input{typing}

The typing rules are shown in \cref{fig:typing}.
Typing judgements for expressions $e$ are of the form $\Gamma; \Delta \vdash e : \tau$, where $\Gamma$ and $\Delta$ are typing contexts for nonlinear and linear bindings, respectively.
Nonlinear (or \emph{intuitionistic}) bindings can be used any number of times, whereas linear bindings must each be used exactly once \citep{barber-dual}.
In these typing rules, global definitions always populate the nonlinear context, whereas local bindings only enter the linear context.
It~is useful in practical applications (including the examples in \cref{sec:motivation}) to relax the linearity requirement to allow linear bindings to be used affinely (zero times or once),
or to allow local nonlinear bindings under certain circumstances. Please see \cref{sec:relaxing-linearity} for more details.

An $\kw{amb}$ expression can be used in any linear context, and both branches are type-checked in that linear context. Since $\kwe{fail}$ is like $\kw{amb}$ but with 0 branches, it can also be used in any linear context.

The difference between multiplicative and additive tuples is reflected in their distinct typing rules.
Whereas the components of an additive tuple $\withterm{e_1,\ldots,e_n}$ are all type-checked in the same linear context~$\Delta$, the components of a multiplicative tuple $\timesterm{e_1, \ldots, e_n}$ are type-checked by partitioning the linear context $\Delta_1,\dotsc,\Delta_n$.
Thus, the empty tuple~$\unitterm$ requires the empty linear context.

For a program to type check under the global context~$\Gamma$, the $\kw{define}$s of the program must provide global variables with types exactly as promised in~$\Gamma$. To~enforce this consistency, the judgement for a program~$p$ takes the form $\Gamma \vdash p : \Gamma' ; \tau$ to track the assumed context~$\Gamma$ and the provided context~$\Gamma'$ separately, and the judgement for a complete program~$p:\tau$ requires $\Gamma$ and~$\Gamma'$ to be identical.

%% file: syntax.tex
\begin{figure*}
\(\begin{array}{@{}lr@{}>{{}}l@{}}
  \text{Programs} & p &\bnf e \alt {\kw{define} x = e} \seq p \\[3pt]
  \text{Expressions} & e &\bnf x 
  \alt \lambda x:\tau.\,e
  \alt e~e
  \alt {\kw{amb} e~e} \alt \kwe{fail} \alt {\kw{factor} w \kw{in} e} 
  \alt (e, \ldots, e) \alt \withterm{e, \ldots, e} \\[3pt]
  &&\bnfalt {\kw{let} (x, \ldots, x) = e \kw{in} e}
  \alt e.i
  \alt \ctor{c} e
  \alt {\kw{case} e \kw{of} {\ctor{c} x} \casearrow e \casealt \cdots \casealt {\ctor{c} x} \casearrow e} \\[3pt]
  \text{Types} & \tau &\bnf \tau \llp \tau \alt \unittype \alt \tau \otimes \cdots \otimes \tau \alt \tau \with \cdots \with \tau \alt \ctor{c}\tau \oplus \cdots \oplus \ctor{c}\tau \\[3pt]
  \text{Constructors} & \ctor{c} &\in \ctorset \qquad \text{$\ctorset$ is an infinite set of constructors} \\[3pt]
  \text{Weights} & w &\in \mathbb{Q}_{\geq 0} \\[3pt]
  \multicolumn{3}{@{}l@{}}{\makebox[\textwidth][s]{Syntactic sugar\hss
\begin{array}[t]{@{}c@{}}
  {\kw{let} x_1 = e_1 \kw{in} e'} \equiv (\lambda x_1. e')~e_1 \\[3pt]
  \kwe{Bool} \equiv {\ctor{\id{True}} \unittype} \oplus {\ctor{\id{False}} \unittype} \qquad \kwe{true} \equiv {\ctor{\id{True}} \unitterm} \qquad \kwe{false} \equiv {\ctor{\id{False}} \unitterm} \\[3pt]
  {\kw{if} e \kw{then} e'_1 \kw{else} e'_2} \equiv {\kw{case} e \kw{of} {\ctor{\id{True}} u} \casearrow {\kw{let} \unitterm = u \kw{in} e'_1} \casealt {\ctor{\id{False}} u} \casearrow {\kw{let} \unitterm = u \kw{in} e'_2}} \\[3pt]
  e_1 \kw{and} e_2 \equiv {\kw{if} e_1 \kw{then} e_2 \kw{else} \kwe{false}} \qquad
  \kw{not} e \equiv {\kw{if} e \kw{then} \kwe{false} \kw{else} \kwe{true}}
\end{array}}}
\end{array}\)
\caption{Syntax. We use a taller vertical bar ($\alt$) to separate different right-hand sides of a BNF production, and a shorter vertical bar ($\casealt$) to separate the branches of a $\kw{case}$ expression.}
\label{fig:syntax}
\end{figure*}

%% file: typing.tex
\begin{figure*}
\vspace{-3ex}
$\renewcommand{\arraystretch}{3}\begin{array}{@{}l@{}c@{}}
  \text{Variables} & \inferrule*{ }{\Gamma, x:\tau; \cdot \vdash x:\tau} \qquad
                     \inferrule*{ }{\Gamma; x:\tau \vdash x:\tau} \\[-1.5ex]
\text{Functions} & \inferrule{\Gamma; \Delta_0, x_1: \tau_1 \vdash e':\tau'}{\Gamma; \Delta_0 \vdash \lambda x_1 . e' : \tau_1 \llp \tau'} \qquad
  \inferrule{\Gamma; \Delta_0 \vdash e_0 : \tau_1 \llp \tau' \\ \Gamma; \Delta_1 \vdash e_1 : \tau_1}{\Gamma; \Delta_0, \Delta_1 \vdash e_0~e_1 : \tau'} \\
  \text{Distributions} & \inferrule{\Gamma; \Delta \vdash e_1 : \tau \\ \Gamma; \Delta \vdash e_2 : \tau}{\Gamma; \Delta \vdash {\kw{amb} e_1~e_2}: \tau} \qquad
  \inferrule{ }{\Gamma; \Delta \vdash \kwe{fail} : \tau} \qquad
  \inferrule{\Gamma; \Delta \vdash e : \tau }{\Gamma; \Delta \vdash {\kw{factor} w \kw{in} e} : \tau} \\
  \text{Tuples} & \inferrule{\Gamma; \Delta_i \vdash e_i : \tau_i}{\Gamma; \Delta_1, \dotsc, \Delta_n \vdash (e_1, \dotsc, e_n) : \tau_1 \otimes \dotsb \otimes \tau_n} \qquad
  \inferrule{\Gamma; \Delta \vdash e_i : \tau_i}{\Gamma; \Delta \vdash \withterm{e_1, \dotsc, e_n} : \tau_1 \with \dotsb \with \tau_n} \\
  & \hspace{-1em} \inferrule{\Gamma; \Delta \vdash e : \tau_1 \otimes \dotsb \otimes \tau_n \\ \Gamma; \Delta', x_1: \tau_1, \dotsc, x_n: \tau_n \vdash e' : \tau'}{\Gamma; \Delta, \Delta' \vdash {\kw{let} (x_1, \dotsc, x_n) = e \kw{in} e' : \tau'}} \qquad
  \inferrule{\Gamma; \Delta \vdash e : \tau_1 \with \dotsb \with \tau_n}{\Gamma; \Delta \vdash e.i : \tau_i} \\
  \text{Unions} & \hspace{-1.5em} \inferrule{\Gamma; \Delta \vdash e_i : \tau_i}{\Gamma; \Delta \vdash {\ctor{c_i} e_i}: \ctor{c_1}\tau_1 \oplus \dotsb \oplus \ctor{c_n}\tau_n} \quad
  \inferrule{\Gamma; \Delta \vdash e : \ctor{c_1}\tau_1 \oplus \dotsb \oplus \ctor{c_n}\tau_n \\ \Gamma; \Delta', x_i: \tau_i \vdash e'_i: \tau'}{\Gamma; \Delta, \Delta' \vdash {\kw{case} e \kw{of} {\ctor{c_1} x_1} \casearrow e'_1 \casealt \dotsb \casealt {\ctor{c_n} x_n} \casearrow e'_n} : \tau'}
\\
  \text{Programs} &
  \inferrule{\Gamma; \cdot \vdash e : \tau \\ \Gamma \vdash p : \Gamma'; \tau'}{\Gamma \vdash (\kw{define} x = e \seq p) : (x: \tau, \Gamma'); \tau'} \qquad
  \inferrule{\Gamma; \cdot \vdash e:\tau}{\Gamma \vdash e:\cdot;\tau}
  \qquad \qquad \inferrule{\Gamma \vdash p : \Gamma ; \tau}{p : \tau}
\end{array}$
\caption{Typing rules. We write $\cdot$ for an empty typing context.}
\label{fig:typing}
\end{figure*}

%% file: semantics.tex
Next, we give a standard operational semantics and a slightly nonstandard denotational semantics for PERPL\@, then prove that they agree thanks to the linear use of functions and additive tuples.

\subsection{Operational Semantics}

Our operational semantics is defined by a one-step reduction judgement $\gamma \vdash e \Longrightarrow E'$, shown in Figure~\ref{fig:operational}.
The expression being reduced is~$e$.
The reduction result $E'$ is a distribution over expressions, written as a set of weight-expression pairs; this is used to reduce $\kw{amb}$, $\kw{fail}$, and $\kw{factor}$.
The global environment $\gamma = (x_1=e_1,\dotsc)$ maps global names to their definitions; this is used to reduce global names in~\eqref{eq:reduction_global}.
Our evaluation contexts allow evaluation under~$\lambda$ and on both sides of application, so the order of probabilistic choices is left unspecified; this is benign, as denotations are preserved regardless of the order (\cref{thm:preserve}).

If $E$ is a distribution over expressions, then we can reduce $E$ by reducing any element of~$E$:
\begin{definition}
\label{def:red_distr}
If $E = \{(w,e)\} + E_0$ and $\gamma \vdash e \Longrightarrow E'$, then $\gamma \vdash E \Longrightarrow w \cdot E' + E_0$ (where $\cdot$ and $+$ denote scaling and addition of distributions).
We also write $\Longrightarrow^*$ for a sequence of zero or more reductions.
\end{definition}

For example, starting with the expression $\kw{factor} w_1 \kw{in} {\kw{factor} w_2 \kw{in} e}$,
we can apply \cref{def:red_distr} to \eqref{eq:reduction_factor} twice to get
\begin{equation}
    \gamma \enspace\vdash\enspace
    \bigl\{\bigl(1, \kw{factor} w_1 \kw{in} {\kw{factor} w_2 \kw{in} e} \bigr)\bigr\}
    \enspace\Longrightarrow\enspace
    \bigl\{\bigl(w_1, \kw{factor} w_2 \kw{in} e \bigr)\bigr\}
    \enspace\Longrightarrow\enspace
    \bigl\{\bigl(w_1 w_2, e \bigr)\bigr\}
    \text.
\end{equation}

\input{operational}


We conclude this section with a proof sketch of type soundness.

\begin{lemma}[Linear substitution preserves typing]
    \label{thm:subst_preserves_typing_linear}%
          Suppose\/ \(\Gamma; \Delta_0, x_1\colon\tau_1 \vdash e' : \tau'\)
          and\/ \(\Gamma; \Delta_1 \vdash e_1 : \tau_1\).
          Then \(\Gamma; \Delta_0,\Delta_1 \vdash e'\{x_1:=e_1\} : \tau'\).
\end{lemma}
\begin{proof}
  By induction on the typing derivation of $e'$.
\end{proof}

\begin{definition}
We say that a global environment $\gamma = (x_1=e_1,\dotsc)$ is \emph{well-typed} for a context $\Gamma = (x_1:\tau_1,\dotsc)$ iff $\Gamma; \cdot \vdash e_i \colon \tau_i$ for all~$i$.
\end{definition}

\begin{proposition}[Reduction preserves typing]\label{thm:reduction_preserves_typing}%
  Let\/ $\gamma$ be well-typed for~$\Gamma$.
  If\/ \(\Gamma; \Delta \vdash e \colon \tau\) and $\gamma \vdash e \Longrightarrow E'$, then $\Gamma; \Delta \vdash e' \colon \tau$ for all $(w, e') \in E'$.
\end{proposition}
\begin{proof}
  By induction on the derivation of $\gamma \vdash e \Longrightarrow E'$.
  Case \eqref{eq:reduction_global} uses the well-typedness of~$\gamma$.
  Cases \eqref{eq:reduction_function}, \eqref{eq:reduction_product}, \eqref{eq:reduction_sum} use \cref{thm:subst_preserves_typing_linear}.
  Case \eqref{eq:reduction_congruence} uses the induction hypothesis.
\end{proof}

\begin{proposition}[Progress]
  Let\/ $\gamma$ be well-typed for~$\Gamma$.
  If\/ \(\Gamma; \cdot \vdash e \colon \tau\) then either $e$ is a value or $\gamma \vdash e \Longrightarrow E'$ for some $E'$.
\end{proposition}
\begin{proof}
  By induction on the typing derivation of~$e$.
  The benign nondeterminism mentioned above is not used.
\end{proof}

\subsection{Denotational Semantics}
\label{sec:denotational}

In this section, we define the denotation of a complete program $p:\tau$, which is a distribution over the set denoted by the type~$\tau$.
First we define a \emph{distribution} over a set~$X$ to be a mapping $X \rightarrow [0, \infty]$.
This makes sense in measure theory whenever $X$ is countable; for us, $X$ is always finite.
Note that weights can be irrational or infinite (unlike in the syntax and operational semantics). If $\chi_1, \chi_2$ are distributions over~$X$, we define the (complete) partial order $\chi_1 \leq \chi_2$ iff for all $x \in X$, $\chi_1(x) \leq \chi_2(x)$.

We handle recursive calls in our denotational semantics in a standard way.
Without the global environment (our sole source of recursive calls), we would just let each expression~$e$ denote
a mapping from environments to distributions over semantic values,
where an environment maps each free variable of~$e$ to its semantic value.
With the global environment, we have to first define this denotation relative to an interpretation of each global name as a distribution.
Thus, the denotation of a program's global definitions maps an interpretation of each global name to another interpretation of each global name.
Finally, we take the least fixed point of this monotonic map~\citep{kozen-semantics}.

\begin{definition}
\label{def:denote_type}%
A type \(\tau\) denotes a finite set $\denote{\tau}$ of \emph{semantic values}, defined by induction on~$\tau$:
\begin{align*}
  \denote{\tau_1 \otimes \dots \otimes \tau_n} &= \denote{\tau_1} \times \dots \times \denote{\tau_n} \\
  \denote{\ctor{c_1}\tau_1 \oplus \dots \oplus \ctor{c_n}\tau_n} &= \{\, \plusdenotation{c_i}{v} \mid 1 \le i \le n,\ v \in \denote{\tau_i} \,\} \\
  \denote{\tau_1 \llp \tau_2} &= \{\, \arrowdenotation{v_1}{v_2} \mid v_1 \in \denote{\tau_1},\ v_2 \in \denote{\tau_2} \,\} \\
  \denote{\tau_1 \with \dots \with \tau_n} &= \{\, \withdenotation{i}{v} \mid 1 \le i \le n,\ v \in \denote{\tau_i} \,\}
\end{align*}%
The first two cases are unsurprising:
the tuple type $\tau_1 \otimes \dots \otimes \tau_n$ denotes a set of tuples, and
the union type $\ctor{c_1}\tau_1 \oplus \dots \oplus \ctor{c_n}\tau_n$ denotes a disjoint union, whose elements $\plusdenotation{c_i}{v}$ are just pairs of $\ctor{c_i}$ and~$v$.

The last two cases are where our denotational semantics is nonstandard.
One might expect $\denote{\tau_1\llp\tau_2}$ to be the set of functions $\denote{\tau_2}^{\denote{\tau_1}}$.
But the denotation of an expression of type $\tau_1\llp\tau_2$ will involve a distribution over $\denote{\tau_1\llp\tau_2}$, and that would have $|\denote{\tau_2}|^{|\denote{\tau_1}|}$ weights and lead to programs that require exponential time and space.
Linearity permits a more efficient way.
Intuitively, when a function is created, it~nondeterministically guesses what (one) argument value it will receive. Later, applying the function (once) is just a matter of unifying the actual argument with this guess.
The semantic value $\arrowdenotation{v_1}{v_2}$ is just a suggestively notated pair of $v_1$ and~$v_2$.
Hence, the cardinality $\bigl|\denote{\tau_1 \llp \tau_2}\bigr|$ is only $\bigl|\denote{\tau_1}\bigr|\cdot\bigl|\denote{\tau_2}\bigr|$, enabling PERPL to handle many programs in polynomial time.

Similarly, when an additive tuple is created, it~nondeterministically guesses which (one) component will be demanded.
The semantic value $\withdenotation{i}{v}$ is just a pair of $i$ and~$v$, and the cardinality $\bigl|\denote{\tau_1 \with \dots \with \tau_n}\bigr|$ is not $\bigl|\denote{\tau_1}\bigr| \cdots \bigl|\denote{\tau_n}\bigr|$ but only $\bigl|\denote{\tau_1}\bigr| + \dots + \bigl|\denote{\tau_n}\bigr|$.
\end{definition}

\begin{example}
Since $\kw{Unit}$ is the type of 0-tuples, the set $\denote{\kw{Unit}}$ is $\{\unitdenotation\}$.
Since $\kw{Bool}$ is syntactic sugar for ${\ctor{\id{True}}\unittype} \oplus {\ctor{\id{False}}\unittype}$, the set $\denote{\kw{Bool}}$ is $\{\plusdenotation{\ctor{\id{True}}}{\unitdenotation},\plusdenotation{\ctor{\id{False}}}{\unitdenotation}\}$.

Recall from \cref{sec:adt} that
$\id{Position} =
 \let\sep\relax
 \def\do#1{\sep
           \ctor{\id{#1}} \unittype
           \def\sep{\oplus}}
 \do{0} \do{1} \do{2} \do{3}
$, so
$\denote{\id{Position}} = \{
 \let\sep\relax
 \def\do#1{\sep
           \plusdenotation{\ctor{\id{#1}}}{\unitdenotation}
           \def\sep{,}}
 \do{0} \do{1} \do{2} \do{3}
\}$.
Then $\denote{\id{Position} \llp \id{Position}}$ has $4 \cdot 4 = 16$ members, not $4^4 = 256$:
\begin{equation*}
\denote{\id{Position} \llp \id{Position}} = \bigl\{
\def\sep{}
\def\DO#1{\sep
          &\let\sep\relax
           \def\do##1{\sep
                      \bigl( \plusdenotation{\ctor{\id{#1}}}{\unitdenotation}, \plusdenotation{\ctor{\id{##1}}}{\unitdenotation} \bigr)
                      \def\sep{,}}
           \do{0} \do{1} \do{2} \do{3}
          \def\sep{,\\}}
\begin{aligned}[t]
   \DO{0} \DO{1} \DO{2} \DO{3}
\smash[t]{\bigr\}\text.}
\end{aligned}
\end{equation*}
Furthermore, $\denote{(\id{Position} \llp \id{Position}) \llp \id{Position}}$ has $4 \cdot 4 \cdot 4 = 64$ members, not $4^{4^4} \approx 10^{154}$.
\end{example}

A~global typing context $\Gamma$ denotes the set of all mappings from variables in $\Gamma$ to distributions over semantic values. That is, $\denote{\Gamma}$ contains all possible $\eta$ such that for each $x:\tau \in \Gamma$, we have a distribution $\eta(x) \colon \denote{\tau} \rightarrow [0,\infty]$.
Thus, $\denote{\Gamma}$ is generally uncountable.

A~local typing context $\Delta$ denotes the set of all mappings from variables in $\Delta$ to semantic values. That is, $\denote{\Delta}$ contains all possible $\delta$ such that if $x:\tau \in \Delta$, then $\delta(x) \in \denote{\tau}$.
Thus, $\denote{\Delta}$ is a Cartesian product of finite sets and itself finite.

\begin{example}
In \cref{ex:inconsistent}, the global typing context is $\Gamma=(\id{flip}:\kwe{Bool}, \id{gen}:\unittype)$. Since $\kwe{Bool}$ has two values and $\unittype$ has one, each $\eta\in\denote{\Gamma}$ stores 3 weights:
\begin{align*}
\eta(\id{flip})\bigl(\plusdenotation{\ctor{\id{True}}}{\unitdenotation}\bigr) &\in [0,\infty] &
\eta(\id{flip})\bigl(\plusdenotation{\ctor{\id{False}}}{\unitdenotation}\bigr) &\in [0,\infty] &
\eta(\id{gen})\bigl(\unitdenotation\bigr) &\in [0,\infty]
\end{align*}
There are no local variables in \cref{ex:inconsistent}, so the local typing context $\Delta$ is everywhere empty, and the unique $\delta\in\denote{\Delta}$ is also empty.
\end{example}

A~typing judgement \(\Gamma; \Delta \vdash e:\tau\) denotes a mapping that takes an $\eta \in \denote{\Gamma}$ to a distribution
over \(\denote{\Delta} \times \denote{\tau}\).
We write this distribution as $\denote{\Gamma; \Delta \vdash e:\tau}_\eta$, but
as $\denote{e}_\eta$
or even $\denote{e}$ for short, and we define it compositionally on the typing derivation, using the equations in \cref{fig:denotational}.
The Iverson bracket $\mathbb{I}[\mathord\cdot]$ is 1 if its contents are true, 0 otherwise.

\input{denotational}

\begin{example} \label{ex:flip_denotation}
Consider the expression $\kw{amb} (\kw{factor} p \kw{in} \kwe{true})~(\kw{factor} q \kw{in} \kwe{false})$ from \cref{ex:coin}. We build up the denotation of this expression as follows.
\begin{gather}
  \denote{\unitterm}(\emptyset,\unitdenotation) \eqby{eq:denote_mulprod} 1 \notag \\
\begin{flalign*}
  \denote{\kwe{true}}(\emptyset,\plusdenotation{c}{\unitdenotation}) &\eqby{eq:denote_inj} \mathbb{I}[c={\ctor{\id{True}}}] &
  \denote{\kwe{false}}(\emptyset,\plusdenotation{c}{\unitdenotation}) &\eqby{eq:denote_inj} \mathbb{I}[c={\ctor{\id{False}}}] \\
  \denote{\kw{factor} p \kw{in} \kwe{true}}(\emptyset,\plusdenotation{c}{\unitdenotation}) &\eqby{eq:denote_factor} p \cdot \mathbb{I}[c={\ctor{\id{True}}}] &
  \denote{\kw{factor} q \kw{in} \kwe{false}}(\emptyset,\plusdenotation{c}{\unitdenotation}) &\eqby{eq:denote_factor} q \cdot \mathbb{I}[c={\ctor{\id{False}}}]
\end{flalign*} \\
  \denote{\kw{amb} (\kw{factor} p \kw{in} \kwe{true})~(\kw{factor} q \kw{in} \kwe{false})}\bigl(\emptyset,\plusdenotation{c}{\unitdenotation}\bigr) \eqby{eq:denote_amb} \begin{shortcases} p &c={\ctor{\id{True}}} \\ q &c={\ctor{\id{False}}}. \end{shortcases} \label{eq:denote-coin}
\end{gather}
\end{example}
\begin{example} \label{ex:inconsistent_denotation}
Moving on to \cref{ex:inconsistent}, assume that $\eta$ is given.
\begin{align}
    \denote{\id{flip}}_\eta\bigl(\emptyset, \plusdenotation{c}{\unitdenotation}\bigr) &\eqby{eq:denote_var} \eta(\id{flip})\bigl(\plusdenotation{c}{\unitdenotation}\bigr) \notag \\
    \denote{\id{gen}}_\eta\bigl(\emptyset, \unitdenotation\bigr) &\eqby{eq:denote_var} \eta(\id{gen})\bigl(\unitdenotation\bigr) \notag \\
    \denote{\kw{let} \unitterm = \id{gen} \kw{in} \unitterm}_\eta\bigl(\emptyset, \unitdenotation\bigr) &\eqby{eq:denote_mulprodelim} \eta(\id{gen})\bigl(\unitdenotation\bigr) \cdot 1 \notag \\
    \denote{\kw{let} \unitterm = \id{gen} \kw{in} {\kw{let} \unitterm = \id{gen} \kw{in} \unitterm}}_\eta \bigl(\emptyset, \unitdenotation\bigr) &\eqby{eq:denote_mulprodelim} \eta(\id{gen})\bigl(\unitdenotation\bigr) \cdot \eta(\id{gen})\bigl(\unitdenotation\bigr) \cdot 1 \notag \\
    \begin{aligned}[b]
    \denote{&{\kw{if} \id{flip} \kw{then} {\kw{let} \unitterm = \id{gen}}} \\[-\jot] &\quad {\kw{in} {\kw{let} \unitterm = \id{gen} \kw{in} \unitterm} \kw{else} \unitterm}}
    \end{aligned}_\eta \bigl(\emptyset, \unitdenotation\bigr)
    &\eqby{eq:denote_case} \eta(\id{flip})\bigl(\plusdenotation{\ctor{\id{True}}}{\unitdenotation}\bigr) \cdot \eta(\id{gen})\bigl(\unitdenotation\bigr) \cdot \eta(\id{gen})\bigl(\unitdenotation\bigr) \cdot 1 \notag \\*[-\jot]
    &\quad+ \eta(\id{flip})\bigl(\plusdenotation{\ctor{\id{False}}}{\unitdenotation}\bigr) \cdot 1.
\label{eq:denote-gen}
\end{align}
\end{example}

Given a set of global definitions $\gamma = (x_1=e_1:\tau_1,\dotsc,x_n=e_n:\tau_n)$,
we define its denotation~$\denote{\gamma}$ to be the global environment~$\eta$
that is the least solution to the equations
\begin{equation}
    \eta(x_i)(v_i) = \denote{e_i}_\eta(\emptyset,v_i)
\label{eq:fixpoint}
\end{equation}
for each $i=1,\dotsc,n$ and each $v_i$ in~$\denote{\tau_i}$.
Finally, the program
${\kw{define} x_1 = e_1} \seq \dotso \seq {\kw{define} x_n = e_n} \seq e$
denotes the distribution $\denote{e}_{\denote{\gamma}}(\emptyset,v)$
over semantic values~$v$.

\begin{example} \label{ex:inconsistent-equations}
For \cref{ex:coin,ex:inconsistent}, the equations are
\begin{align}
    \eta(\id{flip})\bigl(\plusdenotation{c}{\unitdenotation}\bigr) &\eqby{eq:fixpoint} \denote{\kw{amb} (\kw{factor} p \kw{in} \kwe{true})~(\kw{factor} q \kw{in} \kwe{false})}\bigl(\emptyset,\plusdenotation{c}{\unitdenotation}\bigr) \label{eq:denote-global-flip} \\
    &\eqby{eq:denote-coin} \begin{shortcases} p & c={\ctor{\id{True}}} \\ q & c={\ctor{\id{False}}} \end{shortcases} \label{eq:denote-global-flip-sub} \\
    \eta(\id{gen})\bigl(\unitdenotation\bigr) &\eqby{eq:fixpoint} \denote{\kw{if} \id{flip} \kw{then} {\kw{let} \unitterm = \id{gen} \kw{in} {\kw{let} \unitterm = \id{gen} \kw{in} \unitterm}} \kw{else} \unitterm}\bigl(\emptyset, \unitdenotation\bigr) \label{eq:denote-global-gen} \\
    &\eqby{eq:denote-gen} \eta(\id{flip})\bigl(\plusdenotation{\ctor{\id{True}}}{\unitdenotation}\bigr) \cdot \eta(\id{gen})\bigl(\unitdenotation\bigr) \cdot \eta(\id{gen})\bigl(\unitdenotation\bigr) \cdot 1 + \eta(\id{flip})\bigl(\plusdenotation{\ctor{\id{False}}}{\unitdenotation}\bigr) \cdot 1. \label{eq:denote-global-gen-sub}
\end{align}
In all, there are three equations in three unknowns, $\eta(\id{flip})\bigl(\plusdenotation{\ctor{\id{True}}}{\unitdenotation}\bigr)$, $\eta(\id{flip})\bigl(\plusdenotation{\ctor{\id{False}}}{\unitdenotation}\bigr)$, and $\eta(\id{gen})\bigl(\unitdenotation\bigr)$.
The whole program in \cref{ex:inconsistent}, of type $\unittype$, denotes a single weight, which we write as $z$ for short:
\begin{equation}
z = \denote{\id{gen}}\bigl(\emptyset, \unitdenotation\bigr) \eqby{eq:denote_var} \eta(\id{gen})\bigl(\unitdenotation\bigr).
\label{eq:denote-inconsistent}
\end{equation}
Putting \labelcref{eq:denote-global-flip-sub,eq:denote-global-gen-sub,eq:denote-inconsistent} together, we get \cref{e:inconsistent-equation} as promised:
\begin{align}
  z &= p z^2 + q.
\label{eq:denote-inconsistent-sub}
\end{align}
\end{example}

\subsection{Soundness}

We justify our denotational semantics by showing
that it is consistent with our standard operational semantics,
thanks to the linear type system.

\begin{definition}
\label{def:deterministic}
We say that the denotation \(\denote{\Gamma; \Delta \vdash e:\tau}_\eta\) is \emph{deterministic} if, for each $\delta \in \denote{\Delta}$,
it~assigns weight 1 to just one $v \in \denote{\tau}$ and weight 0 to all other $v \in \denote{\tau}$.
\end{definition}

\begin{lemma}[Linear substitution preserves denotation]
\leavevmode

\postdisplaypenalty=500
          Suppose\/ \(\Gamma; \Delta_0, x_1\colon\tau_1 \vdash e' : \tau'\)
          and\/ \(\Gamma; \Delta_1 \vdash e_1 : \tau_1\)
          (as in \cref{thm:subst_preserves_typing_linear}).
          Then\label{thm:subst_preserves_denote_linear}
  \begin{equation}
    \denote{e'\{x_1:=e_1\}}_\eta(\delta_0 \cup \delta_1, v') 
    = \sum_{v_1\in\denote{\tau_1}} \, \denote{e'}_\eta(\delta_0 \cup \{(x_1, v_1)\}, v') \cdot \denote{e_1}_\eta(\delta_1, v_1)
  \end{equation}
for all $\eta \in \denote{\Gamma}$,
$\delta_0 \in \denote{\Delta_0}$,
$\delta_1 \in \denote{\Delta_1}$, and
$v' \in \denote{\tau'}$.
\end{lemma}
\begin{proof}
By induction on the typing derivation of~$e'$.
See \cref{sec:proof_subst_preserves_denote_linear} for more details.
\end{proof}

\begin{theorem}[Reduction preserves denotation]\label{thm:preserve}%
If\/ $\Gamma; \Delta \vdash e: \tau$ and $\gamma \vdash e \Longrightarrow E'$ in the operational semantics,
then in the denotational semantics, for all $\delta \in \denote{\Delta}$ and $v \in \denote{\tau}$,
\[
    \denote{e}_{\denote{\gamma}}(\delta,v)
    = \sum_{(w,e')\in E'} w \cdot \denote{e'}_{\denote{\gamma}}(\delta,v).
\]
\end{theorem}
\begin{proof}
By induction on the reduction derivation of $e$.
We show the case $\gamma \vdash (\lambda x_1. e')\,e_1 \Longrightarrow \{(1, e'\{x_1:=e_1\})\}$ where $e_1$ is a syntactic value:
\begin{align*}
  \denote{(\lambda x_1. e')\,e_1}(\delta_0 \cup \delta_1, v') &\eqby{eq:denote_app}
  \sum_{v_1} \, \denote{\lambda x_1. e'}(\delta_0, \arrowdenotation{v_1}{v'}) \cdot \denote{e_1}(\delta_1, v_1) \\
  &\eqby{eq:denote_abs} \sum_{v_1} \, \denote{e'}(\delta_0 \cup \{(x_1,v_1)\}, v') \cdot \denote{e_1}(\delta_1, v_1) \\
  &\eqby{} 1 \cdot \denote{e'\{x_1:=e_1\}}(\delta_0 \cup \delta_1, v') \text{ by \cref{thm:subst_preserves_denote_linear}.} \qedhere
\end{align*}
\end{proof}

\section{Inference}
\label{sec:mspe}

The denotation of a program under the above semantics is the least fixed point of a system of equations, which has a particular form that has been well studied in other contexts.

\begin{definition}
  A \emph{monotone system of polynomial equations}, or MSPE \citep{etessami+yannakakis:2009,esparza+:2008}, is a system of equations
  \begin{equation*}
  \begin{aligned}
    z_1 &= P_1(z_1, \dots, z_n) \\[-2ex]
    &\vdotswithin{=} \\[-1ex]
    z_n &= P_n(z_1, \dots, z_n)
  \end{aligned}
  \end{equation*}
  where
  the $z_i$ are called \emph{weight variables} (to distinguish them from variables in PERPL) and
  each $P_i$ is a polynomial with nonnegative real coefficients. 
  We write $\mathbf{z}$ for a vector in $[0,\infty]^n$ of assignments to the weight variables. If $\mathbf{z}, \mathbf{z}'$ are such assignments, we write $\mathbf{z} \leq \mathbf{z}'$ iff for all $i$, $z_i \leq z'_i$.
  We write $\mathbf{P}(\mathbf{z})$ for the vector $[P_i(z_1, \ldots, z_n)]_{i=1}^n \in [0,\infty]^n$. Thus the MSPE can be written succinctly as $\mathbf{z} = \mathbf{P}(\mathbf{z})$.
\end{definition}

\subsection{Constructing an MSPE}

The first step of inference is to instantiate \cref{fig:denotational,eq:fixpoint} for all the expressions of the program, which has the form of an MSPE.

\begin{proposition}
  The denotation of a program is a distribution whose weight values are components of the least solution (in $[0,\infty]^n$) of an MSPE.
\end{proposition}

\begin{proof}
  The MSPE has two kinds of weight variables:
  \begin{enumerate}
  \item For every subexpression $\Gamma; \Delta \vdash e: \tau$, every $\delta \in \denote{\Delta}$, and every $v \in \denote{\tau}$, there is a weight variable $\denote{e}_\eta(\delta, v)$.
        \Cref{fig:denotational} gives an equation whose left-hand side is this weight variable and whose right-hand side is a polynomial in the weight variables with nonnegative coefficients.
        Examples of such equations are \labelcref{eq:denote-coin,eq:denote-gen,eq:denote-inconsistent}.
  \item For every global variable $x : \tau$ and every $v \in \denote{\tau}$, there is a weight variable $\eta(x)(v)$.
        \Cref{eq:fixpoint} gives an equation whose left-hand side is this weight variable and whose right-hand side is a weight variable of the first kind.
        Examples of such equations are \labelcref{eq:denote-global-flip,eq:denote-global-gen}.
  \end{enumerate}
  The denotation of a program
  ${\kw{define} x_1 = e_1} \seq \dotso \seq {\kw{define} x_n = e_n} \seq e$
  with type $\tau$
  is the distribution that assigns to each $v \in \denote{\tau}$ the weight $\denote{e}(\emptyset; v)$, which is a weight variable of the first kind above.

  All that needs to be shown is that the MSPE has finite size.
  Under \cref{def:denote_type}, $\denote{\tau}$~is finite for every~$\tau$, which can be shown by induction on the structure of~$\tau$. Also, $\denote{\Delta}$~is finite, being a finite Cartesian product of finite sets. So the number of weight variables is finite.
\end{proof}

Even though the number of weight variables is finite, it can be exponential in the size of the program, because $\denote{\Delta}$ is the product of as many sets as $\Delta$ has variables.
If~only for this reason, PERPL inference is intractable in general.
This blowup is not surprising given that, even without any recursive calls or data, PERPL can easily express all discrete Bayes nets \citep{cooper-computational} and conjunctive queries \citep{chandra-optimal}, and inherits their intractability.

Even for a family of programs---such as the typical parser---for which the number of weight variables is polynomial in the size of the program, an unwisely chosen dependency can dramatically increase the degree of the polynomial; this is a concern as well for the polynomial-time dynamic-programming inference algorithms that we are trying to recover automatically.
As with Bayes nets and conjunctive queries, finding the optimal inference strategy (tree decomposition) is NP-hard, but many heuristics help in practice.
Also, in practice the weight variables can be packed into a tensor to take advantage of vectorized parallelism.

\subsection{Solving the MSPE}
\label{sec:solve_mspe}

The general strategy for solving an MSPE automatically is to decompose it into smaller MSPEs.
If~$z_1, z_2$ are weight variables, we write $z_1 \prec z_2$ if there is an equation whose left-hand side is $z_2$ and whose right-hand side contains $z_1$.
Then find the strongly connected components (SCCs) of~$\prec$.
In~\cref{ex:inconsistent-equations},
\begin{align*}
\eta(\id{flip})\bigl(\plusdenotation{\ctor{\id{True}}}{\unitdenotation}\bigr) &\prec \eta(\id{gen})\bigl(\unitdenotation\bigr) &
\eta(\id{flip})\bigl(\plusdenotation{\ctor{\id{False}}}{\unitdenotation}\bigr) &\prec \eta(\id{gen})\bigl(\unitdenotation\bigr) &
\eta(\id{gen})\bigl(\unitdenotation\bigr) &\prec \eta(\id{gen})\bigl(\unitdenotation\bigr)
\end{align*}
so each of the three variables forms its own SCC\@.
We visit each SCC in topological order, solving it and substituting its solution into the remaining equations.

The easiest, and most common, case is when an SCC has just one weight variable, whose equation (after substituting earlier weight variables) must be of the form $z = w$ where $w$ is a constant.
Many existing PPLs with exact inference handle only this case, that is, when $\prec$ is acyclic.

The next easiest case is when an SCC has more than one weight variable but (after substituting earlier weight variables) its equations are all linear. These can be solved \emph{directly}---as opposed to iteratively---in the semiring $[0,\infty]$ by one of the algorithms of \citet{lehmann:1977}, such as Gaussian elimination. This case includes all loops (\cref{ex:loop}) and arises often in practice; for example, PCFG rules of the form $(X \rightarrow X)$ can be used an unbounded number of times. Most existing parsers limit how many times such rules can be applied \citep[e.g.,][]{collins:1999,taskar+:2004,finkel+:2008}, but PERPL makes it easy to write code that handles such situations exactly.

The most general case is when the equations for the SCC are nonlinear. An example is parsing using a PCFG with rules of the form $(X \rightarrow \epsilon)$. Again, implementations usually resort to arbitrary limits on how such rules can be applied \citep[e.g.,][]{cai+:2011}.
But these cases can be solved using a generalization of Newton's method to $\omega$-continuous semirings \citep{esparza+:2007,esparza+:2010semiring},
which is guaranteed to converge to the least solution of an MSPE\@:
\begin{align*}
  \mathbf{z}^{(0)} &= \mathbf{0} \\
  \mathbf{z}^{(i+1)} &= \mathbf{z}^{(i)} + \bigl(\tfrac{\partial\mathbf{P}}{\partial\mathbf{z}}\bigl(\mathbf{z}^{(i)}\bigr)\bigr)^* \bigl(\mathbf{P}(\mathbf{z}^{(i)}) - \mathbf{z}^{(i)}\bigr)
\end{align*}
where $\tfrac{\partial\mathbf{P}}{\partial\mathbf{z}}$ is the \emph{formal} derivative of $\mathbf{P}$ (that is, defined only using the sum and product rules, not using limits).
For any matrix $\mathbf{A}$, the \emph{closure} $\mathbf{A}^* = \sum_{i=0}^\infty \mathbf{A}^i$ can be computed by an algorithm of \citet{lehmann:1977}.
And $\infty - \infty$ can be defined to be anything (say, zero).

Although the iterates of Newton's method are in general only approximations, there is a known number of iterations that is guaranteed to reach any desired level of accuracy \citep{stewart+:2015}. This is also true of, for example, logarithms, and in the context of machine learning, such computations are generally considered ``exact,'' in contrast to methods based on random sampling.

\begin{example}
  In \cref{eq:denote-inconsistent-sub},
  let $p = \tfrac23$ and $q = \tfrac13$. Then
  $P(z) = \tfrac23 z^2 + \tfrac13$ and
  $\tfrac{\partial P}{\partial z}(z) = \tfrac43z$,
  and using Newton's method to solve $z = P(z)$ gives
  \begin{align*}
    z^{(0)} &= 0 \\
    z^{(1)} &= 0 + 0^* \left(\tfrac13 - 0\right) = \tfrac13 \approx 0.33333 \\
    z^{(2)} &= \tfrac13 + \left(\tfrac49\right)^* \left(\tfrac{11}{27} - \tfrac13\right) = \tfrac7{15} \approx 0.46667 \\
    z^{(3)} &= \tfrac7{15} + \left(\tfrac{28}{45}\right)^* \left(\tfrac{323}{675} - \tfrac7{15}\right) = \tfrac{127}{255} \approx 0.49804 \\
    z^{(4)} &= \tfrac{127}{255} + \left(\tfrac{508}{765}\right)^* \left(\tfrac{97283}{195075} - \tfrac{127}{255}\right) = \tfrac{32767}{65536} \approx 0.49999
  \end{align*}
  and so on.
  But suppose $p=q=1$. Then
  $P(z) = z^2 + 1 $ and $\tfrac{\partial P}{\partial z}(z) = 2z$,
  and Newton's method gives
  \begin{align*}
    z^{(0)} &= 0 \\
    z^{(1)} &= 0 + 0^* (1 - 0) = 1 \\
    z^{(2)} &= 1 + 2^* (2 - 1) = \infty \\
    z^{(3)} &= \infty + \infty^* (\infty - \infty) = \infty \text.
  \end{align*}
  Accordingly, our implementation actually outputs \texttt{inf}.
\end{example}

In machine learning applications, it is often useful to optimize a quantity involving the distribution computed by a program, by adjusting parameters that the program depends on. In \cref{ex:parity-producer-consumer} for instance, having computed the likelihood $L$ of some observed strings, we might want to maximize~it by adjusting~$p$. To adjust~$p$ by gradient descent, we can differentiate the MSPE denoted by the program with respect to $p$ and then solve the resulting (linear) system of equations for $\frac{\mathrm{d}L}{\mathrm{d}p}$. 

\subsection{Factor Graph Grammars}

Our implementation actually does not compile PERPL programs directly to MSPEs, but to \emph{factor graph grammars} (FGGs) \citep{chiang+riley:2020}, whose inference algorithm constructs an MSPE in turn and solves it while taking advantage of vectorized parallelism. FGGs are a formalism for describing probability models on recursive structures that is more general than factor graphs, case--factor diagrams \citep{mcallester+:2008} and sum--product networks \citep{poon+domingos:2011}. Details of the translation to FGGs are in \cref{sec:fggs}.

%% file: operational.tex
\begin{figure*}
\begin{minipage}{\linewidth}
\begin{subequations}    
\begin{alignat}{2}
    \gamma &\vdash{}&
    x
    & \Longrightarrow \dirac{ e }
    \quad \text{if $(x=e)\in\gamma$}
    \label{eq:reduction_global}
\\
    \gamma &\vdash{}&
    (\lambda x_1.e')~v_1
    & \Longrightarrow \dirac{ e'\{x_1:=v_1\} }
    \label{eq:reduction_function}
\\
    \gamma &\vdash{}&
    {\kw{amb} e_1~e_2}
    & \Longrightarrow \bigl\{\bigl(1,e_1\bigr), \bigl(1,e_2\bigr)\bigr\}
\\
    \gamma &\vdash{}&
    {\kw{fail}}
    & \Longrightarrow \emptyset
\\
    \gamma &\vdash{}&
    {\kw{factor} w \kw{in} e}
    & \Longrightarrow \bigl\{\bigl(w, e\bigr)\bigr\}
    \label{eq:reduction_factor}
\\
    \gamma &\vdash{}&
    {\kw{let} (x_1, \dots, x_n) = (v_1, \dots, v_n) \kw{in} e'}
    & \Longrightarrow \dirac{ e'\{x_1:=v_1, \dots, x_n:=v_n\} }
    \label{eq:reduction_product}
\\
    \gamma &\vdash{}&
    {\withterm{e_1, \dots, e_n}.i}
    & \Longrightarrow \dirac{ e_i }
\\
    \gamma &\vdash{}&
    {\kw{case} {\ctor{c_i} v} \kw{of} {\ctor{c_1} x_1} \casearrow e'_1 \casealt \dots \casealt {\ctor{c_n} x_n} \casearrow e'_n}
    & \Longrightarrow \dirac{ e'_i\{x_i:=v\} }
    \label{eq:reduction_sum}
\end{alignat}
\begin{equation}
    \inferrule{\gamma \vdash e \Longrightarrow E'}
              {\gamma \vdash C[e] \Longrightarrow \left\{(w, C[e']) \mid (w,e') \in E'\right\}} \label{eq:reduction_congruence}
\end{equation}
\end{subequations}

\begin{align*}
&\text{Syntactic values} &  v &\bnf \lambda x. e \alt (v, \dots, v) \alt {\ctor{c} v} \alt \withterm{e, \dots, e} \\
&\text{Evaluation contexts} & C &\bnf [\cdot] \alt \lambda x:\tau.\, C \alt C~e \alt e~C \alt {\kw{factor} w \kw{in} C} \\
&&    &\bnfalt (e, \dots, e, C, e, \dots, e) \alt {\kw{let} (x, \dots, x) = C \kw{in} e}\alt {\kw{let} (x, \dots, x) = e \kw{in} C} \\
&&    &\bnfalt C.i \alt {\ctor{c} C} \alt {\kw{case} C \kw{of} {\ctor{c_1} x} \casearrow e \casealt \cdots \casealt {\ctor{c_n} x} \casearrow e}
\end{align*}
\end{minipage}

\caption{Operational semantics.}
\label{fig:operational}
\end{figure*}

%% file: denotational.tex
\begin{figure*}
\begin{minipage}{\linewidth}
\begin{subequations} \label{eq:denotational}
\begin{align}
  \denote{x}_\eta(\delta, v) &= \begin{shortcases}
    \eta(x)(v) & x\in\domain{\eta} \\
    \mathbb{I}[v = \delta(x)] & x\in\domain{\delta}
  \end{shortcases} \label{eq:denote_var} \\
  \denote{\lambda x_1. e'}(\delta_0, \arrowdenotation{v_1}{v'}) &= \denote{e'}(\delta_0 \cup \{(x_1,v_1)\}, v') \label{eq:denote_abs} \\
  \denote{e_0~e_1}(\delta_0 \cup \delta_1, v') &= \sum_{v_1\in\denote{\tau_1}} \, \denote{e_0}(\delta_0, \arrowdenotation{v_1}{v'}) \cdot \denote{e_1}(\delta_1, v_1) \label{eq:denote_app} \\
  \denote{\kw{amb} e_1~e_2}(\delta, v) &= \denote{e_1}(\delta, v) + \denote{e_2}(\delta, v) \label{eq:denote_amb} \\
  \denote{\kwe{fail}}(\delta, v) &= 0 \label{eq:denote_fail} \\
  \denote{\kw{factor} w \kw{in} e}(\delta, v) &= w \cdot \denote{e}(\delta, v) \label{eq:denote_factor} \\
  \denote{(e_1, \dots, e_n)} (\delta_1 \cup \dots \cup \delta_n, (v_1, \dots, v_n)) &= \denote{e_1}(\delta_1, v_1) \cdots \denote{e_n}(\delta_n, v_n) \label{eq:denote_mulprod} \\
  \denote{\kw{let} (x_1, \dots, x_n) = e \kw{in} e'} (\delta \cup \delta', v') &=
  \smash[b]{\sum_{v_1, \dots, v_n}} \denote{e}(\delta, (v_1, \dots, v_n)) \cdot {} \notag \\
    &\hspace{5em} \denote{e'}(\delta' \cup \{(x_1, v_1), \dots, (x_n, v_n)\}, v') \label{eq:denote_mulprodelim} \\
  \denote{\withterm{e_1, \dots, e_n}} (\delta, \withdenotation{i}{v}) &= \denote{e_i}(\delta, v) \label{eq:denote_addprod} \\
  \denote{e.i} (\delta, v) &= \denote{e}(\delta, \withdenotation{i}{v}) \label{eq:denote_addprodelim} \\
  \denote{\ctor{c} e}(\delta, \plusdenotation{c'}{v}) &= \begin{shortcases}
    \denote{e}(\delta, v) & c'=c \\
    0 & c' \neq c
  \end{shortcases} \label{eq:denote_inj} \\
  \begin{array}[b]{@{}r@{}l@{}}
  \denote{
    \kw{case} e \kw{of}{} & {\ctor{c_1} x_1} \casearrow e'_1 \casealt \dots \\
             {}\casealt{} & {\ctor{c_n} x_n} \casearrow e'_n
  }(\delta \cup \delta', v')
  \end{array} &=
    \sum_{i=1}^n \sum_{v} \, \denote{e}(\delta, \plusdenotation{c_i}{v}) \cdot
    \denote{e'_i}(\delta' \cup \{(x_i, v)\}, v')
    \label{eq:denote_case}
\end{align}
\end{subequations}
\end{minipage}
\caption{Denotational semantics. In $\denote{e}(\delta, v)$, the domain of $\delta$ is always the set of free variables in $e$, so in \labelcref{eq:denote_app}, the left-hand side uniquely determines $\delta_0$ and $\delta_1$ on the right-hand side, and similarly in \labelcref{eq:denote_mulprod,eq:denote_mulprodelim,eq:denote_case}.}
\label{fig:denotational}
\end{figure*}

%% file: recursive.tex
\newcommand\transform[2]{#1\llbracket #2\rrbracket}
\newcommand\defunc[1]{\transform{\mathcal{D}_{#1}}}
\newcommand\refunc[1]{\transform{\mathcal{R}_{#1}}}
\newcommand\transformsto[1]{\xrightarrow{#1}}
\newcommand\defuncsto{\transformsto{\mathcal{D}}}
\newcommand\refuncsto{\transformsto{\mathcal{R}}}

\newcommand\rec{\sigma} 
\newcommand\recbody{\bar\sigma} 
\newcommand\recone{\recbody\{\alpha:=\rec\}} 
\newcommand\nonrec{\varphi} 
\newcommand\nonrecone{\recbody\{\alpha:=\nonrec\}} 
\newcommand\occ{m} 
\newcommand\occtype{\varphi} 
\newcommand\free{\vec{y}} 
\newcommand\freesubi{e} 
\newcommand\freetype{\vec{\varphi}} 
\newcommand\freelength{{|\free_i|}}

\tikzset{
  dr/.style={
    every edge/.style={draw,->,>=latex,auto=left},
    every node/.style={execute at begin node={\strut}},inner ysep=0pt
  }
}

\section{Recursive Types}
\label{sec:recursive}

In this section, we extend the language to allow recursive types (\cref{sec:mu_definition}), and then, because our denotational semantics does not include recursive types and because known exact inference methods do not support recursive types, we show how to eliminate them (\crefrange{sec:eliminate_mu}{sec:infer_dr_sequence}).

\subsection{Definitions}
\label{sec:mu_definition}

As shown in \cref{fig:recursive}, we add iso-recursive types \citep{pierce} with introduction and elimination forms $\kw{fold}$ and $\kw{unfold}$ (also sometimes called $\kw{roll}$ and $\kw{unroll}$; they do not mean catamorphism and anamorphism).
\begin{figure}
\begin{align}
&\text{Syntax} && \makebox[33em][c]{$\displaystyle
  e \bnf {\kw{fold}_{\mu^t\alpha.\tau} e} \alt {\kw{unfold}_{\mu^t\alpha.\tau} x = e \kw{in} e'} \hspace{4em}
  \tau \bnf \mu^t \alpha.\, \tau \alt \alpha
$} \notag\\
&\text{Typing} && \makebox[33em][c]{$\displaystyle
\inferrule{\Gamma; \Delta \vdash e : \tau\{\alpha:=\mu^t\alpha.\tau\}}{\Gamma; \Delta \vdash {\kw{fold}_{\mu^t\alpha.\tau} e} : \mu^t\alpha.\,\tau} \qquad
\inferrule{\Gamma; \Delta \vdash e : \mu^t\alpha.\,\tau \\ \Gamma; \Delta', x: \tau\{\alpha:=\mu^t\alpha.\tau\} \vdash e' : \tau'}{\Gamma; \Delta, \Delta' \vdash {\kw{unfold}_{\mu^t\alpha.\tau} x = e \kw{in} e'} : \tau'}
$} \notag\\
&\text{Reduction} && \makebox[30em][c]{\hspace{3em}$\displaystyle
\label{e:unfold-fold}
\gamma \vdash {\kw{unfold}_{\mu^t\alpha.\tau} x = (\kw{fold}_{\mu^t\alpha.\tau} v) \kw{in} e'} \Longrightarrow \dirac{e'\{x:=v\}}
$} \\
&\makebox[5.5em][l]{Syntactic values} && \makebox[33em][c]{$\displaystyle
    v \bnf {\kw{fold}_{\mu^t\alpha.\tau} v}
$} \notag\\
&\makebox[5.5em][l]{Evaluation contexts} && \makebox[33em][c]{$\displaystyle
    C \bnf {\kw{fold} C} \alt {\kw{unfold} x = C \kw{in} e} \alt {\kw{unfold} x = e \kw{in} C}
$} \notag
\end{align}
\caption{Syntax and semantics of recursive types.}
\label{fig:recursive}
\end{figure}
Our definition of recursive types has two slightly unusual features.
First, $\kw{unfold}$ expressions have a scope~$e'$, which will be used below (\cref{sec:refunctionalization}) when eliminating recursive types. Invariably, $e'$ is a $\kw{case}$ expression, so we often write $\kw{case} {\kw{unfold} e} \kw{of} \ldots$ as shorthand for $\kw{unfold} x = \nolinebreak e \kw{in} {\kw{case} x \kw{of} \ldots\,}$.

Second, the superscript $t$ is a \emph{tag} drawn from any infinite set such as~$\mathbb{N}$.
It~is used to distinguish recursive types that would otherwise be considered equal, which is sometimes necessary for eliminating them (illustrated in \cref{ex:parity-producer-consumer-again} below).
As with nominal typing, we consider two recursive types that differ only in their tag to be two different types.
Additionally, we require that different recursive types must have different tags; that is, if the program contains both $\mu^t\alpha.\tau_1$ and $\mu^t\alpha.\tau_2$, then $\tau_1=\tau_2$.
(This requirement is easy to satisfy, and used in \cref{thm:evolve}.)
Tags can be inferred by using a different tag variable for each occurrence of~$\mu$, and unifying tag variables during type checking as necessary. (Our implementation of PERPL also infers tag polymorphism in global definitions and datatypes, by a straightforward generalization of Hindley--Milner--Damas type inference. It then eliminates this inferred polymorphism by monomorphization.)

\begin{example} \label{ex:odd}
In \cref{ex:odd_defunc,ex:odd_refunc}, we will transform the following simple program to eliminate the recursive type $\text{\lstinline{Nat}} \equiv \mu^t\alpha.\, (\ctor{\id{Zero}} \unittype \oplus \ctor{\id{Succ}} \alpha)$. This program samples a natural number $n$ from an exponential distribution (by flipping a coin repeatedly and counting how many flips are \lstinline{true} before the first \lstinline{false}), and tests whether $n$ is odd.
\begin{lstlisting}[belowskip=0pt]
data Nat = Zero | Succ Nat
define sample : Nat = if flip then fold (Succ sample) else fold Zero
define odd (n: Nat) : Bool =
  unfold n' = n in case n' of Zero => false | Succ m => not (odd m)
odd sample
\end{lstlisting}
\end{example}

\subsection{Eliminating Recursive Types}
\label{sec:eliminate_mu}

Recursive types such as \lstinline{Nat} above
cannot be handled by the conversion to an MSPE (\cref{sec:mspe}) because they would give rise to an infinite number of weight variables. However, they can often be transformed away.
In this section, we define these transformations and show how they work on \cref{ex:odd}.
In \cref{sec:correctness_mu_stub,sec:correctness_mu}, we argue that they preserve program meaning, even in the presence of probabilistic effects.
In \cref{sec:infer_dr_sequence}, we provide a greedy algorithm that finds a successful sequence of transformations whenever one exists.

Our transformations are closely related to defunctionalization and refunctionalization, which are well-known. However, our formulation is slightly nonstandard, and our usage of them in a language with probabilistic effects is novel. Most notably, whereas traditional defunctionalization often introduces a recursive type, our usage of it eliminates a recursive type.

\subsubsection{Defunctionalization}
\label{sec:defunctionalization}

Suppose we want to eliminate a recursive type $\rec = \mu^t \alpha.\, \recbody$ from a given program.
Let the $n$ occurrences of $\kw{fold}_\rec$ in the program be $\kw{fold}_\rec \occ_1, \ldots, \kw{fold}_\rec \occ_n$.
For each~$i$, let the nonglobal free variables of~$\occ_i$ form the tuple $\free_i:\freetype_i$.
We define a transformation $\defunc\rec{\cdot}$ that changes terms of type $\rec$ to type $\nonrec = {\ctor{\id{Fold}_1} \freetype_1} \oplus \cdots \oplus {\ctor{\id{Fold}_n} \freetype_n}$, where $\ctor{\id{Fold}_1},\dotsc,\ctor{\id{Fold}_n}$ are constructors.
Because our purpose is to eliminate~$\rec$, it~doesn't help to replace~$\rec$ by another type that contains~$\rec$, and replacing $\rec$ in $\nonrec$ would cause infinite regress, so we require that $\nonrec$ not contain~$\rec$.

On types, our transformation $\defunc\rec{\cdot}$ just changes occurrences of $\rec$ to~$\nonrec$, so it leaves $\freetype_i$ unchanged.
On programs, it~adds a global function $u_\rec$ ($u$~stands for unfold), of type $\nonrec \llp \nonrecone$:
\begin{equation}
\label{e:defunc_u}
  \kw{define} u_\rec = \lambda x\colon\nonrec.\, {\kw{case} x \kw{of} {\ctor{\id{Fold}_1} \free_1} \casearrow \defunc\rec{\occ_1} \casealt \cdots \casealt {\ctor{\id{Fold}_n} \free_n} \casearrow \defunc\rec{\occ_n}}
\end{equation}
On terms, it~postpones from $\kw{fold}_\rec$ to $\kw{unfold}_\rec$ the work done by $u_\rec$:
\begin{align}
\label{e:defunc_terms}
  \defunc\rec{\kw{fold}_\rec \occ_i} &= {\ctor{\id{Fold}_i} \free_i} &
  \defunc\rec{\kw{unfold}_\rec x = e \kw{in} e'} &= {\kw{let} x = u_\rec~\defunc\rec{e} \kw{in} \defunc\rec{e'}}
\end{align}
The other cases are uninteresting.

If the $\occ_i$ are all abstractions, and if the $e'$ in each $\kw{unfold}_\rec x = e \kw{in} e'$ is an application of~$x$, then the $\mathcal{D}_\rec$ transformation is equivalent to defunctionalization \citep{reynolds:1972,danvy-defunctionalization}, with $u_\rec$ usually called \lstinline{apply}$_\rec$. Although the original purpose of defunctionalization was to get rid of $\lambda$-abstractions by postponing them to their applications, in general it can be used to get rid of any introduction forms by postponing them to their elimination forms; here, we use it to get rid of $\kw{fold}$s by postponing them to their $\kw{unfold}$s.

\begin{example} \label{ex:odd_defunc}
  In \cref{ex:odd}, there are two occurrences of $\kw{fold}$, inside $\id{sample}$. Since neither of them has any free variables, $\nonrec = \ctor{\id{Fold}_1}\unittype \oplus \ctor{\id{Fold}_2}\unittype$ and so we can defunctionalize $\id{Nat}$. Create new types \lstinline{NatFolded} and \lstinline{NatUnfolded}, which take the roles of $\nonrec$ and $\nonrecone$ above, respectively:
\begin{lstlisting}
data NatFolded = in1 | in2
data NatUnfolded = Zero | Succ NatFolded
\end{lstlisting}
The bodies of the $\kw{fold}$ expressions move into a new
function \lstinline{u_Nat}, which plays the role of $u_\rec$:
\begin{lstlisting}
define u_Nat (n: NatFolded) : NatUnfolded =
  case n of in1 => Succ sample | in2 => Zero
\end{lstlisting}
The rest of the program is transformed by \cref{e:defunc_terms}:
\begin{lstlisting}[belowskip=0pt]
define sample : NatFolded = if flip then in1 else in2
define odd (n: NatFolded) : Bool =
  let n' = u_Nat n in case n' of Zero => false | Succ m => not (odd m)
odd sample
\end{lstlisting}
\end{example}

\subsubsection{Refunctionalization}
\label{sec:refunctionalization}

Again suppose we want to eliminate a recursive type $\rec = \mu^t \alpha. \recbody$ from a given program.
Instead of moving the work done at $\kw{fold}$ing, we can move the work done at $\kw{unfold}$ing.
Without loss of generality, assume that every $\kw{unfold}_\rec$ expression binds the same variable name~$x$, and let the $n$ occurrences of $\kw{unfold}_\rec$ in the program be $\kw{unfold}_\rec x = \dotsb \kw{in} \occ_i$, where $1\le i\le n$.
For each~$i$, let the nonglobal free variables of the body $\occ_i:\occtype_i$, except~$x$, form the tuple $\free_i : \freetype_i$.
We~define a transformation $\refunc\rec{\cdot}$ that changes terms of type $\rec$ to type
\begin{equation} \label{e:refunc_type}
    \nonrec = (\freetype_1 \llp \occtype_1) \with \cdots \with (\freetype_n \llp \occtype_n) \text.
\end{equation}
Again because our purpose is to eliminate~$\rec$, we require that $\nonrec$ not contain~$\rec$.

On types, the transformation $\refunc\rec\cdot$ just changes occurrences of $\rec$ to $\nonrec$, so it leaves $\freetype_i$ and $\occtype_i$ unchanged.
On programs, it~adds a global function $f_\rec$ ($f$~stands for fold), of type $\nonrecone \llp \nonrec$:
\begin{equation}
  \kw{define} f_\rec = \lambda x\colon \nonrecone.\,
    \withterm{ \lambda \free_1.\, \refunc\rec{\occ_1}, \ldots,
               \lambda \free_n.\, \refunc\rec{\occ_n} }
  \label{e:refunc_f}
\end{equation}
On terms, it~moves from $\kw{unfold}_\rec$ to $\kw{fold}_\rec$ the work done by $f_\rec$:
\begin{align}
  \refunc\rec{\kw{fold}_\rec e} &= f_\rec~\refunc\rec{e} &
  \refunc\rec{\kw{unfold}_\rec x = e \kw{in} \occ_i} &= (\refunc\rec{e}.i)~\free_i \label{e:refunc_terms}
\end{align}
This is essentially the same as refunctionalization \citep{danvy-defunctionalization,danvy-refunctionalization}; in \citeauthor{danvy-refunctionalization}'s terminology, we have \emph{disentangled} the $\kw{unfold}_\rec$ expressions \pagebreak[1]into the apply functions $\lambda \free_i.\, \refunc\rec{\occ_i}$ and \emph{merged} them using an additive tuple (which works even if they have different types). The original definition of refunctionalization assumed that $\recbody$ was a union type and created a function for each case; we create a single function $f_\rec$ for the same purpose.

\begin{example} \label{ex:odd_refunc}
In \cref{ex:odd}, there is just one $\kw{unfold}$ expression. It has type $\kw{Bool}$, and its body has no free variables other than \lstinline{n'}. So instead of defunctionalizing $\id{Nat}$, we can also refunctionalize~it. Create new types \lstinline{NatFolded} and \lstinline{NatUnfolded}, which take the roles of $\nonrec$ and $\nonrecone$ above, respectively:
\begin{lstlisting}
type NatFolded = Unit -> Bool
data NatUnfolded = Zero | Succ NatFolded
\end{lstlisting}
(According to \cref{e:refunc_type}, \lstinline{NatFolded} should be a unary additive tuple type, which has no notation in this paper. To keep the presentation simple, we just omit unary tuple types here.)

The body of the $\kw{unfold}$ expression moves into a new function \lstinline{f_Nat}, which plays the role of $f_\rec$:
\begin{lstlisting}
define f_Nat (n': NatUnfolded) : NatFolded =
  \\(): Unit. case n' of Zero => false | Succ m => not (odd m)
\end{lstlisting}
The rest of the program is transformed by \cref{e:refunc_terms}:
\begin{lstlisting}[belowskip=0pt]
define sample : NatFolded = if flip then f_Nat (Succ sample) else f_Nat Zero
define odd (n: NatFolded) : Bool = n ()
odd sample
\end{lstlisting}
\end{example}

\subsection{Correctness}
\label{sec:correctness_mu_stub}

It is easy to show that the transforms $\mathcal{D}$ and~$\mathcal{R}$ preserve types; see \cref{sec:defunc_preserves_types,sec:refunc_preserves_types}.
In \cref{sec:correctness_mu}, we~show that they furthermore preserve the overall meaning of a program (a distribution over $\unittype$, without loss of generality).
Intuitively, this is because they merely change where the work done by $u_\rec$ or~$f_\rec$ is expressed:
$\mathcal{D}_\rec$~moves the work done by~$u_\rec$ from $\kw{fold}_\rec$ to $\kw{unfold}_\rec$, whereas $\mathcal{R}_\rec$ moves the work done by~$f_\rec$ from $\kw{unfold}_\rec$ to $\kw{fold}_\rec$.

\subsection{Example: CFG Parsing}
\label{sec:parsing_examples}
\label{ex:parity-producer-consumer-again}

We next show how defunctionalization and refunctionalization together enable PERPL to handle some programs that neither transformation alone can.
Specifically, we eliminate recursive types from the CFG parser in \cref{ex:parity-producer-consumer}.
This larger example demonstrates how the classic CYK algorithm is recovered, as well as some additional aspects of the transformations.

We repeat \cref{ex:parity-producer-consumer} below, with a few more details spelled out.
First, we write \lstinline{fold}s and \lstinline{unfold}s explicitly. 
Second, as mentioned above, tag inference automatically distinguishes two uses of \lstinline{String}, which we call \lstinline{String[1]} and \lstinline{String[2]}. If it were not for tags, we would not be able to handle this program: it would not be possible to defunctionalize \lstinline{String} because of the free variable $\id{acc}$ in \lstinline{fold (Cons A acc)}, and it would not be possible to refunctionalize \lstinline{String} either because of the free variable $\id{xs}$ inside the second $\kw{case} {\kw{unfold} \id{ys}}$ expression.
  
Finally, some variables are used nonlinearly. For the nonrecursive type \lstinline{Terminal}, this is harmless; \cref{sec:robust} discusses how to handle it. But for recursive types, we need to add a \lstinline{discard} function to turn non-use into linear use; \cref{sec:affine} shows how this can be done automatically instead.

\begin{lstlisting}
data String[1] = Nil[1] | Cons[1] Terminal String[1]
data String[2] = Nil[2] | Cons[2] Terminal String[2]

define gen (lhs: Nonterminal) (acc: String[1]) : String[1] =
  case lhs of S => if flip then gen S (gen S acc) else fold (Cons[1] A acc)

define equal (xs: String[1]) (ys: String[2]) : Bool =
  case unfold xs of
    Nil[1] => case unfold ys of Nil[2] => true | Cons[2] y ys => discard[2] ys
    Cons[1] x xs => case unfold ys of Nil[2] => discard[1] xs
                                    Cons[2] y ys => x = y and equal xs ys

define discard[1] (xs: String[1]) : Bool =
  case unfold xs of Nil[1] => false | Cons[1] x xs => discard[1] xs
define discard[2] (ys: String[2]) : Bool =
  case unfold ys of Nil[2] => false | Cons[2] y ys => discard[2] ys

equal (gen S (fold Nil[1]))
      (fold (Cons[2] A (fold (Cons[2] A (fold (Cons[2] A (fold Nil[2])))))))
\end{lstlisting}

There are four occurrences of $\kw{fold}_{\text{\lstinline{String[2]}}}$, all on the last line. Since none of them have any free variables, we can defunctionalize \lstinline{String[2]}. Create new types \lstinline{String[2]Folded} and \lstinline{String[2]Unfolded}, which take the roles of $\nonrec$ and $\nonrecone$ in \cref{sec:defunctionalization}:
\begin{lstlisting}
data String[2]Folded = in1 | in2 | in3 | in4
data String[2]Unfolded = Nil[2] | Cons[2] Terminal String[2]Folded
\end{lstlisting}
The four values of type \lstinline{String[2]Folded} can be thought of as positions between the symbols of the input string $aaa$ (including the beginning and end of the string).
The new function \lstinline{u_String[2]} can then be thought of as taking a string position and returning the input symbol at that position, together with its successor position.
\begin{lstlisting}
define u_String[2] (ys: String[2]Folded) : String[2]Unfolded =
  case ys of in1 => Cons[2] A in2 | in2 => Cons[2] A in3
             in3 => Cons[2] A in4 | in4 => Nil[2]
\end{lstlisting}
Functions \lstinline{gen} and \lstinline{discard[1]} remain the same, but \lstinline{equal}, \lstinline{discard[2]}, and the last line become{\divide\smallskipamount 3
\begin{lstlisting}
define equal (xs: String[1]) (ys: String[2]Folded) : Bool =
  case unfold xs of
    Nil[1] => case u_String[2] ys of Nil[2] => true | Cons[2] y ys => discard[2] ys
    Cons[1] x xs => case u_String[2] ys of Nil[2] => discard[1] xs
                                       Cons[2] y ys => x = y and equal xs ys

define discard[2] (ys: String[2]Folded) : Bool =
  case u_String[2] ys of Nil[2] => false | Cons[2] y ys => discard[2] ys

equal (gen S (fold Nil[1])) in1
\end{lstlisting}}

As for \lstinline{String[1]}, it cannot be defunctionalized, because \lstinline{fold (Cons[1] A acc)} has a free variable \lstinline{acc} of type \lstinline{String[1]}.
Can it be refunctionalized?
There are two $\kw{case} {\kw{unfold}_{\text{\lstinline{String[1]}}}}$ expressions, one in \lstinline{equal} and one in \lstinline{discard[1]}. Both bodies have type $\kw{Bool}$; the first has a free variable \lstinline{ys} with type \lstinline{String[2]Folded}, and the second has no free variables. So yes, we can refunctionalize \lstinline{String[1]} by creating new types \lstinline{String[1]Folded} and \lstinline{String[1]Unfolded}, which take the roles of $\nonrec$ and $\nonrecone$ in \cref{sec:refunctionalization}:
\begin{lstlisting}
type String[1]Folded = (String[2]Folded -> Bool) & (Unit -> Bool)
data String[1]Unfolded = Nil[1] | Cons[1] Terminal String[1]Folded
\end{lstlisting}
A \lstinline{String[1]Folded} can be thought of as a string accessed through two ``methods'' \citep[cf.][]{rendel-automatic}: the first (with type \lstinline{String[2]Folded->Bool}) compares it with a suffix of the input string, and the second (with type \lstinline{Unit->Bool}) always returns $\kwe{false}$.
These methods are implemented in the new function \lstinline{f_String[1]}:
\begin{lstlisting}
define f_String[1] (xs: String[1]Unfolded) : String[1]Folded = 
  <\\ys: String[2]Folded. case xs of
     Nil[1] => case u_String[2] ys of Nil[2] => true | Cons[2] y ys => discard[2] ys
     Cons[1] x xs => case u_String[2] ys of Nil[2] => discard[1] xs
                                        Cons[2] y ys => x = y and equal xs ys,
    \\(): Unit. case xs of Nil[1] => false | Cons[1] x xs => discard[1] xs>

define equal (xs: String[1]Folded) (ys: String[2]Folded) : Bool = xs.1 ys
define discard[1] (xs: String[1]Folded) : Bool = xs.2 ()
\end{lstlisting}
And occurrences of $\kw{fold}_{\text{\lstinline{String[1]}}}$ are changed to calls to \lstinline{f_String[1]}.
\begin{lstlisting}
define gen (lhs: Nonterminal) (acc: String[1]Folded) : String[1]Folded =
  case lhs of S => if flip then gen S (gen S acc) else f_String[1] (Cons[1] A acc)

equal (gen S (f_String[1] Nil[1])) in1
\end{lstlisting}

\newcommand{\semval}[1]{\overline{#1}}
No recursive types remain, so this program converts to an MSPE with a finite number of weight variables.
For an input string $w = w_1 \dotsm w_n$, there are $O(n^2)$ equations in $O(n^2)$ weight variables, and each equation has $O(n)$ terms.
They can be ordered so that each weight variable depends only on earlier variables, and solving them is equivalent to the CYK algorithm.

To show this equivalence with less clutter, we define
\begin{equation}
  \semval{\ell} \enspace=\enspace \withdenotation{1}{(\arrowdenotation{\plusdenotation{\ctor{\id{Fold}_\ell}}{\unitdenotation}}{\plusdenotation{\ctor{\id{True}}}{\unitdenotation}})} \enspace\in\enspace \denote{\text{\lstinline{String[1]Folded}}} \qquad 1 \le \ell \le n+1.
\end{equation}
The set $\denote{\text{\lstinline{String[1]Folded}}} = \denote{\text{\lstinline{(String[2]Folded -> Bool) & (Unit -> Bool)}}}$ has $2n+4$ possible values (not $O(2^n)$, thanks to our treatment of functions as nondeterministic pairs).
Of these, the $\semval{\ell}$ are the interesting ones; each corresponds to a generated suffix that is equal to the input suffix starting at position~$\ell$. Thus, it will be easiest to think of them as string positions.

Then the equations simplify to
\begin{align}
  \eta(\id{gen})(\arrowdenotation{\plusdenotation{\ctor{\id{S}}}{\unitdenotation}}{\arrowdenotation{\semval{k}}{\semval{i}}}) &=
  \sum_j p \cdot \eta(\id{gen})(\arrowdenotation{\plusdenotation{\ctor{\id{S}}}{\unitdenotation}}{\arrowdenotation{\semval{j}}{\semval{i}}}) \cdot \eta(\id{gen})(\arrowdenotation{\plusdenotation{\ctor{\id{S}}}{\unitdenotation}}{\arrowdenotation{\semval{k}}{\semval{j}}}) \notag \\
  &\quad + q \cdot \eta(\text{\lstinline{f_String[1]}})(\arrowdenotation{\plusdenotation{\ctor{\id{Cons}}}{(\plusdenotation{\ctor{\id{A}}}{\unitdenotation}, \semval{k})}}{\semval{i}}) \label{eq:cky1} \\
  \eta(\text{\lstinline{f_String[1]}})(\arrowdenotation{\plusdenotation{\ctor{\id{Cons}}}{(\plusdenotation{\ctor{\id{A}}}{\unitdenotation}, \semval{k})}}{\semval{i}}) &= \mathbb{I}[w_i = a \land i = k-1] \label{eq:cky2}
\end{align}
and the probability of the whole program being true is $\eta(\id{gen})(\arrowdenotation{\plusdenotation{\ctor{\id{S}}}{\unitdenotation}}{\arrowdenotation{\semval{n+1}}{\semval{1}}})$.
\Cref{eq:cky2} corresponds to the initialization in CYK\@; it gives weight 1 to occurrences of terminal symbols.
\Cref{eq:cky1} corresponds to the main triple loop in CYK\@; it computes the weight of all derivations $S \Rightarrow^\ast w_i \cdots w_{k-1}$.

\subsection{Inferring the Sequence of Transformations}
\label{sec:infer_dr_sequence}

In the previous section, we saw that not all recursive types are amenable to both defunctionalization and refunctionalization.
Furthermore, some transformations can preclude others. For example, even though \lstinline{String[2]} could be refunctionalized, it would preclude refunctionalizing \lstinline{String[1]} because, in function \lstinline{equal}, it would change the type of variable \lstinline{ys} from \lstinline{String[2]} to \lstinline{(Unit -> Bool) & ((Terminal ** String[1]) -> Bool) & (Unit -> Bool)}. This type in turn contains \lstinline{String[1]}, and \lstinline{ys} occurs free within the ${\kw{case} {\kw{unfold}_{\text{\lstinline{String[1]}}}} \id{xs}}$ expression.
Although it may seem that such dependencies between transformations would make it difficult to find a successful sequence of transformations that eliminates all recursive types,
the good news is that there is a simple greedy algorithm for deciding whether a successful sequence of transformations exists, and if so, to find~it.
\begin{trivlist}
\item While there are any recursive types remaining:
  \begin{itemize}
  \item If there is a recursive type $\rec$ such that $\defunc\rec{\rec}$ contains no recursive type, defunctionalize $\rec$.
  \item If there is a recursive type $\rec$ such that $\refunc\rec{\rec}$ contains no recursive type, refunctionalize~$\rec$.
  \item Else, there is no successful sequence of transformations.
  \end{itemize}
\end{trivlist}
We give the formal statement of correctness here and defer its proof to the appendix.
\begin{definition}
  A \emph{DR-sequence} $S$ for a program~$p$ is a sequence $\mathcal{F}^{(1)}, \dots, \mathcal{F}^{(n)}$,
  where each $\mathcal{F}^{(i)}$ is either $\mathcal{D}_\rec$ or $\mathcal{R}_\rec$ for some recursive type $\rec$.
  We say that $S$ is \emph{successful} if
  it is empty and $p$ has no recursive types,
  or if $S = \mathcal{F} \cdot S'$, where
  $\transform{\mathcal{F}}{p}$ is well-defined (that is, $\transform{\mathcal{F}}{\rec}$ does not contain~$\rec$) and
  $S'$ is a successful DR-sequence for~$\transform{\mathcal{F}}{p}$.
\end{definition}
\begin{proposition}
\label{thm:infer_dr_sequence}
  Let $p$ be a program with at least one recursive type and a successful DR-sequence.
  \begin{enumerate}
  \item\label{item:greedy1} There is a recursive type $\rec$ and a transformation $\mathcal{F} \in \{\mathcal{D}, \mathcal{R}\}$ such that $\transform{\mathcal{F}_\rec}{\rec}$ contains no recursive type.
  \item\label{item:greedy2} For any such $\rec$ and $\mathcal{F}$, the program $\transform{\mathcal{F}_\rec}{p}$ is well-defined and has a successful DR-sequence.
  \end{enumerate}
\end{proposition}
\begin{proof}
    See \cref{sec:proof_infer_dr_sequence}.
\end{proof}

A note on complexity:
The transformations $\mathcal{D}$ and~$\mathcal{R}$ only move code around, so any DR-sequence will not make the program much bigger.
However, some successful DR-sequences may lead to more efficient inference than others.
Moreover, there exist programs whose size is blown up exponentially by the monomorphization that needs to take place before the DR-sequence.

%% file: examples.tex
\section{Further Examples: Pushdown Automata}
\label{sec:pda_examples}

PERPL automates the derivation of nonobvious inference algorithms from natural probabilistic models expressed using unbounded recursion.
We've seen already how
\cref{ex:parity-producer-consumer-again} derived the CYK algorithm from code that generates strings from a PCFG and compares them with an input string.
In this section, we present three further examples, all related to pushdown automata (PDAs), to illustrate what kinds of recursive data structures PERPL can and cannot handle.

\subsection{One Stack}
\label{eg:pda}

The parser of \cref{sec:parsing_examples} could also be implemented by converting the weighted CFG to a weighted nondeterministic PDA and running the PDA\@.
Define a type for stacks:
\begin{lstlisting}
data Stack = StkNil | StkCons Nonterminal Stack
\end{lstlisting}
(Our implementation of PERPL supports polymorphic datatypes such as \lstinline{List} and monomorphizes their uses such as \lstinline{List} \lstinline{Nonterminal}.)
The stack is initially just one~\lstinline{S}\@.
The \lstinline{run_pda} function below runs a nondeterministic PDA on a string.
This PDA keeps no state, and accepts by empty stack.
\begin{lstlisting}
define run_pda (zs: Stack) (ws: String) : Bool =
  case unfold zs of
    StkNil => case unfold ws of Nil => true                             -- accept
                                Cons w ws => discard ws                 -- reject
    StkCons z zs => z = S and                                           -- pop S
      if flip then                                                      
        run_pda (fold (StkCons S (fold (StkCons S zs)))) ws             -- push SS
      else                                                              
        case unfold ws of Nil => discard_stk zs                         -- reject
                          Cons w ws => w = A and run_pda zs ws          -- scan a
define discard (ws: String) =
  case unfold ws of Nil => false | Cons w ws => discard ws
define discard_stk (zs: Stack) =
  case unfold zs of StkNil => false | StkCons z zs => discard_stk zs

run_pda (fold (StkCons S (fold StkNil)))
        (fold (Cons A (fold (Cons A (fold (Cons A (fold Nil)))))))
\end{lstlisting}
For simplicity, we have hard-coded the PDA into this program; it~would be easy to add more symbols and transitions.

Unlike our previous examples using recursive data, a program that uses a stack like this is not structured as a producer--consumer pipeline.
Nevertheless, \lstinline{String} can be defunctionalized as in \cref{ex:parity-producer-consumer-again},
yielding the following, where
\lstinline{StringFolded}, \lstinline{StringUnfolded}, \lstinline{u_String}, and \lstinline{discard} are as before:
\begin{lstlisting}
define run_pda (zs: Stack) (ws: StringFolded) : Bool =
  case unfold zs of
    StkNil => case u_String ws of Nil => true                           -- accept
                                  Cons w ws => discard ws               -- reject
    StkCons z zs => z = S and                                           -- pop S
      if flip then                                                       
        run_pda (fold (StkCons S (fold (StkCons S zs)))) ws             -- push SS
      else                                                               
        case u_String ws of Nil => discard_stk zs                       -- reject
                            Cons w ws => w = A and run_pda zs ws        -- scan a
    
run_pda (fold (StkCons S (fold StkNil))) in1
\end{lstlisting}
Although \lstinline{Stack} cannot be defunctionalized because there are two occurrences of $\kw{fold}_\id{Stack}$ with a free variable \lstinline{zs} of type \lstinline{Stack}, the two occurrences of $\kw{unfold}_\id{Stack}$ meet the criterion for refunctionalization. So~we refunctionalize \lstinline{Stack} by creating new types \lstinline{StackFolded} and \lstinline{StackUnfolded}:
\begin{lstlisting}
type StackFolded = (StringFolded -> Bool) & (Unit -> Bool)
data StackUnfolded = StkNil | StkCons Nonterminal StackFolded

define f_Stack (zs: StackUnfolded) : StackFolded =
  <\\ws: StringFolded. case zs of
       StkNil => case u_String ws of Nil => true                        -- accept
                                     Cons w ws => discard ws            -- reject
       StkCons z zs => z = S and                                        -- pop S
         if flip then                                                   
           run_pda (f_Stack (StkCons S (f_Stack (StkCons S zs)))) ws    -- push SS
         else                                                           
           case u_String ws of Nil => discard_stk zs                    -- reject
                               Cons w ws => w = A and run_pda zs ws,    -- scan a
    \\(): Unit. case zs of StkNil => false | StkCons z zs => discard_stk zs>
  
define run_pda (zs: StackFolded) (ws: StringFolded) : Bool = zs.1 ws
define discard_stk (zs: StackFolded) : Bool = zs.2 ()
        
run_pda (f_Stack (StkCons S (f_Stack StkNil))) in1
\end{lstlisting}
Solving the corresponding MSPE is similar to Lang's algorithm for PDA recognition \citep{lang:1974}; as a matter of fact, refunctionalization has rederived a faster version of the algorithm found only recently \citep{butoi+:2022}. As observed by \citet{danvy-refunctionalization}, this final program resembles the final program in \cref{ex:parity-producer-consumer-again}, but written in continuation-passing style: the \lstinline{Stack} data structure has been replaced by \lstinline{StackFolded}, the type of continuations expecting a \lstinline{StringFolded} or a \lstinline{Unit}.

\subsection{Two Stacks}
\label{sec:2pda}

An example of a program where defunctionalization and refunctionalization fail would be one that processes strings while maintaining two stacks instead of one.
Since a two-stack PDA is equivalent to a Turing machine, such a program cannot be tractable, and it is desirable for PERPL to reject it at compile time.
We omit a full program listing, but it would have a form like:
\begin{lstlisting}
define run_2pda (left: Stack[1]) (right: Stack[2]) (ws: String) : Bool =
  case unfold left of
    $\cdots$ case unfold right of
          $\cdots$ case unfold ws of
                $\cdots$ fold (StkCons[1] z left) $\cdots$  -- push onto left stack
                $\cdots$ fold (StkCons[2] z right) $\cdots$ -- push onto right stack
\end{lstlisting}

As in \cref{eg:pda}, the push operations render both \lstinline{Stack} types non-defunctionalizable.
Moreover, here \lstinline{case unfold left} has free variable \lstinline{right}, so refunctionalizing \lstinline{Stack[1]} would change it into a type containing \lstinline{Stack[2]}, while \lstinline{case unfold right} has free variable \lstinline{left}, so refunctionalizing \lstinline{Stack[2]} would change it into a type containing \lstinline{Stack[1]}. Thus, by \cref{thm:infer_dr_sequence}, neither type is refunctionalizable.

\subsection{A Stack of Stacks}
\label{sec:epda}

In contrast to the two-stack PDA in \cref{sec:2pda}, if multiple stacks are nested, refunctionalization does succeed, and it turns out to convert nested stacks into nested continuations.
\emph{Embedded PDAs} (EPDAs) are a generalization of PDAs where memory is structured not as a stack of symbols but as a stack of \emph{stacks} of symbols \citep{vijay-shanker:1994}. They are equivalent to tree-adjoining grammars (TAGs), which generalize CFGs. Moreover, TAGs can be parsed using a CYK-style algorithm \citep{vijay-shankar-joshi-1985-computational}. Both of these results were major achievements in 1980s computational linguistics. Yet we show below that PERPL automatically takes a program that runs an EPDA and converts it into equations for a CYK-style TAG parser \citep[cf.][]{alonso-etal-2000-redefinition}.

In the following code, we append a \lstinline{*} to names relating to the stack-of-stacks.
\begin{lstlisting}
data Stack  = StkNil  | StkCons  Nonterminal Stack
data Stack* = StkNil* | StkCons* Stack       Stack*
\end{lstlisting}
Whenever the top stack is empty, it is automatically popped.
A move of the EPDA consists of popping a symbol \lstinline{z} from the stack and then doing one of the following:
\begin{lstlisting}
data Action =
  Pop                                -- do nothing further
  Scan Terminal                      -- scan terminal from input string
  Push Nonterminal Nonterminal       -- push y then x onto top stack
  PushAbove Nonterminal Nonterminal  -- push y then push stack [x] above top stack
  PushBelow Nonterminal Nonterminal  -- push x then push stack [y] below top stack
\end{lstlisting}
We assume a function called \lstinline{transition} with type \lstinline{Nonterminal -> Action}.
The following function decides whether the EPDA accepts a string. To sidestep the complication of explicit \lstinline{discard} functions, it~simply returns $\unitterm$ if the EPDA accepts and $\kw{fail}$s otherwise.
\begin{lstlisting}
define run_epda (zs*: Stack*) (ws: String) : Unit = case zs* of
  StkNil* => case ws of Nil => () | Cons w ws => fail
  StkCons* z* zs* => case z* of
    StkNil => run_epda zs* ws
    StkCons z zs => case transition z of
      Pop => run_epda (StkCons* zs zs*) ws
      Scan a => case ws of Nil => fail
                           Cons w ws => if w = a then run_epda (StkCons* zs zs*) ws
                                                 else fail
      Push x y => run_epda (StkCons* (StkCons x (StkCons y zs)) zs*) ws
      PushAbove x y =>
        run_epda (StkCons* (StkCons x StkNil) (StkCons* (StkCons y zs) zs*)) ws
      PushBelow x y =>
        run_epda (StkCons* (StkCons x zs) (StkCons* (StkCons y StkNil) zs*)) ws

run_epda (StkCons* (StkCons Z StkNil) StkNil*) (Cons A (Cons A (Cons A Nil)))
\end{lstlisting}

\noindent
Like in \cref{eg:pda}, we defunctionalize \lstinline{String} and refunctionalize \lstinline{Stack} and \lstinline{Stack*}, then simplify:
\begin{lstlisting}
type StackFolded = Stack*Folded -> StringFolded -> Unit
type Stack*Folded = StringFolded -> Unit

define f_StkNil* : Stack*Folded = \\ws: StringFolded. case u_String ws of
  Nil => () | Cons w ws => fail

define f_StkNil : StackFolded = \\zs*: Stack*Folded. \\ws: StringFolded. zs* ws
define f_StkCons (z: Nonterminal) (zs: StackFolded) : StackFolded =
  \\zs*: Stack*Folded. \\ws: StringFolded. case transition z of
    Pop => zs zs* ws
    Scan a => case u_String ws of Nil => fail
                                  Cons w ws => if w = a then zs zs* ws else fail
    Push x y => f_StkCons x (f_StkCons y zs) zs* ws
    PushAbove x y => f_StkCons x f_StkNil (f_StkCons y zs zs*) ws
    PushBelow x y => f_StkCons x zs (f_StkCons y f_StkNil zs*) ws

f_StkCons Z f_StkNil f_StkNil* in1
\end{lstlisting}

As before, \lstinline{StringFolded} can be thought of as a position in the input string, as can \lstinline{Stack*Folded}. Then \lstinline{StackFolded} can be thought of as a pair of string positions.
The most computationally expensive expression is \lstinline{f_StkCons x (f_StkCons y zs)}, which represents $O(n^4)$ weight variables because it and its free variable \lstinline{zs} both have type \lstinline{StackFolded}. Each of these weight variables is a sum over $O(n^2)$ possible values of the subexpression \lstinline{f_StkCons y zs}. Thus the program denotes a system of equations of size $O(n^6)$.
Readers familiar with TAG will recognize $O(n^6)$ as the time complexity of CYK-style TAG parsing.
Indeed, as we detail in \cref{sec:epda_details}, when this program is converted to an FGG\@, the FGG is essentially a TAG\@, and PERPL's inference amounts to CYK-style TAG parsing.

Above, we noted that the result of refunctionalizing \cref{eg:pda} was written in continuation-passing style. The twice-refunctionalized program above is written in extended continuation-passing style \citep{danvy+filinski:1990}, with \lstinline{StackFolded} being the type of continuations and \lstinline{Stack*Folded} being the type of metacontinuations.
Moreover, generalizations of EPDAs to automata with $k$-nested stacks (so-called $k$-EPDAs) can be expressed by similar PERPL programs and successfully compiled.
We hence conjecture more generally that PERPL's program transformations relate the \emph{control hierarchy} of context-free grammars \citep{weir-geometric} to the (serendipitously same-named) \emph{control hierarchy} of continuations \cite{sitaram-control,danvy+filinski:1990}.

%% file: benchmark.tex
\section{Benchmark}

We empirically demonstrate the improved speed and accuracy enabled by PERPL\@, by comparing against two existing probabilistic languages: WebPPL \citep{goodman-design}, which supports unbounded recursive calls and data, but does not perform exact inference on them; and Dice \citep{holtzen2020dice}, which performs exact inference, but only supports bounded loops.
Because so much of PERPL's expressive power lies in its novel support for unbounded recursive data, it~is tricky to find a benchmark that allows a quantitative comparison.
We took the PCFG parsing problem in \cref{ex:parity-producer-consumer} and expressed it as idiomatically as we could (\cref{s:pcfg-gen}): in WebPPL using unbounded recursion, and in Dice by unrolling a loop that builds up long parses from shorter ones.
All experiments were performed on a 2.8\thinspace GHz CPU with 16\thinspace GB of RAM\@.

\paragraph{Comparison with approximate inference}

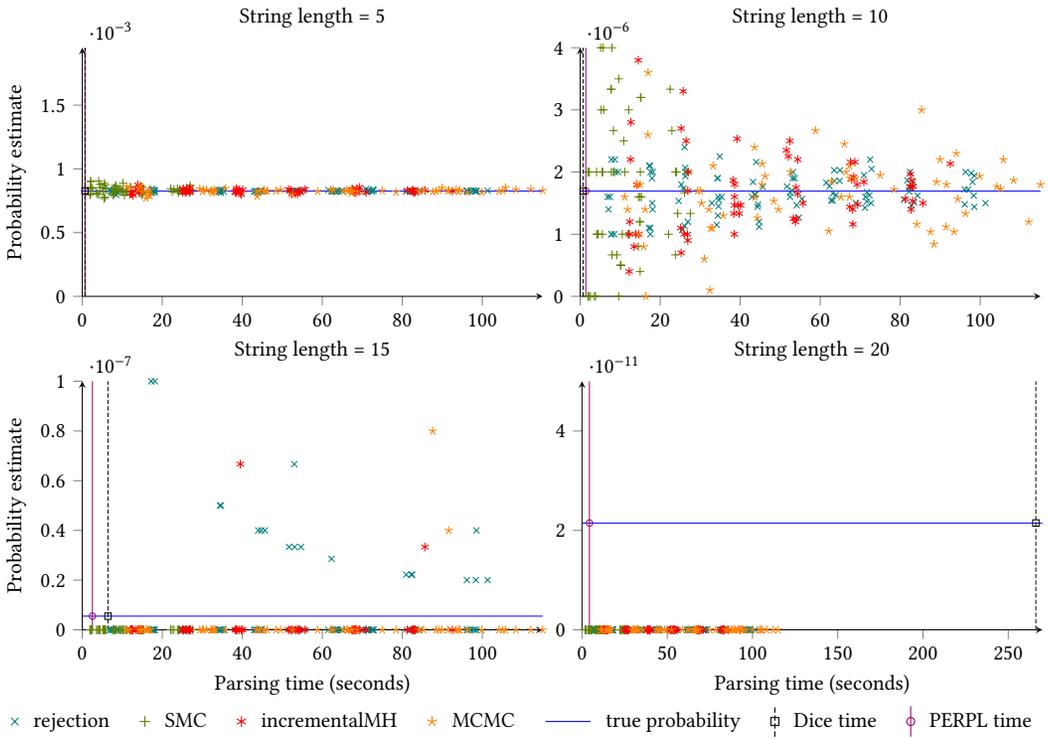
\begin{figure}
\centering
\newcommand\WebPPL[1]{%
        \begin{axis}[width=3in, height=1.6in, scale only axis, axis lines=left, tick align=outside,
                     title={String length = #1},
                     scatter/classes={rejection={mark=x,teal},
                                      SMC={mark=+,lime!50!black},
                                      incrementalMH={mark=asterisk,red},
                                      MCMC={mark=star,orange}}]
            \addplot [scatter, only marks, scatter src=explicit symbolic]
                table [col sep=comma, x=time, y=prob, meta=method]
                {benchmarks/pcfg-wppl-#1.csv};
            \csvreader[separator=tab, no head, filter strcmp={\csvcoli}{#1}]
                {benchmarks/pcfg-exact.tsv}{}{\edef\prob{\csvcolii}}
            \addplot [blue] coordinates { (0,0) }
                (axis cs:\pgfkeysvalueof{/pgfplots/xmin}, \prob) --
                (axis cs:\pgfkeysvalueof{/pgfplots/xmax}, \prob);
            \csvreader[separator=comma, filter strcmp={\csvcolix}{#1}]
                {benchmarks/pcfg-perpl.csv}{}{\edef\pplTime{\csvcolii}}
            \addplot [perpl] coordinates { (\pplTime,\prob) }
                (axis cs:\pplTime, \pgfkeysvalueof{/pgfplots/ymin}) --
                (axis cs:\pplTime, \pgfkeysvalueof{/pgfplots/ymax});
            \csvreader[separator=comma, filter strcmp={\csvcolix}{#1}]
                {benchmarks/pcfg-dice.csv}{}{\edef\diceTime{\csvcolii}}
            \addplot [dice] coordinates { (\diceTime,\prob) }
                (axis cs:\diceTime, \pgfkeysvalueof{/pgfplots/ymin}) --
                (axis cs:\diceTime, \pgfkeysvalueof{/pgfplots/ymax});
        \end{axis}}%
\pgfkeys{/pgfplots/xlabel near ticks,
         /pgfplots/ylabel near ticks,
         /pgfplots/xmin=0, /pgfplots/ymin=0,
         /pgfplots/vertical legend/.style={
             legend image code/.code={\draw [mark repeat=2,mark phase=2,#1] plot coordinates {(0,-0.2cm) (0,0) (0,0.2cm)};}
         }
}%
\resizebox{\textwidth}{!}{%
\begin{tabular}{@{}c@{}c@{}c@{}}
    \begin{tikzpicture}
        \pgfkeys{/pgfplots/ylabel={Probability estimate},
                 /pgfplots/ymax=1.9505654407077242e-3}
        \WebPPL{5}
    \end{tikzpicture}
&
    \begin{tikzpicture}
        \pgfkeys{/pgfplots/ymax=4e-6}
        \WebPPL{10}
    \end{tikzpicture}
\\
    \begin{tikzpicture}
        \pgfkeys{/pgfplots/xlabel={Parsing time (seconds)},
                 /pgfplots/ylabel={Probability estimate},
                 /pgfplots/ymax=1e-7}
        \WebPPL{15}
    \end{tikzpicture}
&
    \begin{tikzpicture}
        \pgfkeys{/pgfplots/xlabel={Parsing time (seconds)},
                 /pgfplots/xmax=270,
                 /pgfplots/ymax=5e-11}
        \WebPPL{20}
    \end{tikzpicture}
\end{tabular}}
\\
\begin{tikzpicture}
  \begin{axis}[hide axis,height=0.7in,
               legend style={
                 draw=none, inner sep=0pt, font={\footnotesize},
                 /tikz/every odd column/.append style={column sep=0.2cm},
                 /tikz/every even column/.append style={column sep=0.4cm, every node/.style={inner sep=0, anchor=mid}}
               },
               legend columns=-1,
               xmin=0, xmax=1, ymin=0, ymax=1]
    \addlegendimage{mark=x,teal,only marks}
    \addlegendentry{rejection}
    \addlegendimage{mark=+,lime!50!black,only marks}
    \addlegendentry{SMC}
    \addlegendimage{mark=asterisk,red,only marks}
    \addlegendentry{incrementalMH}
    \addlegendimage{mark=star,orange,only marks}
    \addlegendentry{MCMC}
    \addlegendimage{blue}
    \addlegendentry{true probability}
    \addlegendimage{vertical legend, dice}
    \addlegendentry{Dice time}
    \addlegendimage{vertical legend, perpl}
    \addlegendentry{PERPL time}
  \end{axis}
\end{tikzpicture}

\caption{The WebPPL version of \protect\cref{ex:parity-producer-consumer} (PCFG parsing), using a variety of inference methods, displays high variance for all but the shortest strings. Dice and PERPL obtain exact probabilities much faster. For string length 15, one data point with probability estimate $3\cdot10^{-7}$ has been omitted. We~ran rejection sampling up to 50\thinspace M samples, SMC up to 3\thinspace M particles, and incrementalMH and MCMC up to 30\thinspace M samples.}
\label{fig:webppl}
\end{figure}

\Cref{fig:webppl} shows that general-purpose approximate inference algorithms in WebPPL are inadequate for PCFG parsing, especially as the string parsed gets longer.
Each scatterplot depicts the result of trying to parse a different string length, and each point represents one inference run.
The horizontal axis shows the time taken; variation is caused primarily by different inference methods and different sample sizes.
The vertical axis shows the probability estimated; variation is caused by random sampling, which has trouble explaining low-probability events: when the string length is 15\@, all probability estimates are either 0 or a wild overestimate, and when the string length exceeds 19\@, all estimates are~0\@.
In~contrast, Dice and PERPL produce answers within $0.000001\%$ of the truth (shown by horizontal lines), in a fraction of the time taken by WebPPL (shown by vertical lines; also see below).

\input{dice}
\paragraph{Comparison with exact inference}

\Cref{fig:dice} shows that PERPL scales better than Dice to parsing longer strings.
The vertical axis shows time in log scale.
Along the horizontal axis, we varied the string length between 1 and~100\@, but Dice ran out of memory at length~22\@.
The plot shows that the running time of Dice (not to mention WebPPL) grows much more quickly than that of PERPL.

%% file: dice.tex
\begin{wrapfigure}[15]{r}{.52\textwidth}
\vspace*{-3ex}
\scalebox{.68}{%
\begin{tikzpicture}
    \begin{semilogyaxis}[width=.75\textwidth, height=.5\textwidth, axis lines=left, tick align=outside,
                         legend pos=south east, legend style={draw=none},
                         xlabel={String length}, ylabel={Parsing time (seconds)},
                         xlabel near ticks, ylabel near ticks,
                         mark size=1.5pt,
                         error bars/error mark options/.append={,solid},
                         error bars/error mark options/.append={,xshift=-1pt}, 
                         error bars/y dir=both, error bars/y explicit]
        \addplot [mark=x,teal]
            table [row sep=crcr, col sep=comma, x=length, y=mean,
                   y error expr=\thisrow{stddev}/sqrt(\thisrow{count})] {
command,mean,stddev,median,user,system,min,max,parameter_size,length,count\\
node '/tmp/159 rejection.js',0.18163816695000004,0.0023313041465508086,0.1810613572,0.20747699249999996,0.0159657625,0.1791810402,0.1868099262,159,1,16\\
node '/tmp/17484 rejection.js',0.22076230358461538,0.0026441655076991847,0.2201412792,0.29432871846153846,0.027742710769230763,0.2168911772,0.2256466162,17484,2,13\\
node '/tmp/103730 rejection.js',0.3755898757000001,0.038252558682726666,0.39829241170000007,0.48539378000000005,0.03474468,0.31992369820000005,0.4149076702000001,103730,3,10\\
node '/tmp/468570 rejection.js',0.8564300963000001,0.007853115896333763,0.8566300662000002,0.98615598,0.04218738,0.8451779032000001,0.8719463472000001,468570,4,10\\
node '/tmp/1851340 rejection.js',3.2892243228,0.9664557775378929,2.8826662417,3.44931608,0.07638948,2.5591205992,5.2800514972000006,1851340,5,10\\
node '/tmp/6886680 rejection.js',11.741958543299997,0.5804831325393106,11.8538665222,12.012017279999998,0.19969057999999998,10.6481063582,12.585761375199999,6886680,6,10\\
node '/tmp/24346848 rejection.js',42.3095981988,0.5395645640985457,42.28227547269999,43.97272798,0.73422298,41.2170988882,43.1503111822,24346848,7,10\\
            }
            -- (axis cs:8,143.7345099440861652) 
            -- (axis cs:9,475.2831637011542388) 
            -- (axis cs:10,1567.5328378334634987) 
            ;
        \addlegendentry{WebPPL}
        \addplot [dice]
            table [col sep=comma, x=parameter_length, y=mean,
                   y error expr=\thisrow{stddev}/sqrt(10)]
            {benchmarks/pcfg-dice.csv};
        \addlegendentry{Dice}
        \addplot [perpl]
            table [col sep=comma, x=parameter_length, y=mean,
                   y error expr=\thisrow{stddev}/sqrt(10)]
            {benchmarks/pcfg-perpl.csv};
        \addlegendentry{PERPL}
    \end{semilogyaxis}
\end{tikzpicture}}
\caption{The PCFG parser of \protect\cref{ex:parity-producer-consumer-again} scales to longer strings much better than its Dice equivalent. Parsing time ($y$-axis) is on a logarithmic scale.
         The WebPPL curve shows how long it takes rejection sampling just to get estimates within 5\% of the truth with probability 95\%.}
\label{fig:dice}
\end{wrapfigure}
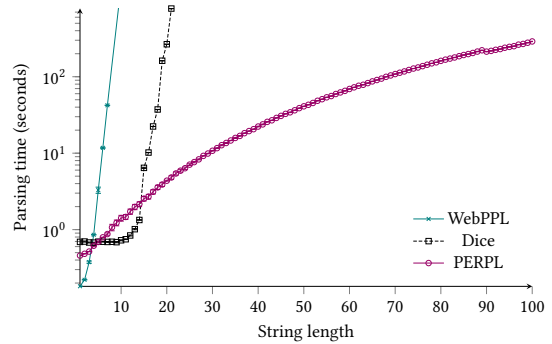%

%% file: appendix.tex
\section{Additional Proofs}

\subsection{Proof of Lemma \ref{thm:subst_preserves_denote_linear}} 
\label{sec:proof_subst_preserves_denote_linear}
By induction on the typing derivation of~$e'$. We show a few core cases.

\begin{trivlist}
\item Case $x_1$:
  Because $\Delta_0$ is empty, $\delta_0$ is empty as well.
  \begin{align*}
    \denote{x_1\{x_1 := e_1\}}(\delta_1, v') = \denote{e_1}(\delta_1, v')
    &= \sum_{v_1} \mathbb{I}[v'=v_1] \cdot \denote{e_1}(\delta_1, v_1) \\
    &= \sum_{v_1} \denote{x_1}(\{(x_1, v_1)\}, v') \cdot \denote{e_1}(\delta_1, v_1).
  \end{align*}

\item Case $x \neq x_1$: Ruled out by inversion.
  
\item Case $\lambda x. e$: By Barendregt's variable convention, $x \neq x_1$ and $x \not\in \fv{e_1}$.
  \begin{align*}
    \denote{(\lambda x. e) \{x_1 := e_1\}}(\delta_0 \cup \delta_1, \arrowdenotation{v_x}{v_e})
    &= \denote{\lambda x. (e \{x_1 := e_1\})}(\delta_0 \cup \delta_1, \arrowdenotation{v_x}{v_e}) \\
    &= \denote{e\{x_1:=e_1\}}(\delta_0 \cup \delta_1 \cup \{(x, v_x)\}, v_e) \\
    &= \sum_{v_1} \denote{e}(\delta_0 \cup \{(x,v_x), (x_1,v_1)\}, v_e) \cdot \denote{e_1}(\delta_1, v_1) \\
    &= \sum_{v_1} \denote{\lambda x. e}(\delta_0 \cup \{(x_1,v_1)\}, \arrowdenotation{v_x}{v_e}) \cdot \denote{e_1}(\delta_1, v_1).
  \end{align*}
  
\item Case $e'_0~e'_1$:
  Thanks to linearity, $x_1$ occurs free in either $e'_0$ or $e'_1$ but not both. If $x_1 \in \fv{e'_0}$,
  \begin{align*}
    \makebox[2em][l]{$\displaystyle \denote{(e'_0~e'_1) \{x_1 := e_1\}}(\delta'_0\cup\delta'_1\cup\delta_1, v') $}\\
    &= \denote{(e'_0 \{x_1 := e_1\})~e'_1}(\delta'_0\cup\delta'_1\cup\delta_1, v') \\
    &= \sum_{v'_1} \denote{e'_0 \{x_1 := e_1\}}(\delta'_0\cup\delta_1, \arrowdenotation{v'_1}{v'}) \cdot \denote{e'_1}(\delta'_1, v'_1) \\
    &= \sum_{v'_1} \sum_{v_1} \denote{e'_0} (\delta'_0\cup\{(x_1, v_1)\}, \arrowdenotation{v'_1}{v'}) \cdot \denote{e_1}(\delta_1, v_1) \cdot \denote{e'_1}(\delta'_1, v'_1) \\
    &= \sum_{v_1} \denote{e'_0~e'_1} (\delta'_0\cup\delta'_1\cup\{(x_1, v_1)\}, v') \cdot \denote{e_1}(\delta_1, v_1).
  \end{align*}
  Similarly if $x_1 \in \fv{e'_1}$.
\qed
\end{trivlist}

\subsection{Proofs of Type Preservation}

\label{sec:defunc_preserves_types}
\begin{proposition} \label{thm:defunc_preserves_types}
  If $e : \tau$, then $\defunc\rec{e} : \defunc\rec{\tau}$.
\end{proposition}
\begin{proof}
By induction on the typing derivation of $e : \tau$. The only interesting cases are $\kw{fold}_\rec$ and $\kw{unfold}_\rec$ expressions.
The nonlinear typing context is unchanged and elided below.

Case ${\kw{fold}_\rec \occ_i}$:
The translation $\defunc\rec{\kw{fold}_\rec \occ_i} = {\ctor{\id{Fold}_i} \free_i}$ has the following typing derivation.
\begin{gather*}
  \inferrule{\free_i : \freetype_i}{{\ctor{\id{Fold}_i} \free_i} : \nonrec}
\end{gather*}

Case ${\kw{unfold}_\rec x = e \kw{in} e'}$:
Because $\occ_i : \recone$, thanks to the induction hypothesis, the function $u_\rec$ has the following typing derivation.
\begin{gather*}
  \inferrule*{
    \inferrule*
        {\free_i : \freetype_i \vdash \defunc\rec{\occ_i} : \nonrecone \quad \text{for each~$i$}}
        {x : \nonrec \vdash {\kw{case} x \kw{of} {\ctor{\id{Fold}_1} \free_1} \casearrow \defunc\rec{\occ_1} \casealt \cdots \casealt {\ctor{\id{Fold}_n} \free_n} \casearrow \defunc\rec{\occ_n}} : \nonrecone}
  }
  {\lambda x.\, {\kw{case} x \kw{of} {\ctor{\id{Fold}_1} \free_1} \casearrow \defunc\rec{\occ_1} \casealt \cdots \casealt {\ctor{\id{Fold}_n} \free_n} \casearrow \defunc\rec{\occ_n}} : \nonrec \llp \nonrecone}
\end{gather*}
Then the translation of ${\kw{unfold}_\rec x = e \kw{in} e'}:\tau$ has the following typing derivation.
\begin{equation*}
  \inferrule*{
    \inferrule*{
      u_\rec : \nonrec \llp \nonrecone \quad
      \defunc\rec{e} : \nonrec
    }{
      u_\rec~\defunc\rec{e} : \nonrecone
    }
    \enspace
    x : \nonrecone \vdash \defunc\rec{e'} : \defunc\rec{\tau}
  }{
    \kw{let} x = u_\rec~\defunc\rec{e} \kw{in} \defunc\rec{e'} : \defunc\rec{\tau}
  }
\qedhere
\end{equation*}
\end{proof}

\label{sec:refunc_preserves_types}
\begin{proposition} \label{thm:refunc_preserves_types}
If $e : \tau$, then $\refunc\rec{e} : \refunc\rec{\tau}$.
\end{proposition}
\begin{proof}
By induction on the typing derivation of $e : \tau$. The only interesting cases are $\kw{fold}_\rec$ and $\kw{unfold}_\rec$ expressions.
The nonlinear typing context is unchanged and elided below.

Case ${\kw{fold}_\rec e}$:
Because $\occ_i : \occtype_i$, thanks to the induction hypothesis,
the function $f_\rec$ has the following typing derivation.
\begin{gather*}
  \inferrule*{
    \inferrule*{
      \inferrule*{
          x : \nonrecone, \free_i : \freetype_i
          \vdash \refunc\rec{\occ_i} : \occtype_i
      }{
        x : \nonrecone
        \vdash \lambda \free_i.\, \refunc\rec{\occ_i} : \freetype_i \llp \occtype_i}
    }{
      x : \nonrecone
      \vdash \withterm{ \lambda \free_1.\, \refunc\rec{\occ_1}, \ldots,
                        \lambda \free_n.\, \refunc\rec{\occ_n}}
      : \nonrec
    }
  }{
    \lambda x.\, \withterm{ \lambda \free_1.\, \refunc\rec{\occ_1}, \ldots,
                            \lambda \free_n.\, \refunc\rec{\occ_n} }
    : \nonrecone \llp \nonrec
  }
\end{gather*}
Then the translation of ${\kw{fold}_\rec e} : \nonrec$ has the following typing derivation.
\begin{gather*}
  \inferrule*{
    f_\rec : \nonrecone \llp \nonrec \\
    \refunc\rec{e} : \nonrecone
  }{
    f_\rec~\refunc\rec{e} : \nonrec
  }
\end{gather*}

Case ${\kw{unfold}_\rec x = \occ_i \kw{in} \occ_i}$:
The translation $\refunc\rec{\kw{unfold}_\rec x = \occ_i \kw{in} \occ_i} = (\refunc\rec{\occ_i}.i)~\free_i$ has the following typing derivation.
\begin{equation*}
  \inferrule*{\inferrule*{\refunc\rec{\occ_i} : \nonrec}
                         {\refunc\rec{\occ_i}.i : \freetype_i \llp \occtype_i} \\
              \free_i : \freetype_i}
             {(\refunc\rec{\occ_i}.i)~\free_i : \occtype_i}
\qedhere
\end{equation*}
\end{proof}

\subsection{Correctness of Transformations}
\label{sec:correctness_mu}

Sometimes, the programs before and after transformation can be executed in lockstep, given that our operational semantics allows evaluation anywhere in an expression.
In \cref{ex:parity-producer-consumer-again}, having fixed the target type $\rec=\text{\lstinline{String[2]}}$ and located the $n=4$ occurrences of $\kw{fold}_{\text{\lstinline{String[2]}}}$, we can treat the transformation $\defunc{\text{\lstinline{String[2]}}}{\cdot}$ as a relation~$\sim$ on terms,
\begin{equation}
\label{e:simplistic-lockstep}
    e^r \sim e^d
    \quad\text{iff}\quad
    \defunc{\text{\lstinline{String[2]}}}{e^r} = e^d
    \text.
\end{equation}
\Cref{fig:lockstep} shows an instance of such lockstep execution.

\newcommand\rotreduce{\raisebox{0.45cm}[\height+0.5cm]{\rotatebox{-90}{$\Longrightarrow$}}}
\newcommand\rotreducestar{\raisebox{0.45cm}[\height+0.5cm]{\rotatebox{-90}{$\Longrightarrow^{\mathclap{\ast}}$}}}
\newcommand\lstbox[1]{\text{\begin{lstlisting}#1\end{lstlisting}}}
\begin{figure}
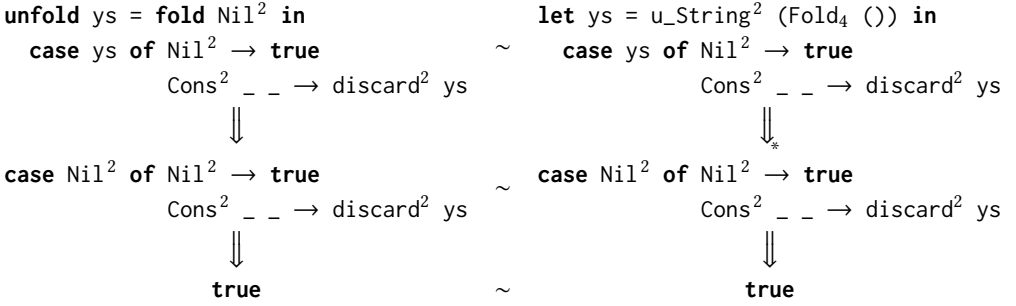

\begin{tabular}{ccc}
\begin{lstlisting}
unfold ys = fold Nil[2] in
  case ys of Nil[2] => true
             Cons[2] _ _ => discard[2] ys
\end{lstlisting}
& $\sim$ &
\begin{lstlisting}
let ys = u_String[2] (in4 ()) in
  case ys of Nil[2] => true
             Cons[2] _ _ => discard[2] ys
\end{lstlisting}
\\
\multicolumn{1}{c}{\rotreduce} & & \multicolumn{1}{c}{\rotreducestar} \\
\begin{lstlisting}
case Nil[2] of Nil[2] => true
             Cons[2] _ _ => discard[2] ys
\end{lstlisting}
& $\sim$ &
\begin{lstlisting}
case Nil[2] of Nil[2] => true
             Cons[2] _ _ => discard[2] ys
\end{lstlisting}
\\
\multicolumn{1}{c}{\rotreduce} & & \multicolumn{1}{c}{\rotreduce} \\
\multicolumn{1}{c}{\lstinline{true}}
& $\sim$ &
\multicolumn{1}{c}{\lstinline{true}}
\end{tabular}
\caption{Lockstep execution in a simple case of defunctionalization.
    We write $e$ for each distribution $\{(1, e)\}$.}
\label{fig:lockstep}
\end{figure}

We'd like to be able to show that if $e^r \sim e^d$, then for any reduction performed on one side, both sides can be reduced (in zero or more steps) to related distributions $E^r \sim E^d$, as in the above example.
Unfortunately, the relation defined
in~\eqref{e:simplistic-lockstep} is too simple to keep track of
reductions, especially
probabilistic choices, made inside $\kw{fold}_\rec$ expressions.
For example, instead of parsing a literal string in
\cref{ex:parity-producer-consumer-again}, it may be useful to parse a random string. When
an expression, say, $e^r = {\kw{fold}_\rec
(\kw{amb} {\dots}\enspace{\dots})}$ reduces on the left
side, the corresponding reduction(s)
on the right side would have to be made inside the
global definition of $u_\rec$.  Worse, different
occurrences of~$e^r$ can make different choices, and the
global definition of $u_\rec$ only has room to record one of them.

\input{sim}

\Cref{fig:sim} defines an enriched relation~$\sim$ that accounts for this complexity.
Like the simplistic attempt at~$\sim$ above, this relation depends implicitly on a defunctionalization target (fixing $\rec$, $\occ_i$, and $\free_i$).
However, what's related are no longer just terms ($e^r \sim e^d$), but distributions over terms ($E^r \sim E^d$).
These distributions are built by the Mixture rule, by taking linear combinations of related distributions.
Also, because the reduction relation $\Longrightarrow^*$ is parameterized by the global definitions~$\gamma$ on the left side, the relation~$\sim$ is parameterized by~$\gamma$ too.

We no longer define the relation~$\sim$ in terms of the transform~$\mathcal{D}_\rec$ as in~\eqref{e:simplistic-lockstep}.
Still, the rest of the inference rules in \cref{fig:sim} give rise to the transform~$\mathcal{D}_\rec$ as a consequence:
\begin{proposition}
    $ \dirac{e^r} \sim \dirac{\defunc\rec{e^r}} $.
\end{proposition}
\begin{proof}
    By induction on~$e^r$. The interesting cases---those shown in~\eqref{e:defunc_terms}---use
    the Fold and Unfold rules. The other, uninteresting cases use the
    Congruence rules.
\end{proof}

The Fold rule's second premise allows $\occ_i$ inside $\kw{fold}_\rec$ to be reduced arbitrarily and probabilistically to an entire distribution of terms and yet stay related to the same $\ctor{\id{Fold}_i}$ expression.
These reductions have executed already on the left side, in $E^r$, but are ``pent up'' on the right side.
They must wait for the $\kw{fold}_\rec$ on the left side to meet its annihilating $\kw{unfold}_\rec$, whose translation on the right side provides the $u_\rec$ whose definition contains $\occ_i$.

\newcommand\fststep[1]{E^{#1\dag}}
\newcommand\sndstep[1]{E^{#1\ddag}}

\begin{lemma}[Substitution preserves $\sim$]
\label{lemma:subst_preserves_sim}
    Suppose $\gamma\vdash E^r \sim E^d$ and $\gamma\vdash E^r_1 \sim E^d_1$.
    If $x_1$ is a linear variable, or if $x_1$ is a nonlinear variable but $E^r_1$ and~$E^d_1$ use no linear variables, then
    \begin{align*}
        \gamma\vdash{}&
        \bigl\{\bigl(w w_1, e^r\{x_1:=e^r_1\}\bigr) \mid (w,e^r)\in E^r, (w_1,e^r_1)\in E^r_1\bigr\}
    \\  {}\sim{}&
        \bigl\{\bigl(w w_1, e^d\{x_1:=e^d_1\}\bigr) \mid (w,e^d)\in E^d, (w_1,e^d_1)\in E^d_1\bigr\}
        \text.
    \end{align*}
\end{lemma}
\begin{proof}
    By induction on the derivation of $\gamma\vdash E^r \sim E^d$.
\end{proof}
\begin{theorem}
\label{thm:d-correct}Whenever $E^r \sim E^d$:
\begin{enumerate}
\item
    If $E^r \Longrightarrow \fststep{r}$, then
    there exist $\sndstep{r}$ and $\sndstep{d}$ such that
    $\fststep{r} \Longrightarrow^* \sndstep{r}$,
    $E^d         \Longrightarrow^* \sndstep{d}$, and
    $\sndstep{r} \sim \sndstep{d}$.
\item
    If $E^d \Longrightarrow \fststep{d}$, then
    there exist $\sndstep{r}$ and $\sndstep{d}$ such that
    $E^r         \Longrightarrow^* \sndstep{r}$,
    $\fststep{d} \Longrightarrow^* \sndstep{d}$, and
    $\sndstep{r} \sim \sndstep{d}$.
\end{enumerate}
\end{theorem}
\begin{proof}
    First observe that any use of the Mixture rule in a derivation of $E^r \sim E^d$ either occurs at the outermost (bottommost) level or can be commuted to occur there.
    And if $E^r \sim E^d$ is derived using the Mixture rule at the outermost level, then the result follows from the definition of $\Longrightarrow$ on distributions (\cref{def:red_distr}).
    Thus, we may assume from here on that the Mixture rule is not used, and so the only source of nondeterminism (\cref{def:deterministic}) is the second premise of the Fold rule.
    We proceed by simultaneous induction on the given derivations of $\sim$ and $\Longrightarrow$.

    If $E^r \sim E^d$ is derived using a Congruence rule, say, the one for application (the last rule shown in \cref{fig:sim}), then the given reduction ($E^r \Longrightarrow \fststep{r}$ or $E^d \Longrightarrow \fststep{d}$) is either a $\beta$-reduction at the top level or a reduction of a subexpression ($e^r_0$ or $e^r_1$ or $e^d_0$ or~$e^d_1$) inside an evaluation context. Handle the $\beta$-reduction case using \cref{lemma:subst_preserves_sim}. Handle the subexpression case using the induction hypothesis. In both cases, if the given reduction operates on a nondeterministic distribution (more precisely, if~$E_0$ in \cref{def:red_distr} is nonzero), then use the flexibility afforded by $\fststep{r} \Longrightarrow^* \sndstep{r}$ or $\fststep{d} \Longrightarrow^* \sndstep{d}$ to mimic the given reduction on the rest of the distribution.

    If $E^r \sim E^d$ is derived using the Unfold rule, then the given reduction is \eqref{e:unfold-fold} or \eqref{eq:reduction_congruence}.
    Handle these possibilities like with a Congruence rule. The only difference is that, after handling \eqref{e:unfold-fold} using \cref{lemma:subst_preserves_sim}, any pent-up reductions from the Fold rule need to be released into~$E^d \Longrightarrow^* \sndstep{d}$.

    If $E^r \sim E^d$ is derived using the Fold rule, then a reduction $E^r \Longrightarrow \fststep{r}$ can just be pent up in the witness for the second premise.
\end{proof}
Finally, we can show that the program (possibly surrounded with some code that converts the answer to $\unitterm$) denotes the same distribution before and after $\mathcal{D}$.
\begin{proposition}
    $ \bigl\{\bigl(w^r,\unitterm\bigr)\bigr\} \sim
      \bigl\{\bigl(w^d,\unitterm\bigr)\bigr\}$
    just in case $w^r = w^d$.
\end{proposition}
\begin{proof}
    By induction on the derivation of~$\sim$.
\end{proof}

\input{sim-r}

Refunctionalization is treated similarly, using the $\sim$ relation shown in \cref{fig:sim-r}.

\subsection{Inferring the Sequence of Transformations}
\label{sec:proof_infer_dr_sequence}

At first glance, DR-sequences are difficult to reason about, because each transformation potentially changes the types available to subsequent transformations. Instead, we define a graph structure on the recursive types occurring in~$p$ and show that finding a successful DR-sequence is equivalent to a finding a particular kind of subgraph.
\begin{definition}
  The \emph{DR-graph} of a program~$p$ is an edge-labeled, directed graph
  whose nodes are the recursive types in~$p$.
  There is an edge $\rec \defuncsto \tau$
  iff $\defunc\rec{\rec}=\nonrec$ contains $\tau$, 
  and an edge $\rec \refuncsto \tau$
  iff $\refunc\rec{\rec}=\nonrec$ contains $\tau$. 
\end{definition}

\begin{example}
  The DR-graph of the program of \cref{ex:parity-producer-consumer-again} is shown in \cref{fig:dr_graph}a.
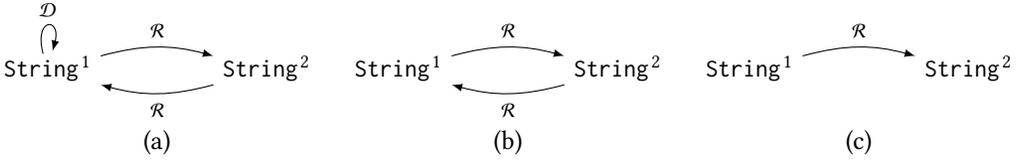
\begin{figure}
  \centering
  \tikzset{node distance=1.5cm}
  \begin{tabular}{@{}ccc@{}}
    \begin{tikzpicture}[dr]
      \node(str){\mbox{\lstinline{String[2]}}}; 
      \node(stk)[left=of str]{\mbox{\lstinline{String[1]}}};
      \draw[bend left=15]  (str) edge node {$\scriptstyle\mathcal{R}$} (stk);
      \draw[loop above]    (stk) edge node {$\scriptstyle\mathcal{D}$} (stk);
      \draw[bend left=15]  (stk) edge node {$\scriptstyle\mathcal{R}$} (str);
    \end{tikzpicture}
    &
    \begin{tikzpicture}[dr]
      \node(str){\mbox{\lstinline{String[2]}}};
      \node(stk)[left=of str]{\mbox{\lstinline{String[1]}}};
      \draw[bend left=15] (str) edge node {$\scriptstyle\mathcal{R}$} (stk);
      \draw[bend left=15] (stk) edge node {$\scriptstyle\mathcal{R}$} (str);
    \end{tikzpicture}
    &
    \begin{tikzpicture}[dr]
      \node(str){\mbox{\lstinline{String[2]}}};
      \node(stk)[left=of str]{\mbox{\lstinline{String[1]}}};
      \draw[bend left=15] (str) edge[draw=none] (stk);
      \draw[bend left=15] (stk) edge node {$\scriptstyle\mathcal{R}$} (str);
    \end{tikzpicture} \\
    (a) & (b) & (c)
  \end{tabular}
  \caption{(a) The full DR-graph of the program of \cref{ex:parity-producer-consumer-again}. (b) The DR-subgraph of refunctionalizing both \lstinline{String[1]} and \lstinline{String[2]}. (c) The DR-subgraph of refunctionalizing \lstinline{String[1]} and defunctionalizing \lstinline{String[2]}.}
  \label{fig:dr_graph}
\end{figure}
  Observe that because type \lstinline{String[1]} has a self-loop labeled $\mathcal{D}$, defunctionalization of \lstinline{String[1]} is not allowed. The $\mathcal{R}$ edges form a cycle, which does not immediately prevent either type from being refunctionalized, but if we did refunctionalize \text{String[2]}, the cycle would become a self-loop.
Then \lstinline{String[1]} would be neither defunctionalizable nor refunctionalizable. Intuitively, then, we want to avoid choices of defunctionalization and refunctionalization that form cycles.
\end{example}

\begin{definition}
  A \emph{DR-subgraph} $G$ for a program~$p$ is a subgraph of $p$'s DR-graph
  where there is a mapping from each recursive type~$\rec$ to a transformation $G(\rec) \in \{\mathcal{D}, \mathcal{R}\}$
  such that $G$ keeps just the edges $\rec \transformsto{G(\rec)} \tau$ in the DR-graph.
  Intuitively, $G(\rec)$~is what $G$ plans to do with~$\rec$.
\end{definition}

\begin{example}
  In \cref{ex:parity-producer-consumer-again}, the DR-subgraph of refunctionalizing both \lstinline{String[1]} and \lstinline{String[2]}, shown in \cref{fig:dr_graph}b, is cyclic, while
  the DR-subgraph corresponding to refunctionalizing \lstinline{String[1]} and defunctionalizing \lstinline{String[2]}, shown in \cref{fig:dr_graph}c, is acyclic.
\end{example}

\begin{definition}
  Let $G$ be a DR-subgraph for a program~$p$.
  For any recursive type~$\rec$ in~$p$ and transformation $\mathcal{F} \in \{\mathcal{D}, \mathcal{R}\}$, we write $\transform{\mathcal{F}_\rec}{G}$ for the unique DR-subgraph of $\transform{\mathcal{F}_\rec}{p}$ such that $\transform{\mathcal{F}_\rec}{G}(\transform{\mathcal{F}_\rec}{\tau}) = G(\tau)$ for all $\tau \neq \rec$.
  Intuitively, $\transform{\mathcal{F}_\rec}{G}$ is what $G$ plans to do with all the other recursive types after $\rec$ is eliminated via $\mathcal{F}$.
\end{definition}

\begin{lemma}
  \label{thm:evolve}
  Let $G$ be a DR-subgraph for a program~$p$, and let $\rec$ be a recursive type in~$p$.
  \begin{enumerate}
  \item \label{subthm:preserve_cycles}
    If $\transform{G(\rec)_\rec}{p}$ is defined (that is, if $G$ does not have a self-loop $\rec \rightarrow \rec$),
    then $\transform{G(\rec)_\rec}{G}$ is acyclic iff $G$ is.
  \item \label{subthm:remove_leaf}
    For any $\mathcal{F} \in \{\mathcal{D}, \mathcal{R}\}$ such that $\transform{\mathcal{F}_\rec}{\rec}$ contains no recursive type,
    $\transform{\mathcal{F}_\rec}{G}$ is acyclic iff $G$ is.
  \end{enumerate}
\end{lemma}
\begin{proof}
    For any $\mathcal{F} \in \{\mathcal{D}, \mathcal{R}\}$ in general, $\transform{\mathcal{F}_\rec}{G}$ is acyclic iff $G$ is.
    This follows from the observation that for any $\tau, \tau' \neq \rec$, the graph $\transform{\mathcal{F}_\rec}{G}$ has the edge $\transform{\mathcal{F}_\rec}{\tau} \rightarrow \transform{\mathcal{F}_\rec}{\tau'}$ iff the graph $G$ has either the edge $\tau \rightarrow \tau'$ or the edges $\tau \rightarrow \rec \rightarrow \tau'$.
    This observation holds because $\mathcal{F}_\rec$ is injective on nodes, thanks to the requirement that different recursive types in~$p$ must have different tags.
\end{proof}

\begin{proposition}
  A program~$p$ has a successful DR-sequence iff it has an acyclic DR-subgraph~$G$.
\end{proposition}
\begin{proof}
  By induction on the number of recursive types in~$p$.

  ($\Rightarrow$) If $p$ has a successful DR-sequence $\mathcal{F}_\rec \cdot S'$, that means $\transform{\mathcal{F}_\rec}{p}$ exists and has a successful DR-sequence (namely $S'$), so it has an acyclic DR-subgraph $G'$ by the induction hypothesis.
  Let $G$ be the unique DR-subgraph for $p$ such that $G(\rec)=\mathcal{F}$ and $G' = \transform{\mathcal{F}_\rec}{G}$. \Cref{thm:evolve}.\ref{subthm:preserve_cycles}($\Rightarrow$) says that $G$ must be acyclic.

  ($\Leftarrow$) Let $\rec$ be a recursive type in $p$, and let $\mathcal{F} = G(\rec)$. Since $G$ is acyclic, the program $p' = \transform{\mathcal{F}_\rec}{p}$ exists and has a DR-subgraph $G' = \transform{\mathcal{F}_\rec}{G}$, which is also acyclic by \cref{thm:evolve}.\ref{subthm:preserve_cycles}($\Leftarrow$). By the induction hypothesis, $p'$ has a successful DR-sequence $S'$, so $\mathcal{F}_\rec \cdot S'$ is a successful DR-sequence for $p$.
\end{proof}

\begin{proof}[Proof of \cref{thm:infer_dr_sequence}]
  Because $p$ has a successful DR-sequence, it has an acyclic DR-subgraph~$G$.
  (\ref{item:greedy1}) Because $G$ is acyclic and finite, it must have a node $\rec$ with no outgoing edges. Let $\mathcal{F} = G(\rec)$. Then $\transform{\mathcal{F}_\rec}{\rec}$ contains no recursive type.
  (\ref{item:greedy2}) For any such $\rec$ and $\mathcal{F}$, by \cref{thm:evolve}.\ref{subthm:remove_leaf}($\Leftarrow$), $\transform{\mathcal{F}_\rec}{G}$ is acyclic. So $\transform{\mathcal{F}_\rec}{p}$ has an acyclic DR-subgraph and therefore a successful DR-sequence.
\end{proof}

%% file: sim.tex
\begin{figure}
\[\begin{array}{@{}lc@{}}
\makebox[2em][l]{Mixture} &
    \inferrule{\gamma\vdash E^r_j \sim E^d_j \enspace \text{for all $j\in J$ finite}}
              {\textstyle \gamma\vdash \sum_{j\in J} w_j \cdot E^r_j \sim \sum_{j\in J} w_j \cdot E^d_j}
\\\\
\text{Fold} &
    \inferrule{\gamma \vdash E^r_j \sim E^d_j \enspace \text{for $1\le j\le \freelength$} \\
               \gamma \vdash \bigl\{\bigl(w_1\dotsm w_\freelength, {\kw{fold}_\rec \occ_i\{\free_i:=(\freesubi^r_1,\dotsc,\freesubi^r_\freelength)\}}\bigr) \mid (w_j,\freesubi^r_j)\in E^r_j \text{ for $1\le j\le \freelength$}\bigr\} \Longrightarrow^* E^r}
              {\gamma \vdash E^r \sim \bigl\{\bigl(w_1\dotsm w_\freelength, {\ctor{Fold}_i(\freesubi^d_1,\dotsc,\freesubi^d_\freelength)}\bigr) \mid (w_j,\freesubi^d_j)\in E^d_j \text{ for $1\le j\le \freelength$}\bigr\}}
\\\\
\makebox[2em][l]{Unfold} &
    \inferrule{\gamma \vdash E^r \sim E^d \\ \gamma \vdash E^{\prime r} \sim E^{\prime d}}
              {\relax
               {\begin{array}[t]{@{}r@{}l@{}}
                 \gamma \vdash{}& \bigl\{\bigl(ww',{\kw{unfold}_\rec x = e^r \kw{in} e^{\prime r}}\bigr) \mid (w,e^r)\in E^r, (w',e^{\prime r})\in E^{\prime r}\bigr\} \\
                        {}\sim{}& \bigl\{\bigl(ww',{\kw{let} x = u_\rec~e^d \kw{in} e^{\prime d}}\bigr) \mid (w,e^d)\in E^d, (w',e^{\prime d})\in E^{\prime d}\bigr\}
                \end{array}}}
\\\\
\makebox[2em][l]{Congruence} &
    \inferrule{}
              {\gamma \vdash \{(1,x)\}\sim\{(1,x)\}}
\qquad
    \inferrule{}
              {\gamma \vdash \{(1,\kw{fail})\}\sim\{(1,\kw{fail})\}}
\\\\ &
    \inferrule{\gamma \vdash E^r \sim E^d}
              {\relax
               {\begin{array}[t]{@{}r@{}l@{}}
                 \gamma \vdash{}& \bigl\{\bigl(w, \lambda x.\, e^r\bigr) \mid (w,e^r)\in E^r\bigr\} \\
                        {}\sim{}& \bigl\{\bigl(w, \lambda x.\, e^d\bigr) \mid (w,e^d)\in E^d\bigr\}
                \end{array}}}
    \qquad
    \inferrule{\gamma \vdash E^r_0 \sim E^d_0 \\
               \gamma \vdash E^r_1 \sim E^d_1}
              {\relax
               {\begin{array}[t]{@{}r@{}l@{}}
                 \gamma \vdash{}& \bigl\{\bigl(w_0 w_1, e^r_0~e^r_1\bigr) \mid (w_0,e^r_0)\in E^r_0, (w_1,e^r_1)\in E^r_1\bigr\} \\
                        {}\sim{}& \bigl\{\bigl(w_0 w_1, e^d_0~e^d_1\bigr) \mid (w_0,e^d_0)\in E^d_0, (w_1,e^d_1)\in E^d_1\bigr\}
                \end{array}}}
\\\\ &
    \hfill\text{and so on for all other syntax constructions}
\end{array}\]
\caption{Relating distributions of expressions before and after defunctionalization. In the Fold and Unfold rules, the definitions of $\rec$, $\occ_i$, and $\free_i$ are as in \cref{sec:defunctionalization}.}
\label{fig:sim}
\end{figure}

%% file: sim-r.tex
\begin{figure}
\[\begin{array}{@{}lc@{}}
\text{Fold} &
    \inferrule{\gamma \vdash E^d_0 \sim E^r_0 \\
               \gamma \vdash \bigl\{\bigl(w,f_\rec~e^r_0\bigr) \mid (w,e^r_0) \in E^r_0 \bigr\} \Longrightarrow^* E^r}
              {\gamma \vdash \bigl\{\bigl(w,\kw{fold}_\rec e^d_0\bigr) \mid (w,e^d_0) \in E^d_0 \bigr\} \sim E^r}
\\\\
\text{Unfold} &
    \inferrule{\gamma \vdash E^d \sim E^r \\ \gamma \vdash E^d_j \sim E^r_j \enspace \text{for $1\le j\le \freelength$}}
              {\relax
               {\begin{array}[t]{@{}r@{}l@{}l@{}}
                 \gamma \vdash{}& \multicolumn{2}{@{}l@{}}{\bigl\{\bigl(w w_1 \dotsm w_\freelength, {\kw{unfold}_\rec x=e^d \kw{in} \occ_i\{\free_i:=(\freesubi^d_1,\dotsc,\freesubi^d_\freelength)\}}\bigr)} \\&&{} \mid (w,e^d)\in E^d, (w_j,\freesubi^d_j)\in E^d_j \text{ for $1\le j\le \freelength$}\bigr\} \\
                        {}\sim{}& \bigl\{\bigl(w w_1 \dotsm w_\freelength, (e^r.i)~(\freesubi^r_1,\dotsc,\freesubi^r_\freelength)\bigr) &{} \mid (w,e^r)\in E^r, (w_j,\freesubi^r_j)\in E^r_j \text{ for $1\le j\le \freelength$}\bigr\}
                \end{array}}}
\end{array}\]
\caption{Relating distributions of expressions before and after refunctionalization. In the Fold and Unfold rules, the definitions of $\rec$, $\occ_i$, and $\free_i$ are as in \cref{sec:refunctionalization}. Mixture and Congruence rules are like in \cref{fig:sim}.}
\label{fig:sim-r}
\end{figure}

%% file: positive.tex
\subsection{Positive Types}
\label{sec:robust}

If $\tau$ has \emph{positive} polarity in the sense of linear logic \citep{girard-new,ehrhard-probabilistic}, that is, it uses only $\oplus$ and $\otimes$, intuitively it should be possible to copy values of type $\tau$, because a contraction (copying) function $\tau\llp(\tau\otimes\tau)$ can be easily defined by induction on $\tau$. Here, we show how to extend the type system to allow linear bindings to be turned into nonlinear bindings if they have positive type.

We say that $\tau$ is positive if it does not contain $\llp$, $\with$, or~$\mu$.
Thus, functions and additive pairs must still be used linearly.
Moreover, we decree that recursive types are not positive (so a value of recursive type must be consumed ($\kw{unfold}$ed) exactly once) in preparation for transforming recursive types to other types that might not be positive.

Let us write $\kw{let} \oc x = e \kw{in} e'$ to use $x:\tau$ any number of times in~$e'$.
\begin{trivlist}
\item Syntax
\begin{align*}
  e &\bnf {\kw{let} \oc x = e \kw{in} e}
\end{align*}
\item Typing
\begin{align*}
  \inferrule{\Gamma;\Delta \vdash e:\tau \\ \tau\text{ is positive} \\ \Gamma,x:\tau;\Delta' \vdash e':\tau'}
            {\Gamma;\Delta,\Delta' \vdash {\kw{let} \oc x = e \kw{in} e'} : \tau'}
\end{align*}
\end{trivlist}


The operational semantics of the new construct is as follows.
\begin{trivlist}
\item Reduction
\begin{equation}
\label{eq:reduction_ofc}
    \gamma \vdash {\kw{let} \oc x = v \kw{in} e'} \Longrightarrow \dirac{ e'\{x:=v\} }
\end{equation}
\item Evaluation contexts
\begin{equation*}
    C \bnf {\kw{let} \oc x=C \kw{in} e} \alt {\kw{let} \oc x=e \kw{in} C}
\end{equation*}
\end{trivlist}

\begin{lemma}[Nonlinear substitution preserves typing]
    (cf.~\cref{thm:subst_preserves_typing_linear})
    \label{thm:subst_preserves_typing_nonlinear}%

          Suppose\/ \(\Gamma, x_1\colon\tau_1; \Delta_0 \vdash e' : \tau'\)
          and\/ \(\Gamma; \cdot \vdash e_1 : \tau_1\).
          Then \(\Gamma; \Delta_0 \vdash e'\{x_1:=e_1\} : \tau'\).
\end{lemma}
\begin{proof}
  By induction on the typing derivation of $e'$.
\end{proof}
Hence the new reduction~\eqref{eq:reduction_ofc} preserves typing (extending \cref{thm:reduction_preserves_typing}):
$v$~uses no linear variables,
because it is a syntactic value whose type is positive.

The denotational semantics of the new construct extends the local environment~$\delta$ with the new nonlinear binding.
\begin{equation}
  \denote{\kw{let} \oc x=e \kw{in} e'}(\delta\cup\delta',v') =
  \sum_{v} \denote{e}(\delta,v) \cdot \denote{e'}(\delta'\cup\{(x,v)\}, v')
\end{equation}
This extension makes the division of names between the nonlinear and linear typing contexts $\Gamma$ and~$\Delta$ no longer match the division of names between the global and local environments $\gamma$ and~$\delta$.
To maintain this match, the denotation could alternatively be defined by extending the global environment~$\gamma$ with a distribution that is always deterministic.

\begin{lemma}[Nonlinear substitution preserves denotation]
\postdisplaypenalty=500
(cf.~\cref{thm:subst_preserves_denote_linear})

          Suppose\/ \(\Gamma, x_1\colon\tau_1; \Delta_0 \vdash e' : \tau'\)
          and\/ \(\Gamma; \cdot \vdash e_1 : \tau_1\).
          Then\label{thm:subst_preserves_denote_nonlinear}
  \begin{equation}
    \denote{e'\{x_1:=e_1\}}_\eta(\delta_0, v')
    = \sum_{v_1} \, \denote{e'}_\eta(\delta_0 \cup \{(x_1, v_1)\}, v') \cdot \denote{e_1}_\eta(\emptyset, v_1)
  \end{equation}
for all $\eta \in \denote{\Gamma}$,
$\delta_0 \in \denote{\Delta_0}$, and
$v' \in \denote{\tau'}$
such that $\denote{e_1}_\eta$ is deterministic.
\end{lemma}
\begin{proof}
By induction on the typing derivation of~$e'$. We show a few core cases.

\begin{trivlist}
\item Case $x_1$: As in \cref{sec:proof_subst_preserves_denote_linear}, but let $\delta_1=\emptyset$.
\item Case $x \neq x_1$: The assumption that $\denote{e_1}$ is deterministic justifies the second step:
  \begin{align*}
    \denote{x\{x_1:=e_1\}}(\delta_0, v') = \denote{x}(\delta_0, v')
    &= \sum_{v_1} \denote{x}(\delta_0 \cup \{(x_1, v_1)\}, v') \cdot \denote{e_1}(\emptyset, v_1).
  \end{align*}
\item Case $\lambda x. e$: As in \cref{sec:proof_subst_preserves_denote_linear}, but let $\delta_1=\emptyset$.
\item Case $e'_0~e'_1$: The assumption that $\denote{e_1}$ is deterministic justifies the second-to-last step:
  \begin{align*}
    \makebox[0em][l]{$\displaystyle \denote{(e'_0~e'_1) \{x_1 := e_1\}}(\delta'_0\cup\delta'_1, v') $}\\
    &= \denote{(e'_0 \{x_1 := e_1\})~(e'_1 \{x_1 := e_1)\}}(\delta'_0\cup\delta'_1, v') \\
    &= \sum_{v'_1} \denote{e'_0 \{x_1 := e_1\}}(\delta'_0, \arrowdenotation{v'_1}{v'}) \cdot \denote{e'_1 \{x_1 := e_1\}}(\delta'_1, v'_1) \\
    &= \sum_{v'_1} \Bigl(\sum_{v_1} \denote{e'_0}(\delta'_0 \cup \{(x_1,v_1)\}, \arrowdenotation{v'_1}{v'}) \cdot \denote{e_1}(\emptyset,v_1)\Bigr)
                   \Bigl(\sum_{v_1} \denote{e'_1}(\delta'_1 \cup \{(x_1,v_1)\}, v'_1) \cdot \denote{e_1}(\emptyset,v_1)\Bigr) \\
    &= \sum_{v'_1} \sum_{v_1} \denote{e'_0}(\delta'_0 \cup \{(x_1,v_1)\}, \arrowdenotation{v'_1}{v'}) \cdot
                              \denote{e'_1}(\delta'_1 \cup \{(x_1,v_1)\}, v'_1) \cdot \denote{e_1}(\emptyset,v_1) \\
    &= \sum_{v_1} \denote{e'_0~e'_1}(\delta'_0 \cup \delta'_1 \cup \{(x_1,v_1)\}, v') \cdot \denote{e_1}(\emptyset,v_1).
  \qedhere
  \end{align*}
\end{trivlist}
\end{proof}
Hence the new reduction~\eqref{eq:reduction_ofc} preserves denotation (extending \cref{thm:preserve}):
$\denote{v}$~is deterministic,
because $v$~is a syntactic value whose type is positive.

%% file: affine.tex
\subsection{Affine Types}
\label{sec:affine}

The previous extension does not add any flexibility for nonpositive types (those containing $\llp$, $\with$, or~$\mu$).
But we can further extend the language to allow values of nonpositive type to be used affinely (zero times or once).
It is easy to express this extension in the type system, by allowing the linear typing context~$\Delta$ to be non-empty in the rules for variables:
\begin{equation*}
  \inferrule{ }{\Gamma, x:\tau; \Delta \vdash x:\tau} \qquad
  \inferrule{ }{\Gamma; \Delta, x:\tau \vdash x:\tau}
\end{equation*}
However, allowing functions to be unused would be problematic for our denotational semantics, because any effects (nondeterminism and probabilities) inside the function body will be incurred regardless of whether the function is called; similarly for other nonpositive types.

To resolve this problem in a modular manner, we define a source-to-source transformation that operates on the typing derivation of a program and produces an equivalent program that uses values of nonpositive type only linearly. The idea is that whenever we want to create an affinely used function $\lambda x.e$, we create a $\with$-pair containing a linearly used function ($\lambda x.e$) or nothing (that is, $\unitterm$). To apply the function, we use the former; to discard~it, we use the latter.
We do something similar with
recursive types, which can now be unfolded zero times or once, and
additive pairs ($\with$), which can also be eliminated zero times or once.

\newcommand\linearize[1]{\mathcal{L}\llbracket #1\rrbracket}
\newcommand\discard[1]{\mathcal{Z}\llbracket #1\rrbracket}

Formally, we define the transformation $\mathcal{L}$ on types:
\begin{align*}
  \linearize{\tau_1\llp\tau'} &= (\linearize{\tau_1}\llp\linearize{\tau'}) \with \unittype \\
  \linearize{\tau_1\otimes\dots\otimes\tau_n} &= \linearize{\tau_1}\otimes\dots\otimes\linearize{\tau_n} \\
  \linearize{\tau_1\with\dots\with\tau_n} &= (\linearize{\tau_1}\with\dots\with\linearize{\tau_n}) \with \unittype \\
  \linearize{\ctor{c_1}\tau_1\oplus\dots\oplus \ctor{c_n}\tau_n} &= \ctor{c_1}\linearize{\tau_1}\oplus\dots\oplus \ctor{c_n}\linearize{\tau_n} \\
  \linearize{\mu\alpha.\,\tau} &= \mu\alpha.\,\linearize{\tau} \\
  \linearize{\alpha} &= \alpha
\end{align*}
On terms, we show the interesting cases:
\begin{align*}
  \linearize{\Gamma,x;\Delta \vdash x} &= {\kw{let} \unitterm = \discard{\Delta} \kw{in} x} \hidewidth \\
  \linearize{\Gamma;\Delta,x \vdash x} &= {\kw{let} \unitterm = \discard{\Delta} \kw{in} x} \hidewidth \\
  \linearize{\Gamma;\Delta_0 \vdash \lambda x_1. e'} &= \withterm{\lambda x_1. \linearize{e'},\discard{\Delta_0}} &
  \linearize{e_0~e_1} &= \linearize{e_0}.1~\linearize{e_1} \\
  \linearize{\Gamma;\Delta \vdash \withterm{e_1,\dots,e_n}} &= \withterm{\withterm{\linearize{e_1},\dots,\linearize{e_n}},\discard{\Delta}} &
  \linearize{e.i} &= \linearize{e}.1.i \\
  \linearize{\Gamma;\Delta \vdash {\kw{fold} e}} &= {\kw{fold} \linearize{e}} &
  \mathllap{\linearize{\kw{unfold}x=e\kw{in}e'}} &= {\kw{unfold}x=\linearize{e}\kw{in}\linearize{e'}}
\end{align*}

The $\discard{\cdot}$ transformation discards variables.
How this is done depends on the pre-transformation type
of each variable:
\begin{align*}
  \discard{\cdot} &= \unitterm \\
  \discard{\Delta, x:\tau} &= {\kw{let} \unitterm = \discard{\Delta} \kw{in} \discard{x : \tau}} \\
  \discard{x : \tau_1 \llp \tau'} &= x.2 \\
  \discard{x : {\tau_1 \with \dots \with \tau_n}} &= x.2 \\
  \discard{x : {\mu^t\alpha.\,\tau} } &= {\kw{discard}_{\mu^t\alpha.\tau}~x} \\
  \discard{x : \tau_1 \otimes \dots \otimes \tau_n} &= {\kw{let} (x_1, \dots, x_n) = x \kw{in} \discard{x_1 : \tau_1, \dotsc, x_n : \tau_n}} \\
  \discard{x: {\ctor{c_1}\tau_1 \oplus \dots \oplus \ctor{c_n}\tau_n}} &= {\kw{case} x \kw{of} {\ctor{c_1} x_1} \casearrow \discard{x_1:\tau_1} \casealt \dots \casealt {\ctor{c_n} x_n} \casearrow \discard{x_n:\tau_n}}
\end{align*}
The $\kw{discard}_{\mu^t\alpha.\tau}$ function used above is defined globally for each recursive type $\mu^t\alpha.\tau$. It~calls itself wherever $\alpha$ appears in~$\tau$.
\begin{align*}
    \kw{define} {\kw{discard}_{\mu^t\alpha.\tau}} = \lambda f\colon(\mu^t\alpha.\tau).\, {\kw{unfold}u=f\kw{in}\discard{u:\tau\{\alpha:=\mu^t\alpha.\tau\}}}
\end{align*}

%% file: letrec.tex
\subsection{Local Recursion}
\label{sec:letrec}

Another mild extension to the language is to allow nonlinear (recursive) bindings locally, not just at the top level of a program.
The body of these bindings cannot contain linear variables free.
This extension is compatible with the previous ones.
\begin{trivlist}
\item Syntax
\begin{align*}
  e &\bnf {\kw{fix} x.\, e}
\end{align*}
\item Typing
\begin{align*}
  \inferrule{\Gamma,x:\tau;\cdot \vdash e:\tau}
            {\Gamma       ;\cdot \vdash {\kw{fix} x.\, e}:\tau}
\end{align*}
\item Reduction
\begin{equation}
  \gamma \vdash {\kw{fix} x.\, e} \Longrightarrow \dirac{ e\{x:=e\} }
\end{equation}
\item Denotation
\begin{equation}
  \denote{\kw{fix} x.\, e}_\eta(\emptyset,v) =
  \mathop{\mathrm{lfp}} \bigl(\lambda f.\, \lambda v.\, \denote{e}_{\eta, x=f}(\emptyset,v)\bigr) (v)
\end{equation}
\end{trivlist}
Unlike the previous two extensions, this extension has not been implemented; but it can be implemented simply by lambda lifting.

%% file: fgg.tex
\subsection{Factor Graph Grammars}

In this section, we give a condensed definition of FGGs \citep{chiang+riley:2020}.
Fix a finite set $L^V$ of \emph{node labels}, a finite set $L^E$ of \emph{edge labels}, and a function $\type \colon L^E \rightarrow (L^V)^\ast$, which says, for each edge label,
how many nodes an edge must be attached to, and what their labels must be.
For any function $f \colon A \rightarrow B$, we extend $f$ to strings over~$A$: $f^\ast(a_1 \cdots a_n) = f(a_1) \cdots f(a_n)$.

\begin{definition}
A \emph{hypergraph} (or \emph{graph} for short) is a tuple $(V, E, \att, \vlab, \elab, \ext)$, where
\begin{itemize}
\item $V$ names a finite set of \emph{nodes}.
\item $E$ names a finite set of \emph{edges}.
\item $\att \colon E \rightarrow V^\ast$ assigns to each edge~$e$ a sequence of zero or more \emph{attachment} nodes.
\item $\vlab \colon V \rightarrow L^V$ and $\elab \colon E \rightarrow L^E$ assign to each node and edge a label such that, for all $e$, $\type(\elab(e)) = \vlab^\ast(\att(e))$.
\item $\ext \in V^\ast$ is a sequence of zero or more \emph{external} nodes.
\end{itemize}
\end{definition}

A graph can be interpreted as a factor graph \citep{kschischang+:2001}, as follows. Fix a mapping $\dom$ from $L^V$ to sets; then each node $v$ corresponds to a random variable over $\dom(\vlab(v))$.
Extend $\dom$ to strings by defining $\dom^\ast(\ell_1 \cdots \ell_n) = \dom(\ell_1) \times \cdots \times \dom(\ell_n)$.
For each edge label $\ell \in L^E$, define a corresponding factor $\fac_\ell \colon \dom^\ast(\type(\ell)) \rightarrow \mathbb{R}_{\geq 0}$.

\begin{example}
  \label{eg:fg_pcfg}
  Below is a factor graph for trees of a certain shape generated by a PCFG with start symbol $S$. The variables range over nonterminal symbols or terminal symbols.
  \begin{center}
      \begin{tikzpicture}[x=1.25cm]
        \node[var] (n1) at (0,0) { };
        \node[var] (n2) at (2,1) { };
        \node[var] (n3) at (2,-1) { };
        \node[var] (w4) at (4,1) { };
        \node[var] (w5) at (4,-1) { };
        \node[fac,label=above:{${} = S$}] at (-1,0) {} edge (n1);
        \node[fac,label=right:{$p(A\rightarrow BC)$}] at (1,0) {} edge node[tent,above] {$A$} (n1) edge node[tent] {$B$} (n2) edge node[tent,auto=right] {$C$} (n3);
        \node[fac,label=above:{$p(B\rightarrow d)$}] at (3,1) {} edge node[tent] {$B$} (n2) edge node[tent,below] {$d$} (w4);
        \node[fac,label=below:{$p(C\rightarrow e)$}] at (3,-1) {} edge node[tent,above] {$C$} (n3) edge node[tent] {$e$} (w5);
      \end{tikzpicture}
  \end{center}
\end{example}
We draw an edge as a square with lines (called \emph{tentacles}) to its attachment nodes. If there is more than one tentacle, we label each with a name. We write the edge's factor function next to it, possibly in terms of the tentacle names. A~Boolean expression has value $1$ if true and $0$ if false.

\begin{definition} \label{def:hrg}
A \emph{hyperedge replacement graph grammar} (HRG) is a tuple $(N, T, P, S)$, where
\begin{itemize}
\item $N \subseteq L^E$ is a finite set of \emph{nonterminal symbols}.
\item $T \subseteq L^E$ is a finite set of \emph{terminal symbols}, disjoint from~$N$.
\item $P$ is a finite set of \emph{rules} of the form $(X \rightarrow R)$, where 
    $X \in N$, and
    $R$ is a hypergraph with node labels $\vlab$, edge labels in $N \cup T$,
    and external nodes $\ext$ such that $\type(X) = \vlab^\ast(\ext)$.
\item $S \in N$ is a distinguished \emph{start nonterminal symbol}.
\end{itemize}
\end{definition}
An HRG generates a set of graphs; we omit a formal definition here. Most definitions require $\type(S) = \epsilon$,
including our original definition \citep{chiang+riley:2020},
but here we relax this requirement, to allow the graphs generated by an HRG to have external nodes.

\input{derivation}

Under $\dom$ and $\fac$, an~HRG can be interpreted as a way of generating factor graphs. We call such an~HRG a \emph{factor graph grammar} (FGG\@).
\begin{example}
\label{eg:fgg_pcfg}
Below is an FGG for derivations of a PCFG in Chomsky normal form. The start symbol of the FGG is $\nt{S'}$.
\begin{align*}
\begin{tikzpicture} 
\node[fac] at (1,0) { $\nt{S'}$ };
\end{tikzpicture} 
&\longrightarrow 
\begin{tikzpicture} 
\node[var] (n) at (1,0) {}; 
\node[fac,label=above:{$\mathclap{{} = S}$}] at (0,0) {} edge (n);
\node[fac] at (2,0) {$\nt{X}$} edge (n); 
\end{tikzpicture} \\
\begin{tikzpicture} 
\node[ext](x) at (0,0) {};
\node[fac] at (1,0) {$\nt{X}$} edge (x); 
\end{tikzpicture} 
&\longrightarrow 
\begin{tikzpicture}
\node[ext] (n) at (0,0) {};
\node[var] (n1) at (2,1) {};
\node[var] (n2) at (2,-1) {};
\node[fac,label=right:{$p(A \rightarrow BC)$}] at (1,0) {} edge node[tent] {$A$} (n) edge node[tent] {$B$} (n1) edge node[tent,auto=right] {$C$} (n2);
\node[fac] at (3,1) {$\nt{X}$} edge (n1);
\node[fac] at (3,-1) {$\nt{X}$} edge (n2);
\end{tikzpicture} \\
\begin{tikzpicture}
  \node[ext](x) at (0,0) {};
  \node[fac] at (1,0) { $\nt{X}$ } edge (x);
\end{tikzpicture}
&\longrightarrow
\begin{tikzpicture}[x=1.2cm]
  \node[ext] (n) at (0,0) {};
  \node[var] (n1) at (2,0) {};
  \node[fac,label=above:{$p(A \rightarrow b)$}] at (1,0) {} edge node[tent] {$A$} (n) edge node[tent,below] {$b$} (n1);
\end{tikzpicture}
\end{align*}
This FGG generates an infinite number of factor graphs, including the one in \cref{eg:fg_pcfg}. The generation process can be pictured as in \cref{fig:derivation}.
\end{example}

We draw external nodes in black. If there is more than one, we write a name inside each.
We draw an edge with nonterminal label $X$ as a square with $X$ inside and tentacles labeled with names.

Although the left-hand side is formally just a nonterminal symbol, we draw it like an edge, with replicas of the external nodes as attachment nodes.

The graphs generated by an FGG can be viewed as factor graphs, each of which defines a distribution over assignments. Recall that we allow the graphs generated by an FGG to have external nodes; for present purposes, we are interested in the marginal distribution over assignments to those external nodes.
\begin{definition}
If $V$ is a set of nodes, an \emph{assignment} of~$V$ is a mapping~$\asst$ that takes each node $v \in V$ to a value in $\dom(\vlab(v))$. If $\vec\ell$ is a string of node labels, an assignment~$\bar\asst$ of $\vec\ell$ is a sequence in $\dom^\ast(\vec\ell)$.
\end{definition}
Let $G = (N, T, P, S)$ be an FGG whose variables have finite domains.
Write down the following system of equations:
For each nonterminal $X \in N$ and assignment~$\bar\asst$ of~$\type(X)$,
\begin{align}
  w_X(\bar\asst) &= \sum_{\substack{R \\ (X\rightarrow R) \in P}} w_R(\bar\asst). \label{eq:fggweight} \\
\intertext{For each right-hand side $R = (V, E, \att, \vlab, \elab, \ext)$ and assignment $\bar\asst$ of $\vlab(\ext)$,}
w_R(\bar\asst) &= \sum_{\substack{\text{assignments~$\asst$ of $V$}\\\asst^\ast(\ext) = \bar\asst}} \prod_{\substack{e \in E}} w_{\elab(e)}(\asst^\ast(\att(e))). \\
\intertext{For each terminal $\ell \in T$ and assignment~$\bar\asst$ of~$\type(\ell)$,}
w_\ell(\bar\asst) &= \fac_{\ell}(\bar\asst).
\end{align}
Take the least fixed point of these equations (under the ordering $w \leq w'$ iff $w_X(\bar\asst) \leq w'_X(\bar\asst)$ for all $X, \bar\asst$). Then $w_X$ is the distribution over the external nodes of the graphs generated starting from $X$, and $w_S$ is the distribution over the external nodes of the graphs generated by $G$. This is a monotone system of polynomial equations, which can be solved as in \cref{sec:solve_mspe}.

%% file: derivation.tex
\begin{figure*}
\begin{center}
  \begin{tikzpicture}
    \tikzset{replace/.style={-latex,snake=snake,line after snake=6pt}}
    \begin{scope}
      \node[var](n1) at (1,0) {};
      \node[fac](f0) at (0,0) {} edge (n1);
      \node[fac](f1) at (2,0) { $\nt{X}$ } edge (n1);
      \begin{scope}[on background layer]
        \node[fit=(f0)(n1)(f1),fill=gray!20,inner xsep=12pt] {};
      \end{scope}
    \end{scope}
    \begin{scope}[xshift=4cm,y=0.8cm]
      \node[ext] (n) at (0,0) {};
      \node[var] (n1) at (2,1) {};
      \node[var] (n2) at (2,-1) {};
      \node[fac] at (1,0) {} edge (n) edge (n1) edge (n2);
      \node[fac] (f4) at (3,1) {$\nt{X}$} edge (n1);
      \node[fac] (f5) at (3,-1) {$\nt{X}$} edge (n2);
      \begin{scope}[on background layer]
        \node[fit=(n)(n1)(n2)(f4)(f5),fill=gray!20,inner xsep=12pt](r1) {};
        \draw[replace] (r1) to (f1);
      \end{scope}
    \end{scope}
    \begin{scope}[xshift=9cm,yshift=0.8cm]
      \node[ext] (n) at (0,0) {};
      \node[var] (n1) at (2,0) {};
      \node[fac] at (1,0) {} edge (n) edge (n1);
      \begin{scope}[on background layer]
        \node[fit=(n)(n1),fill=gray!20,inner xsep=12pt] (r2) {};
        \draw[replace] (r2) to (f4);
      \end{scope}
    \end{scope}
    \begin{scope}[xshift=9cm,yshift=-0.8cm]
      \node[ext] (n) at (0,0) {};
      \node[var] (n1) at (2,0) {};
      \node[fac] at (1,0) {} edge (n) edge (n1);
      \begin{scope}[on background layer]
        \node[fit=(n)(n1),fill=gray!20,inner xsep=12pt] (r3) {};
        \draw[replace] (r3) to (f5);
      \end{scope}
    \end{scope}
  \end{tikzpicture}
\end{center}
\caption{An example derivation of the FGG in \cref{eg:fgg_pcfg}. A wavy arrow means that the nonterminal at the head of the arrow is rewritten with the graph at the tail.}
\label{fig:derivation}
\end{figure*}
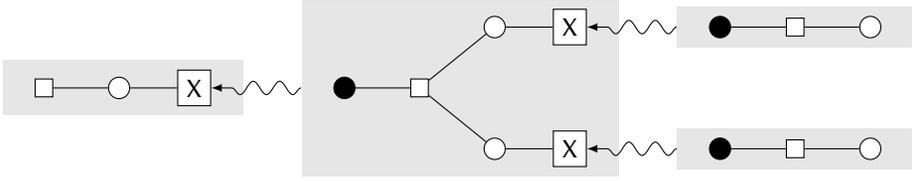

%% file: translation.tex
\subsection{Translation Rules}
\label{sec:translation}

Any PERPL program can be converted into an FGG that
preserves semantics in the sense that the denotation of
a subexpression~$e$ is equal to the weight function $w_e$ of its
translation (\ref{eq:fggweight}).

Each judgement $\Gamma; \Delta \vdash e : \tau$ in the typing derivation of a program translates to an FGG rule. The external nodes are:
\begin{itemize}
\item For each local binding $\env{x}:\env{\tau} \in \Delta$, there is an external node named $\env{x}:\env{\tau}$, which holds the value of $\env{x}$ when $e$ is evaluated.
\item There is an external node named $\val:\tau$, which holds the value of $e$.
\end{itemize}
We use plate notation \citep{buntine-operations,koller-pgm} to depict multiple nodes:
a rounded rectangle indicates that its contents are repeated, for each local variable~$\env{x}$.
We stress, however, that plates are only meta-notation; as will be clear in \cref{ex:inconsistent-fgg} below, actual FGG rules do \emph{not} use plates.

\begin{subequations}
\paragraph{Programs}

If the program has a global definition $\kw{define} x = e$, add the following rule for occurrences of $x$:
\begin{align}
\begin{tikzpicture}
\node[ext](x) {$\val$};
\node[fac,left=of x](f) {$x$} edge (x);
\end{tikzpicture}
&\longrightarrow
\begin{tikzpicture}
\node[ext](x) {$\val$};
\node[fac,left=of x](f) {$e$} edge node[tent](t){$\val$} (x);
\end{tikzpicture}
\label{eq:translate_global}
\end{align}

If the program is of the form $\ldots \seq e$, then the nonterminal $e$ is the start symbol.

\paragraph{Variables}

The rule for a local variable copies the node for that variable from the environment to the value:
\begin{align}
\begin{tikzpicture}
\node[fac](f) {$x$};
\node[ext,left=of f](x) {$x$} edge (f);
\node[ext,right=of f](v) {$\val$} edge (f);
\end{tikzpicture}
&\longrightarrow
\begin{tikzpicture}
\node[ext](xj) {$x$};
\node[fac,right=of xj,label={[name=eqlab]below:$=$}](eq) {} edge (xj);
\node[ext,right=of eq](v) {$\val$} edge (eq);
\end{tikzpicture}
\label{eq:translate_local}
\end{align}

\paragraph{Functions}

A function becomes a node ranging over input--output pairs.
\begin{align}
  \begin{tikzpicture}
    \node[fac](e) {$\lambda x_1. e'$};
    \node[ext,left=of e](ef) {$\env{x}$} edge (e);
    \node[plate,fit=(ef)] {};
    \node[ext,right=of e](v) {$\val$} edge (e);
  \end{tikzpicture}
  &\longrightarrow
  \begin{tikzpicture}
    \node[fac](e2) {$e'$};
    \node[ext,left=2em of e2](e2f) {$\env{x}$} edge node[tent,near start,name=e2ft] {$\env{x}$} (e2);
    \node[plate,fit=(e2f)(e2ft)] {};
    \node[fac,right=6em of e2.center,label={left:$\val_0=\arrowdenotation{\val_1}{\val'}$}](p) {};
    \node[var](x) at ($(e2)!0.5!(p)+(0,0.75cm)$) {} edge node[tent,auto=right,name=vxt] {$x_1$} (e2) edge node[tent]{$\val_1$} (p);
    \node[var](vo) at ($(e2)!0.5!(p)+(0,-0.75cm)$) {} edge node[tent,name=vot] {$\val$} (e2) edge node[tent,auto=right]{$\val'$} (p);
    \node[ext,right=of p](v) {$\val$} edge node[tent]{$\val_0$} (p);
  \end{tikzpicture}
\end{align}  
The rule for applications unpacks the function as an input--output pair, equates the input to the operand, and takes the output.
\begin{align}
  \begin{tikzpicture}
    \node[ext](xi) {$\env{x}$};
    \node[plate,fit=(xi)] {};
    \node[fac,right=of xi](f) {$e_0~e_1$} edge (xi);
    \node[ext,right=of f](v) {$\val$} edge (f);
  \end{tikzpicture}
  &\longrightarrow
  \begin{tikzpicture}
    \node[fac](e1) at (0,0.75cm) {$e_0$};
    \node[ext,left=2em of e1](e1f) {$\env{x}$} edge node[tent,near start](e1ft){$\env{x}$} (e1);
    \node[plate,fit=(e1f)(e1ft)] {};
    \node[var,right=of e1](v1) {} edge node[tent,above](e1tv1) {$\val$} (e1);
    \node[fac](e2) at (0,-0.75cm) {$e_1$};
    \node[ext,left=2em of e2](e2f) {$\env{x}$} edge node[tent,near start](e2ft){$\env{x}$} (e2);
    \node[plate,fit=(e2f)(e2ft)] {};
    \node[var,right=of e2](v2) {} edge node[tent,above](e2tv2) {$\val$} (e2);
    \node[fac,label={left:$\val_0=\arrowdenotation{\val_1}{\val'}$}](p) at ($(v1)!0.5!(v2)+(3em,0)$) {} edge node[tent,auto=right]{$\val_0$} (v1) edge node[tent]{$\val_1$} (v2);
    \node[ext,right=of p] {$\val$} edge node[tent]{$\val'$} (p);
  \end{tikzpicture}
\end{align}

\paragraph{Random variables}

The rules for $\kw{amb}$ and $\kw{factor}$ are very simple:
\begin{align}
\begin{tikzpicture}
\node[fac](f) {$\kw{amb} e_1~e_2$};
\node[ext,left=of f](xi) {$\env{x}$} edge (f);
\node[plate,fit=(xi)] {};
\node[ext,right=of f](v) {$\val$} edge (f);
\end{tikzpicture}
&\longrightarrow
\begin{tikzpicture}
\node[fac](e) {$e_1$};
\node[ext,left=2em of e](xi) {$\env{x}$} edge node[tent,near start](t){$\env{x}$} (e);
\node[plate,fit=(xi)(t)] {};
\node[ext,right=of e](v) {$\val$} edge node[tent]{$\val$} (e);
\end{tikzpicture} \\
\begin{tikzpicture}
\node[fac](f) {$\kw{amb} e_1~e_2$};
\node[ext,left=of f](xi) {$\env{x}$} edge (f);
\node[plate,fit=(xi)] {};
\node[ext,right=of f](v) {$\val$} edge (f);
\end{tikzpicture}
&\longrightarrow
\begin{tikzpicture}
\node[fac](e) {$e_2$};
\node[ext,left=2em of e](xi) {$\env{x}$} edge node[tent,near start](t){$\env{x}$} (e);
\node[plate,fit=(xi)(t)] {};
\node[ext,right=of e](v) {$\val$} edge node[tent]{$\val$} (e);
\end{tikzpicture} \\
\begin{tikzpicture}
\node[fac](f) {$\kw{factor} w \kw{in} e$};
\node[ext,left=of f](xi) {$\env{x}$} edge (f);
\node[plate,fit=(xi)] {};
\node[ext,right=of f](v) {$\val$} edge (f);
\end{tikzpicture}
&\longrightarrow
\begin{tikzpicture}
\node[fac](e) {$e$};
\node[ext,left=2em of e](xi) {$\env{x}$} edge node[tent,near start](t){$\env{x}$} (e);
\node[plate,fit=(xi)(t)] {};
\node[ext,right=of e](v) {$\val$} edge node[tent]{$\val$} (e);
\node[fac,above=0.25cm of e,label=above:{$w$}] {};
\end{tikzpicture}
\end{align}
And there are no rules for $\kw{fail}$, so that nonterminal $\kw{fail}$ has weight 0, as desired.

\paragraph{Unions}

For a union type $\ctor{c_1} \tau_1 \oplus \cdots \oplus \ctor{c_n} \tau_n$, the injections are straightforward to translate. For any constructor $\ctor{c_i}$:
\begin{align}
\begin{tikzpicture}
\node[fac](e) {${\ctor{c_i} e}$};
\node[ext,left=of e](ef) {$\env{x}$} edge (e);
\node[plate,fit=(ef)] {};
\node[ext,right=of e](v) {$\val$} edge (e);
\end{tikzpicture}
&\longrightarrow
\begin{tikzpicture}
\node[fac](e1) {$e_i$};
\node[ext,left=2em of e1](e1f) {$\env{x}$} edge node[tent,near start,above](e1ft){$\env{x}$} (e1);
\node[plate,fit=(e1f)(e1ft)] {};
\node[var,right=of e1](v1) {} edge node[tent](v1t){$\val$} (e1);
\node[fac,right=of v1,label={above:$\val=\plusdenotation{c_i}{\val_i}$}](inj) {} edge node[tent] {$\val_i$} (v1);
\node[ext,right=of inj](v) {$\val$} edge node[tent](v1t){$\val$} (inj);
\end{tikzpicture}
\end{align}
A single $\kw{case}$-expression translates to one rule for each constructor:
\begin{align}
\begin{tikzpicture}
\node[fac](e) {$\kw{case} e \kw{of} {\ctor{c_1} x_1} \casearrow e_1' \casealt \dots \casealt {\ctor{c_n} x_n} \casearrow e_n'$};
\node[ext,above=0.5cm of e](ef) {$\env{x}$} edge (e);
\node[plate,fit=(ef)] {};
\node[ext,right=of e](v) {$\val$} edge (e);
\end{tikzpicture}
&\longrightarrow
\begin{tikzpicture}
\node[fac](e0) {$e$};
\node[ext,above=of e0](e0f) {$\env{x}$} edge node[tent,near start](e0ft){$\env{x}$} (e0);
\node[plate,fit=(e0f)(e0ft)] {};
\node[var,right=of e0](v) {} edge node[tent] {$\val$} (e0);
\node[fac,right=of v,label=above:{$\val=\plusdenotation{c_i}{\val_i}$}](inj) {} edge node[tent]{$\val$} (v);
\node[var,right=of inj](v1) {} edge node[tent]{$\val_i$} (inj);
\node[fac,right=of v1](e1) {$e_i'$} edge node[tent](x1t){$x_i$} (v1);
\node[ext,above=of e1](e1f) {$\env{x}$} edge node[tent,near start](e1ft){$\env{x}$} (e1);
\node[plate,fit=(e1f)(e1ft)] {};
\node[ext,right=of e1](v) {$\val$} edge node[tent](vt){$\val$} (e1);
\end{tikzpicture} \label{eq:translate_case}
\end{align}
Since the desugaring of $\kw{if}$ leads to some superfluous nodes, we provide a direct translation as well:
\begin{align}
\begin{tikzpicture}
\node[fac](e) {$\kw{if} e \kw{then} e'_1 \kw{else} e'_2$};
\node[ext,above=0.5cm of e](ef) {$\env{x}$} edge (e);
\node[plate,fit=(ef)] {};
\node[ext,right=of e](v) {$\val$} edge (e);
\end{tikzpicture}
&\longrightarrow
\begin{tikzpicture}
\node[fac](e0) {$e$};
\node[ext,above=of e0](e0f) {$\env{x}$} edge node[tent,near start](e0ft){$\env{x}$} (e0);
\node[plate,fit=(e0f)(e0ft)] {};
\node[var,right=of e0](v) {} edge node[tent] {$\val$} (e0);
\node[fac,right=of v,label=above:{$\val=\plusdenotation{\ctor{\id{True}}}{\unitdenotation}$}](inj) {} edge node[tent]{$\val$} (v);
\node[fac,right=2em of inj](e1) {$e'_1$};
\node[ext,above=of e1](e1f) {$\env{x}$} edge node[tent,near start](e1ft){$\env{x}$} (e1);
\node[plate,fit=(e1f)(e1ft)] {};
\node[ext,right=of e1](v) {$\val$} edge node[tent](vt){$\val$} (e1);
\end{tikzpicture}
\\
\begin{tikzpicture}
\node[fac](e) {$\kw{if} e \kw{then} e'_1 \kw{else} e'_2$};
\node[ext,above=0.5cm of e](ef) {$\env{x}$} edge (e);
\node[plate,fit=(ef)] {};
\node[ext,right=of e](v) {$\val$} edge (e);
\end{tikzpicture}
&\longrightarrow
\begin{tikzpicture}
\node[fac](e0) {$e$};
\node[ext,above=of e0](e0f) {$\env{x}$} edge node[tent,near start](e0ft){$\env{x}$} (e0);
\node[plate,fit=(e0f)(e0ft)] {};
\node[var,right=of e0](v) {} edge node[tent] {$\val$} (e0);
\node[fac,right=of v,label=above:{$\val=\plusdenotation{\ctor{\id{False}}}{\unitdenotation}$}](inj) {} edge node[tent]{$\val$} (v);
\node[fac,right=2em of inj](e1) {$e'_2$};
\node[ext,above=of e1](e1f) {$\env{x}$} edge node[tent,near start](e1ft){$\env{x}$} (e1);
\node[plate,fit=(e1f)(e1ft)] {};
\node[ext,right=of e1](v) {$\val$} edge node[tent](vt){$\val$} (e1);
\end{tikzpicture}
\end{align}

This translation avoids a problem that \citet[Section 3.1]{vandemeent+:2018} encounter when translating conditional expressions to factor graphs. In a factor graph, both arms of the conditional must be translated, and every assignment to the factor graph must assign values to variables in both arms, even though only one can be active at a time. Their translation requires some complicated machinery to work around this problem, but our translation to an FGG can simply generate a different graph for each arm.

\paragraph{Tuples}

A single additive tuple translates into multiple rules, each generating just one member. The elimination form $e.i$ selects the factor graphs where the $i$-th member was generated.
\begin{align}
\begin{tikzpicture}
\node[ext](x1) {$\env{x}$};
\node[plate,fit=(x1)] {};
\node[fac,right=of x1](f) {$\withterm{e_1,\dots,e_n}$} edge (x1);
\node[ext,right=of f] {$\val$} edge (f);
\end{tikzpicture}
&\longrightarrow
\begin{tikzpicture}
\node[fac](e1) {$e_i$};
\node[ext,left=2em of e1](e1f) {$\env{x}$} edge node[tent,near start,above](e1ft){$\env{x}$} (e1);
\node[plate,fit=(e1f)(e1ft)] {};
\node[var,right=of e1](v1) {} edge node[tent](v1t){$\val$} (e1);
\node[fac,right=of v1,label={above:$\val=\withdenotation{i}{\val_i}$}](inj) {} edge node[tent] {$\val_i$} (v1);
\node[ext,right=of inj](v) {$\val$} edge node[tent](v1t){$\val$} (inj);
\end{tikzpicture}
\\
\begin{tikzpicture}
\node[ext](x1) {$\env{x}$};
\node[plate,fit=(x1)] {};
\node[fac,right=of x1](f) {$e.i$} edge (x1);
\node[ext,right=of f] {$\val$} edge (f);
\end{tikzpicture}
&\longrightarrow
\begin{tikzpicture}
\node[fac](e1) {$e$};
\node[ext,left=2em of e1](e1f) {$\env{x}$} edge node[tent,near start,above](e1ft){$\env{x}$} (e1);
\node[plate,fit=(e1f)(e1ft)] {};
\node[var,right=of e1](v1) {} edge node[tent](v1t){$\val$} (e1);
\node[fac,right=of v1,label={above:$\val=\withdenotation{i}{\val_i}$}](inj) {} edge node[tent] {$\val$} (v1);
\node[ext,right=of inj](v) {$\val$} edge node[tent](v1t){$\val_i$} (inj);
\end{tikzpicture}
\end{align}

Multiplicative tuples are more straightforward:
\begin{align}
\begin{tikzpicture}
\node[ext](x1) {$\env{x}$};
\node[plate,fit=(x1)] {};
\node[fac,right=of x1](f) {$(e_1, \dots, e_n)$} edge (x1);
\node[ext,right=of f](v) {$\val$} edge (f);
\end{tikzpicture}
&\longrightarrow
\begin{tikzpicture}
\node[ext](x1) at (0,1cm) {$\env{x}$};
\node[fac,right=2em of x1](f1) {$e_1$} edge node[tent,near end](x1t){$\env{x}$} (x1);
\node[plate,fit=(x1)(x1t)] {};
\node[var,right=of f1](v1) {} edge node[tent](v1t){$\val$} (f1);
\node at (0,0) {$\vdots$};
\node[ext](x2) at (0,-1cm) {$\env{x}$};
\node[fac,right=2em of x2](f2) {$e_n$} edge node[tent,near end](x2t){$\env{x}$} (x2);
\node[plate,fit=(x2)(x2t)] {};
\node[var,right=of f2](v2) {} edge node[tent](v2t){$\val$} (f2);
\node[fac,label={left:$\val=(\val_1,\dots,\val_n)$}](pair) at ($(v1)!0.5!(v2)+(1cm,0)$) {} edge node[tent,auto=right]{$\val_1$} (v1) edge node[tent]{$\val_n$} (v2);
\node[ext,right=of pair] {$\val$} edge node[tent]{$\val$} (pair);
\end{tikzpicture}
\raisetag{10pt}
\\
\begin{tikzpicture}
\node[fac](f) {$\kw{let} (x_1, \dots, x_n) = e \kw{in} e'$};
\node[ext,above=0.5cm of f](x1) {$\env{x}$} edge (f);
\node[plate,fit=(x1)] {};
\node[ext,right=of f](v) {$\val$} edge (f);
\end{tikzpicture}
&\longrightarrow
\begin{tikzpicture}
\node[fac](e3) {$e$};
\node[ext,above=of e3](x) {$\env{x}$} edge node[tent,near start](xt){$\env{x}$} (e3);
\node[plate,fit=(x)(xt)] {};
\node[var,right=of e3](v3) {} edge node[tent]{$\val$} (e3);
\node[fac,right=of v3,label={[rotate=90,right]$\val=(\val_1,\dots,\val_n)$}](pair) {} edge node[tent]{$\val$} (v3);
\node[fac,right=4.5em of pair](e4) {$e'$};
\node[ext,above=of e4](x) {$\env{x}$} edge node[tent,near start](xt){$\env{x}$} (e4);
\node[plate,fit=(x)(xt)] {};
\node[var](x1) at ($(pair)!0.5!(e4)+(0,0.75cm)$) {} edge node[tent,auto=right,inner sep=0.2mm](v1t) {$\val_1$} (pair) edge node[tent,inner sep=0.2mm](x1t){$x_1$} (e4);
\node[var](x2) at ($(pair)!0.5!(e4)+(0,-0.75cm)$) {} edge node[tent,inner sep=0.2mm](v2t){$\val_n$} (pair) edge node[tent,auto=right,inner sep=0.2mm](x2t){$x_n$} (e4);
\node[ext,right=of e4](v) {$\val$} edge node[tent](vt){$\val$} (e4);
\node at ($(pair)!0.5!(e4)+(0,0)$) {$\vdots$};
\end{tikzpicture}
\raisetag{10pt}
\end{align}

\input{inconsistent_fgg}

\end{subequations}

\begin{example} \label{ex:inconsistent-fgg}
Recall the program of \cref{ex:inconsistent}. This program translates to the FGG in \cref{fig:inconsistent_fgg}. When we compute the sum-product of the FGG\@, we get $\min\bigl(1, \frac{1-p}{p}\bigr)$ as desired.
\end{example}

\subsection{Correctness}

\begin{theorem} \label{thm:fggtranslation}
  Let $p = {\kw{define} x_1 = e_1} \seq \cdots \seq {\kw{define} x_n = e_n} \seq e_0$ be a program, and let $G$ be the FGG translated from $p$ as defined above.
Then $\eta$ as defined by (\ref{eq:fixpoint}) exists iff $w$ as defined by (\ref{eq:fggweight}) exists, and if both exist, then for all $e$ in~$p$,
\begin{equation} \label{eq:correspondence}
w_e(\delta, \val=v) = \denote{e}_\eta(\delta, v).
\end{equation}
In particular, $w_{e_0}(\val=v) = \denote{p}(v)$.
\end{theorem}

\begin{proof}
We first show that (\ref{eq:fixpoint}) has a fixed point $\eta$ iff (\ref{eq:fggweight}) has a solution $w$ such that (\ref{eq:correspondence}) is satisfied; then we will show that (\ref{eq:correspondence}) preserves ordering.

If $\eta$ is a fixed point of (\ref{eq:fixpoint}), then set $w$ according to (\ref{eq:correspondence}). The proof that $w$ is a solution to equation (\ref{eq:fggweight}) is by induction on the typing derivation of $e$; we show just a few cases:

If $e = x_i$ is a global variable:
\newcommand{\eqbynothing}{\stackrel{\hphantom{(00)}}{=}}
\begin{equation*}
  w_{x_i}(\val=v) \eqby{eq:correspondence} \denote{x_i}_\eta(\emptyset, v) \\ \eqby{eq:denote_var} \eta(x_i)(v) \\ \eqby{eq:fixpoint} \denote{e_i}_\eta(\emptyset, v) \\ \eqby{eq:correspondence} w_{e_i}(\val=v).
\end{equation*}
which matches equation (\ref{eq:fggweight}) for rule (\ref{eq:translate_global}). If $e = x$ is a local variable:
\begin{equation*}
w_x(x=v, \val=v') \eqby{eq:correspondence} \denote{x}_\eta(\{(x,v)\}, v') \eqby{eq:denote_var} \mathbb{I}[v=v'].
\end{equation*}
which matches equation (\ref{eq:fggweight}) for rule (\ref{eq:translate_local}).

The most complex case is probably that of $\kw{case}$ expressions. For $\delta \in \denote{\Delta}$ and $\delta' \in \denote{\Delta'}$:
\begin{align*}
  w_e(\delta\cup\delta', \val=v)
  &\eqby{eq:correspondence} \denote{e}(\delta \cup \delta', v) \\
  &\eqby{eq:denote_case} \sum_{i=1}^n \sum_{v_i} \denote{e_0}(\delta, (i, v_i)) \cdot \denote{e_i}(\delta' \cup \{(x_i,v_i)\}, v) \\
  &\eqby{eq:correspondence} \sum_{i=1}^n \sum_{v_i} w_{e_0}(\delta, \val=(i, v_i)) \cdot w_{e_i}(x_i=v_i, \delta', \val=v) \\
  &\eqbynothing \sum_{i=1}^n \sum_{v_0} \sum_{v_i} w_{e_0}(\delta, \val=v_0) \cdot \mathbb{I}[v_0 = (i, v_i)] \cdot w_{e_i}(x_i=v_i, \delta', \val=v)
\end{align*}
which matches equation (\ref{eq:fggweight}) for rules (\ref{eq:translate_case}).

Conversely, if $w$ is a solution to (\ref{eq:fggweight}), for all $x_i \in \Gamma$, set $\eta(x_i)(v) = w_{x_i}(\val=v)$ by (\ref{eq:correspondence}). We want to show that $\eta$ is a fixed point of (\ref{eq:fixpoint}), which is very similar to the above.

Finally, consider two pairs of solutions related by (\ref{eq:correspondence}):
\begin{align*}
w_e(\delta, \val=v) &= \denote{e}_\eta(\delta, v) \\
w'_e(\delta, \val=v) &= \denote{e}_{\eta'}(\delta, v).
\end{align*}
If $w \leq w'$, then it's easy to show that $\eta \leq \eta'$. For each $x_i \in \Gamma$,
\begin{align*}
\eta(x_i)(v) &= w_{x_i}(\val=v) \leq w'_{x_i}(\val=v) = \eta'(x_i)(v).
\end{align*}
Conversely, if $\eta \leq \eta'$, then for any $x_i$, $w_{x_i} \leq w'_{x_i}$. Since the equations defining $w_e$ in terms of the $w_{x_i}$ are polynomials with nonnegative coefficients, they are monotonic in the $w_{x_i}$. So for any $e$, $w_e \leq w'_e$.
\end{proof}

%% file: inconsistent_fgg.tex
\begin{figure*}
\begin{minipage}{\linewidth}  
\tikzset{tent/.append style={execute at begin node={\strut}}}
\tikzset{node distance=0.5cm}
\newcommand{\ExampleCase}{\begin{aligned}[t]
    &{\kw{if} \id{flip} \kw{then} {\kw{let} \unitterm = \id{gen}}} \\ &{{\kw{in} {\kw{let} \unitterm = \id{gen} \kw{in} \unitterm}} \kw{else} \unitterm}
    \end{aligned}}
\begin{align*}
\begin{tikzpicture}
\node[fac] (e) {$\id{gen}$};
\node[ext,right=of e] (vo) {$\val$} edge (e);
\end{tikzpicture}
&\longrightarrow
\begin{tikzpicture}
\node[fac] (e) {$\ExampleCase$};
\node[ext,right=of e] (vo) {$\val$} edge node[tent] {$\val$} (e);
\end{tikzpicture}
\\
\begin{tikzpicture}
\node[fac] (e) {$\ExampleCase$};
\node[ext,right=of e] (v) {$\val$} edge (e);
\end{tikzpicture}
&\longrightarrow
\begin{tikzpicture}
\node[fac](s) {$\id{flip}$};
\node[var,right=of s](vs) {} edge (s);
\node[fac,right=of vs,label={${}=\plusdenotation{\ctor{\id{True}}}{\unitdenotation}$}](fvs) {} edge node[tent] {$\val$} (vs);
\node[fac,right=1cm of fvs] (g) {$\begin{aligned}&{\kw{let} \unitterm = \id{gen} \kw{in}} \\ &{\kw{let} \unitterm = \id{gen} \kw{in} \unitterm}\end{aligned}$};
\node[ext,right=of g] {$\val$} edge node[tent] {$\val$} (g);
\end{tikzpicture} \\
\begin{tikzpicture}
\node[fac] (e) {$\ExampleCase$};
\node[ext,right=of e] (v) {$\val$} edge (e);
\end{tikzpicture}
&\longrightarrow
\begin{tikzpicture}
\node[fac](s) {$\id{flip}$};
\node[var,right=of s](vs) {} edge (s);
\node[fac,right=of vs,label={${}=\plusdenotation{\ctor{\id{False}}}{\unitdenotation}$}](fvs) {} edge node[tent] {$\val$} (vs);
\node[fac,right=1cm of fvs] (g) {$\unitterm$};
\node[ext,right=of g] {$\val$} edge node[tent] {$\val$} (g);
\end{tikzpicture}
\\
\begin{tikzpicture}
\node[fac](f) {$\begin{aligned}&{\kw{let} \unitterm = \id{gen} \kw{in}} \\ &{\kw{let} \unitterm = \id{gen} \kw{in} \unitterm}\end{aligned}$};
\node[ext,right=of f](v) {$\val$} edge (f);
\end{tikzpicture}
&\longrightarrow
\begin{tikzpicture}
\node[fac](e3) {$\id{gen}$};
\node[var,right=of e3](v3) {} edge node[tent]{$\val$} (e3);
\node[fac,right=of v3](e4) {$\kw{let} \unitterm = \id{gen} \kw{in} \unitterm$} ;
\node[ext,right=of e4](v) {$\val$} edge node[tent](vt){$\val$} (e4);
\end{tikzpicture}
\\
\begin{tikzpicture}
\node[fac](f) {$\kw{let} \unitterm = \id{gen} \kw{in} \unitterm$};
\node[ext,right=of f](v) {$\val$} edge (f);
\end{tikzpicture}
&\longrightarrow
\begin{tikzpicture}
\node[fac](e3) {$\id{gen}$};
\node[var,right=of e3](v3) {} edge node[tent]{$\val$} (e3);
\node[fac,right=of v3](e4) {$\unitterm$} ;
\node[ext,right=of e4](v) {$\val$} edge node[tent](vt){$\val$} (e4);
\end{tikzpicture}
\\
\begin{tikzpicture}
\node[fac] (e) {$\unitterm$};
\node[ext,right=of e] (v) {$\val$} edge (e);
\end{tikzpicture}
&\longrightarrow
\begin{tikzpicture}
\node[fac,label=above:{${}=\unitdenotation$}] (e) {};
\node[ext,right=of e] (v) {$\val$} edge (e);
\end{tikzpicture}
\\
\begin{tikzpicture}
  \node[fac] (f) {$\id{flip}$};
  \node[ext,right=of f] {$\val$} edge (f);
\end{tikzpicture}
&\longrightarrow
\begin{tikzpicture}
  \node[fac] (f) {$\kw{amb} \begin{aligned}[t] &({\kw{factor} p \kw{in} {\kw{true}}}) \\ &({\kw{factor} q \kw{in} {\kw{false}}}) \end{aligned}$};
  \node[ext,right=of f] {$\val$} edge (f);
\end{tikzpicture}
\\
\begin{tikzpicture}
  \node[fac] (f) {$\kw{amb} \begin{aligned}[t] &({\kw{factor} p \kw{in} {\kw{true}}}) \\ &({\kw{factor} q \kw{in} {\kw{false}}}) \end{aligned}$};
  \node[ext,right=of f] {$\val$} edge (f);
\end{tikzpicture}
&\longrightarrow
\begin{tikzpicture}
  \node[fac] (f) {${\kw{factor} p \kw{in} {\kw{true}}}$};
  \node[ext,right=of f] {$\val$} edge (f);
\end{tikzpicture}
\\
\begin{tikzpicture}
  \node[fac] (f) {$\kw{amb} \begin{aligned}[t] &({\kw{factor} p \kw{in} {\kw{true}}}) \\ &({\kw{factor} q \kw{in} {\kw{false}}}) \end{aligned}$};
  \node[ext,right=of f] {$\val$} edge (f);
\end{tikzpicture}
&\longrightarrow
\begin{tikzpicture}
  \node[fac] (f) {${\kw{factor} q \kw{in} {\kw{false}}}$};
  \node[ext,right=of f] {$\val$} edge (f);
\end{tikzpicture}
\\
\begin{tikzpicture}
  \node[fac] (f) {${\kw{factor} p \kw{in} {\kw{true}}}$};
  \node[ext,right=of f] {$\val$} edge (f);
\end{tikzpicture}
&\longrightarrow
\begin{tikzpicture}
\node[fac,label=above:$p$](e3) {};
\node[fac,right=of e3](e4) {$\kw{true}$} ;
\node[ext,right=of e4](v) {$\val$} edge node[tent](vt){$\val$} (e4);
\end{tikzpicture}
\\
\begin{tikzpicture}
\node[fac](e) {$\kw{true}$};
\node[ext,right=of e](v) {$\val$} edge (e);
\end{tikzpicture}
&\longrightarrow
\begin{tikzpicture}
\node[fac,label={above:${}=\plusdenotation{\ctor{\id{True}}}{\unitdenotation}$}](inj) {};
\node[ext,right=of inj](v) {$\val$} edge node[tent](v1t){$\val$} (inj);
\end{tikzpicture}
\\
\begin{tikzpicture}
  \node[fac] (f) {${\kw{factor} q \kw{in} {\kw{false}}}$};
  \node[ext,right=of f] {$\val$} edge (f);
\end{tikzpicture}
&\longrightarrow
\begin{tikzpicture}
\node[fac,label=above:$q$](e3) {};
\node[fac,right=of e3](e4) {$\kw{false}$} ;
\node[ext,right=of e4](v) {$\val$} edge node[tent](vt){$\val$} (e4);
\end{tikzpicture}
\\
\begin{tikzpicture}
\node[fac](e) {$\kw{false}$};
\node[ext,right=of e](v) {$\val$} edge (e);
\end{tikzpicture}
&\longrightarrow
\begin{tikzpicture}
\node[fac,label={above:${}=\plusdenotation{\ctor{\id{False}}}{\unitdenotation}$}](inj) {};
\node[ext,right=of inj](v) {$\val$} edge node[tent](v1t){$\val$} (inj);
\end{tikzpicture}
\end{align*}
\end{minipage}
\caption{FGG translated from the program in Example~\ref{ex:inconsistent}. The start symbol is $\id{gen}$.}
\label{fig:inconsistent_fgg}
\end{figure*}

%% file: epda.tex
\section{Details about Embedded Pushdown Automata}
\label{sec:epda_details}

\subsection{Definition and Normal Form}

\newcommand{\rest}{\mathinner{\cdot\cdot}}
\newcommand{\sep}{\bot}

\begin{definition}
  An embedded pushdown automaton is a tuple $P = (Q, \Sigma, \Gamma, q_0, S, \delta)$, where
  \begin{itemize}
  \item $Q$ is a finite set of states
  \item $\Sigma$ is a finite input alphabet
  \item $\Gamma$ is a finite stack alphabet
  \item $q_0 \in Q$ is the initial state
  \item $S \in \Gamma$ is the initial stack symbol
  \item $\delta$ is a set of transitions $q, A \xrightarrow{a} r, \beta$ where $q, r \in Q$, $A \in \Gamma$, $a \in \Sigma \cup \{\epsilon\}$, and $\beta \in (\Gamma^*\sep)^* (\Gamma^* \rest) (\Gamma^*\sep)^*$,
  where $\rest$ stands for the rest of a stack, and $\sep$ for the bottom of a stack.
\end{itemize}
  A configuration of $P$ is a tuple $(q, \alpha, w)$ where $\alpha \in (\Gamma^* \sep)^*$ and $w \in \Sigma^*$.
  If $(q, A\rest \xrightarrow{a} r, \beta) \in \delta$, then we define
  \begin{align*}
    (q, A \alpha \sep \gamma, aw) &\Rightarrow_P (r, \beta\{\rest:=\alpha\sep\} \gamma, w) & \sep \not\in \alpha \\
    (q, \sep \alpha, w) &\Rightarrow_P (q, \alpha, w)
  \end{align*}
  We say that $P$ accepts $w$ (by empty stack) if $(q_0, S\sep, w) \Rightarrow_P^* (q, \epsilon, \epsilon)$ for any $q$.
\end{definition}

\begin{proposition}
  Every EPDA is equivalent to an EPDA whose transitions all have one of the following forms:
  \begin{align*}
    q, X\rest &\xrightarrow{a} q', \rest & a \in \Sigma \cup \{\epsilon\} \\
    q, X\rest &\xrightarrow{\epsilon} q', YZ\rest \\
    q, X\rest &\xrightarrow{\epsilon} q', Y \sep Z\rest \\
    q, X\rest &\xrightarrow{\epsilon} q', Y\rest Z \sep 
  \end{align*}
\end{proposition}

\begin{proposition}
  Every EPDA is equivalent to an EPDA with only one state.
\end{proposition}
\begin{proof}
  Create a new initial stack symbol $S'$.
  For each $X \in \Gamma$, create stack symbols $X_{qrst}$ for all $q, r, s, t \in Q$.
  A stack symbol $X_{qrst}$ on top of the top stack simulates being in state $q$.
  Then
  \[\begin{array}{r@{}lr@{}l}
  \multicolumn{2}{l}{\text{for each}} & \multicolumn{2}{l}{\text{create}} \\
    q, X\rest{} &\xrightarrow{a} q', \rest & X_{qq'rr} \rest{} &\xrightarrow{a} \rest \\
    q, X\rest{} &\xrightarrow{a} q', YZ\rest & X_{qstr} \rest{} &\xrightarrow{a} Y_{q'uvr} Z_{ustv}\rest \\
    q, X\rest{} &\xrightarrow{a} q', Y \sep Z\rest & X_{qtus} \rest{} &\xrightarrow{a} Y_{q'rrr} \sep Z_{rtus} \rest \\
    q, X\rest{} &\xrightarrow{a} q', Y\rest Z\sep & X_{qtus} \rest{} &\xrightarrow{a} Y_{q'tur} \rest Z_{rsss} \sep
  \end{array}\]
  Finally, create transitions
  \begin{align*}
    S' \rest &\xrightarrow{\epsilon} {S}_{q_0rrr} \rest
  \end{align*}
  The transitions maintain the following invariants:
  \begin{itemize}
  \item If $X_{qrst}$ is immediately above $Y_{q'r's't'}$, then $r=q'$ and $s=t'$.
  \item If $X_{qrst}$ is the top symbol of a stack immediately above $Y_{q'r's't'}$, then $t=q'$.
  \item If $X_{qrst}$ is the bottom symbol of a stack, then $r=t$.
  \end{itemize}
  When the pop transition pops a symbol $X_{qq'rr}$, if it is not at the bottom, the invariants ensure that the new top symbol is of the form $Y_{q'r's't'}$, simulating entering state $q'$.
  If $X_{qq'rr}$ is at the bottom, the invariants ensure that $q'=r$ and therefore that the new top symbol is of the form $Y_{q'r's't'}$, again simulating entering state~$q'$.
\end{proof}

\begin{figure}
  \tikzset{node distance={0.3cm and 0.9cm},fac/.append style={text depth=0}}
  \begin{align*}
    \begin{tikzpicture}[trim left=-3em,trim right=3em]
      \node[ext](v) {$\val$};
      \node[fac,below=of v](z){\lstinline{f_StkCons z}} edge (v);
      \node[ext,below=of z] {\lstinline{zs}} edge (z);
    \end{tikzpicture}
    &\longrightarrow
    \begin{tikzpicture}[trim left=-3em,trim right=3em]
      \node[ext](v) {$\val$};
      \node[fac,below=of v,label=right:{concat}](concat){} edge (v);
      \node[ext,below=of concat] {\lstinline{zs}} edge (concat);
      \node[above=of v,anchor=base]{\lstinline{Pop}};
    \end{tikzpicture}
    \longalt
    \begin{tikzpicture}[node distance=0.3cm,trim left=-3em,trim right=3em]
      \node[ext](v) {$\val$};
      \node[fac,below=of v,label=right:{concat}](concat){} edge (v);
      \node[below left=of concat] {\lstinline{a}} edge (concat);
      \node[ext,below right=of concat] {\lstinline{zs}} edge (concat);
      \node[above=of v,anchor=base]{\lstinline{Scan a}};
    \end{tikzpicture}
    \longalt
    \begin{tikzpicture}[trim left=-4em,trim right=4em]
      \node[ext](v) {$\val$};
      \node[fac,below=of v](x){\lstinline{f_StkCons x}} edge node[tent]{$\val$} (v);
      \node[var,below=of x](int2) {} edge node[tent]{\lstinline{zs}} (x);
      \node[fac,below=of int2](z){\lstinline{f_StkCons y}} edge node[tent]{$\val$} (int2);
      \node[ext,below=of z] {\lstinline{zs}} edge node[tent]{\lstinline{zs}} (z);
      \node[above=of v,anchor=base]{\lstinline{Push x y}};
    \end{tikzpicture}
    \\
    \\[2ex]
    &\longalt
    \begin{tikzpicture}
      \node[ext](v) {$\val$};
      \node[fac,below=of v,label=right:{concat}](concat) {} edge (v);
      \node[var,below left=of concat](int1) {} edge (concat);
      \node[fac,below=of int1](x){\lstinline{f_StkCons x}} edge node[tent]{$\val$} (int1);
      \node[var,below=of x](int2) {} edge node[tent]{\lstinline{zs}} (x);
      \node[fac,below=of int2](nil){\lstinline{f_StkNil}} edge node[tent]{$\val$} (int2);
      \node[var,below right=of concat](int3) {} edge (concat);
      \node[fac,below=of int3](z){\lstinline{f_StkCons y}} edge node[tent]{$\val$} (int3);
      \node[ext,below=of z] {\lstinline{zs}} edge node[tent]{\lstinline{zs}} (z);
      \node[above=of v,anchor=base]{\lstinline{PushAbove x y}};
    \end{tikzpicture}
    \longalt
    \begin{tikzpicture}
      \node[ext](v) {$\val$};
      \node[fac,below=of v,label=right:{concat}](concat) {} edge (v);
      \node[var,below left=of concat](int1) {} edge (concat);
      \node[fac,below=of int1](x){\lstinline{f_StkCons x}} edge node[tent]{$\val$} (int1);
      \node[ext,below=of x] {\lstinline{zs}} edge node[tent]{\lstinline{zs}} (x);
      \node[var,below right=of concat](int3) {} edge (concat);
      \node[fac,below=of int3](z){\lstinline{f_StkCons y}} edge node[tent]{$\val$} (int3);
      \node[var,below=of z](int2) {} edge node[tent]{\lstinline{zs}} (z);
      \node[fac,below=of int2](nil){\lstinline{f_StkNil}} edge node[tent]{$\val$} (int2);
      \node[above=of v,anchor=base]{\lstinline{PushBelow x y}};
    \end{tikzpicture}
    \\
    \\[2ex]
    \begin{tikzpicture}[trim left=-3em,trim right=3em]
      \node[ext](v) {$\val$};
      \node[fac,below=of v](z){\lstinline{f_StkNil}} edge (v);
    \end{tikzpicture}
    &\longrightarrow
    \begin{tikzpicture}[trim left=-3em,trim right=3em]
      \node[ext](v) {$\val$};
      \node[fac,below=of v,label=right:{concat}](concat) {} edge (v);
      \node[below=of concat] {$\epsilon$} edge (concat);
    \end{tikzpicture}
    \\
    \\[2ex]
    \begin{tikzpicture}[trim left=-3em,trim right=3em]
      \node[ext](v) {$\val$};
      \node[fac,below=of v](z){$S$} edge (v);
    \end{tikzpicture}
    &\longrightarrow
    \begin{tikzpicture}[trim left=-3em,trim right=3em]
      \node[ext](v) {$\val$};
      \node[fac,below=of v](x){\lstinline{f_StkCons Z}} edge node[tent]{$\val$} (v);
      \node[var,below=of x](int2) {} edge node[tent]{\lstinline{zs}} (x);
      \node[fac,below=of int2](z){\lstinline{f_StkNil}} edge node[tent]{$\val$} (int2);
    \end{tikzpicture}
  \end{align*}
  \caption{Sketch of FGG compiled from final program of \cref{sec:epda}.}
  \label{fig:epda_fgg}
\end{figure}
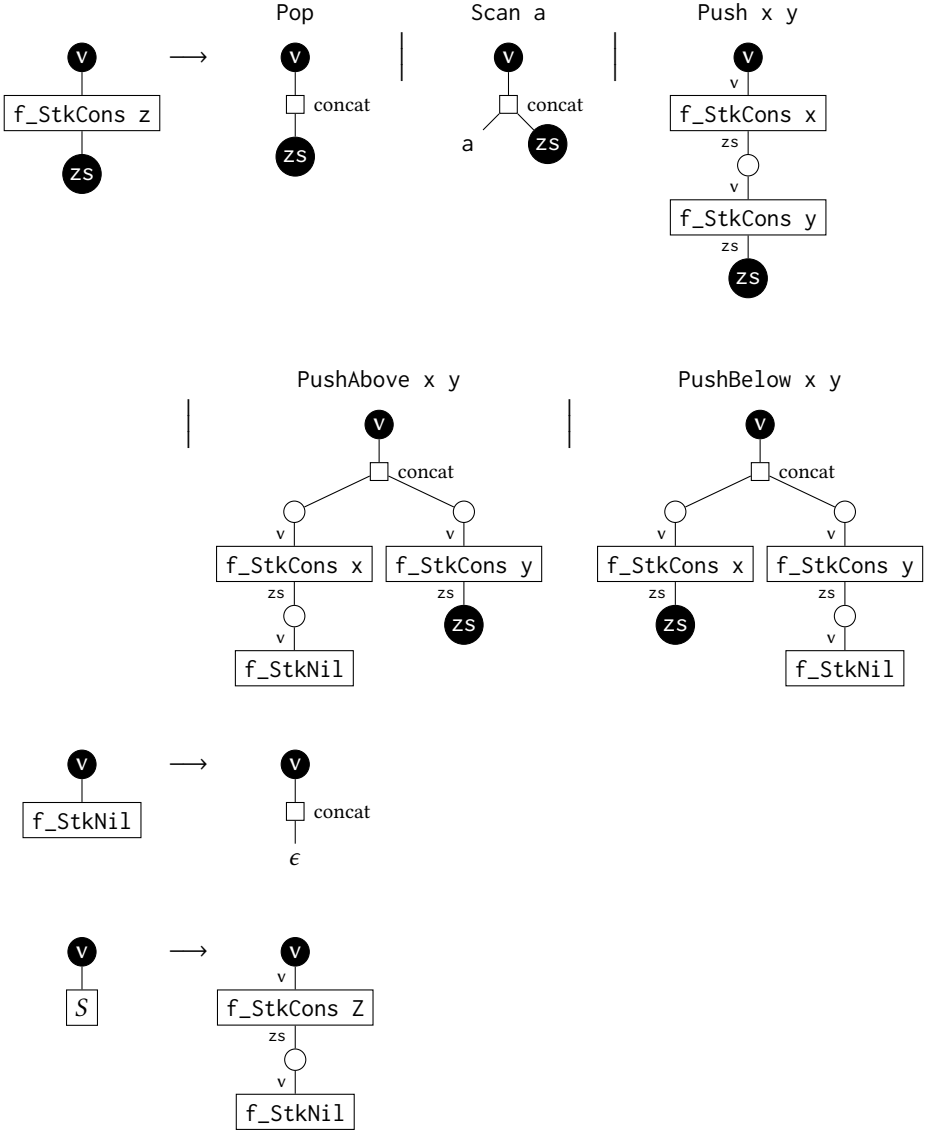

\subsection{Relationship to TAG}

Here we explain in more detail how the transformed EPDA program in \cref{sec:epda} is related to a tree-adjoining grammar (TAG)\@, and how inference on it amounts to TAG parsing.
When the program is converted to an FGG (\cref{sec:fggs}), and smaller rules are fused into larger rules, the result looks like \cref{fig:epda_fgg}.
Each possible \lstinline{Action} that \lstinline{transition z} might take is converted to an FGG rule with \lstinline{f_StkCons z} on the left.
The external nodes \extinline{$\val$} and \extinline{\lstinline{zs}} have type \lstinline{StackFolded},
which can be thought of as a pair of string positions, namely the starting position (usually called \lstinline{ws}) and ending position (usually called \lstinline{zs*}) of the substring that is scanned while the \lstinline{Stack} is consumed.
For clarity, we've taken the liberty of drawing some strings directly into the FGG and writing factors labeled ``concat'' in place of more complicated logic that concatenates one or two substrings.

When arranged in this way, this FGG looks just like a TAG\@.
A~TAG could be defined as an HRG~(\cref{def:hrg}) where all rule right-hand sides are ordered trees, which come in two kinds: \emph{initial} trees have one external node, which is the root node, and \emph{auxiliary} trees have exactly two external nodes, which are the root node and a node called the \emph{foot} node.
(TAGs with no auxiliary trees are called tree-substitution grammars and generate the same string languages as CFGs.)
In~\cref{fig:epda_fgg}, the \lstinline{f_StkCons} rules are the auxiliary trees, because \lstinline{f_StkCons z} has type \lstinline{StackFolded -> StackFolded}: the returned \extinline{$\val$} is the root node, and the argument \extinline{\lstinline{zs}} is the foot node.
On the other hand, the \lstinline{f_StkNil} and $S$ rules are the initial trees, because both \lstinline{f_StkNil} and \lstinline{f_StkCons Z f_StkNil} have type \lstinline{StackFolded}.

Any TAG can be converted to a normal form \citep{lang:1994} analogous to Chomsky normal form, in which every rule has at most two nonterminal symbols. In the two-nonterminal case, the two nonterminals can stand in three possible relationships: both dominating the foot node, one dominating the foot node and the other to its left, and one dominating the foot node and the other to its right. 
These possibilities correspond exactly to our FGG rules for \lstinline{Push}, \lstinline{PushAbove}, and \lstinline{PushBelow}.
There is also a rule that introduces a terminal symbol, which corresponds to our FGG rule for \lstinline{Scan}.

%% file: pcfg-gen.tex
\section{Benchmark Programs}
\label{s:pcfg-gen}

The WebPPL benchmark results in \cref{fig:webppl} were generated using
programs like this:
\begin{verbatim}
var S = function() {
  if (flip(0.9))
    return ['A'];
  else
    return S().concat(S());
}
var model = function() {
  return S().length;
}
display('50000000 rejection');
Infer({ model, method: 'rejection', samples: 50000000 }).getDist();
\end{verbatim}

The Dice benchmark results in \cref{fig:dice} were generated using
programs like this:
\begin{verbatim}
fun A(i:int(3)) { if i < int(3,6) then i + int(3,1) else int(3,7) }
fun S0(i:int(3)) { int(3,7) }
fun S1(i:int(3)) { if flip 0.1 then S0(S0(i)) else A(i) }
fun S2(i:int(3)) { if flip 0.1 then S1(S1(i)) else A(i) }
fun S3(i:int(3)) { if flip 0.1 then S2(S2(i)) else A(i) }
fun S4(i:int(3)) { if flip 0.1 then S3(S3(i)) else A(i) }
fun S5(i:int(3)) { if flip 0.1 then S4(S4(i)) else A(i) }
fun S6(i:int(3)) { if flip 0.1 then S5(S5(i)) else A(i) }
S6(int(3,0)) == int(3,6)
\end{verbatim}